\documentclass[3p,review]{elsarticle}

\usepackage{lineno,hyperref,graphicx,epstopdf}
\usepackage{times}
\usepackage{type1cm}
\fontsize{12pt}{18pt}

\usepackage{lineno,hyperref}
\usepackage{amsmath,amsfonts}
\usepackage{algorithmic}
\usepackage{array}
\usepackage[caption=false,font=normalsize,labelfont=sf,textfont=sf]{subfig}
\usepackage{textcomp}
\usepackage{stfloats}
\usepackage{url}
\usepackage{verbatim}
\usepackage{graphicx}
\usepackage{epstopdf}
\usepackage{dcolumn}
\usepackage{bm}
\usepackage{braket}
\usepackage{color}
\usepackage{url}
\usepackage{algorithmic}
\usepackage{algorithm}
\usepackage{balance}
\usepackage{booktabs}
\allowdisplaybreaks[3]
\modulolinenumbers[5]

\usepackage{braket}
\usepackage{amsmath}
\usepackage{amsfonts}
\usepackage{amssymb}
\usepackage{amsthm}
\usepackage[thicklines]{cancel}

\theoremstyle{plain}
\newtheorem{theorem}{Theorem}

\newtheorem{lemma}{Lemma}

\theoremstyle{definition}
\newtheorem{definition}{Definition}

\makeatletter
\newenvironment{breakablealgorithm}
  {
   \begin{center}
     \refstepcounter{algorithm}
     \hrule height.8pt depth0pt \kern2pt
     \renewcommand{\caption}[2][\relax]{
       {\raggedright\textbf{\ALG@name~\thealgorithm} ##2\par}%
       \ifx\relax##1\relax 
         \addcontentsline{loa}{algorithm}{\protect\numberline{\thealgorithm}##2}%
       \else 
         \addcontentsline{loa}{algorithm}{\protect\numberline{\thealgorithm}##1}%
       \fi
       \kern2pt\hrule\kern2pt
     }
  }{
     \kern2pt\hrule\relax
   \end{center}
  }
\makeatother

\journal{}

\date{}

\begin{document}
\begin{frontmatter}
\title{Distributed exact quantum algorithms for  Bernstein-Vazirani and search problems}

\author[mymainaddress,mysecondaryaddress]{Xu Zhou}
\ead{zhoux229@mail2.sysu.edu.cn}

\author[mymainaddress,mysecondaryaddress,myfourthaddress]{Daowen Qiu\corref{mycorrespondingauthor}}

\cortext[mycorrespondingauthor]{Corresponding author}
\ead{issqdw@mail.sysu.edu.cn}

\author[mythirdaddress,myfourthaddress]{Le Luo}
\ead{luole5@mail.sysu.edu.cn}

\address[mymainaddress]{Institute of Quantum Computing and Computer Theory, School of Computer Science and Engineering, Sun Yat-sen University, Guangzhou 510006, China}
\address[mysecondaryaddress]{The Guangdong Key Laboratory of Information Security Technology, Sun Yat-sen University, Guangzhou 510006, China}
\address[mythirdaddress]{School of Physics and Astronomy, Sun Yat-sen University, Zhuhai 519082, China}
\address[myfourthaddress]{QUDOOR Technologies Inc., Zhuhai, 519099, China}

\begin{abstract}
Distributed quantum computation has gained extensive attention since small-qubit quantum computers seem to be built more practically  in the noisy intermediate-scale quantum (NISQ) era. In this paper, we give a distributed Bernstein-Vazirani algorithm (DBVA) with $t$ computing nodes, and a distributed exact Grover's algorithm (DEGA) that solve the search problem with only one target item in the unordered databases. 
Though the designing techniques are simple in the light of BV algorithm, Grover's algorithm, the improved Grover's algorithm by Long, and the distributed Grover's algorithm by Qiu et al, in comparison to BV algorithm, the circuit depth of DBVA is not greater than $2^{\text{max}(n_0, n_1, \cdots, n_{t-1})}+3$ instead of $2^{n}+3$, and the circuit depth of DEGA is $8(n~\text{mod}~2)+9$, which is less than the circuit depth of  Grover's algorithm, 
$1 + 8\left\lfloor \frac{\pi}{4}\sqrt{2^n} \right\rfloor$. 
In particular, we provide situations of our DBVA and DEGA on MindQuantum (a quantum software) to validate the correctness and practicability of our methods. By simulating the algorithms running in the depolarized channel, it further illustrates that distributed quantum algorithm has superiority of resisting noise.
\end{abstract}

\begin{keyword}
Distributed quantum computation; Noisy intermediate-scale quantum (NISQ) era; Distributed Bernstein-Vazirani algorithm (DBVA); Distributed exact Grover's algorithm (DEGA); MindQuantum
\end{keyword}

\end{frontmatter}


\section{Introduction}\label{introduction}
Quantum computation, a novel computing model, utilizes quantum information units to perform computation depending on the laws of quantum mechanics. It has attracted considerable attention in the past couple of decades since Benioff \cite{RefBenioff1980, RefBenioff1982} proposed the concept of the quantum Turing machine in 1980. Subsequently, a great number of remarkable and valuable quantum algorithms have been presented, such as Deutsch's algorithm \cite{RefDeutsch1985}, Deutsch-Jozsa (DJ) algorithm \cite{RefDeutsch1992}, Bernstein-Vazirani (BV) algorithm \cite{RefBernstein1993}, Simon's algorithm \cite{RefSimon1997}, Shor's algorithm \cite{RefShor1994}, Grover's algorithm \cite{RefGrover1997}, HHL algorithm \cite{RefHarrow2009} and so on. For some specific problems, quantum algorithms have the potential to offer essential advantages over classical algorithms.

Bernstein-Vazirani (BV) algorithm \cite{RefBernstein1993} was proposed after DJ algorithm. It proves that quantum algorithms are capable of not only determining properties of Boolean functions \cite{RefBernstein2015, RefXie2018, RefYounes2015}, but also determining the functions themselves. For  BV problem, it aims to find the hidden string $s\in\{0,1\}^n$ satisfying the function $f_s (x)=\langle s\cdot x\rangle= \left(\sum_{i=0}^{n-1}s_i\cdot x_i\right)~\text{mod}~2: \{0,1\}^{n} \rightarrow \{0,1\}$. Suppose we have a black box (Oracle) that can return $f_s (x_0)$ for an input $x_0$. In order to solve the BV problem, the classical algorithm will have to query Oracle $n$ times, while the BV algorithm will just need to query once. In 2017, Nagata et al. \cite{RefNagata2017} presented a generalized Bernstein-Vazirani algorithm.

Grover's algorithm \cite{RefGrover1997} was proposed by Grover in 1996, which is a quantum search algorithm that solves the search problem in an unordered databases with $2^n$ elements. Specifically, let Boolean function $f: \{0,1\}^n \rightarrow \{0,1\}$. For the $m$ targets search problem, it aims to find out one target string $x_i\in\{0,1\}^n$ satisfying $f (x_i)=1$, where $i\in\{0, 1, \cdots, m-1\}$. Suppose we have a black box (Oracle) that can recognize the inputs, which means it will return $f (x)=1$ when $x$ is a target string, and output $f (x)=0$ otherwise. The classical algorithms have to query Oracle $\mathcal{O} (2^n/m)$ times, while  Grover's algorithm  just need to query $\mathcal{O} (\sqrt{2^n/m})$ times to figure out the target string with high probability, which has square acceleration compared with the classical algorithms.

Grover's algorithm has been widely used in the minimum search problem \cite{RefChristoph1996}, string matching problem \cite{RefRamesh2003}, quantum dynamic programming \cite{RefAmbainis2019} and computational geometry problem \cite{RefAndris2020}, etc. Subsequently, a generalization of  Grover's algorithm known as the quantum amplitude amplification and estimation algorithms \cite{RefBrassard2002} were presented.

However, Grover's algorithm is unable to search for the target state accurately, except in the case of looking for one data in an unordered database with a size of four \cite{Refiao2010}. In 2001, Long \cite{RefLong2001} improved Grover's algorithm, which acquires the goal state exactly. 

Currently, quantum computation is entering the noisy intermediate-scale quantum (NISQ) era \cite{RefPreskill2018}. During this period, a great deal of theoretical and experimental work has been devoted to the development of quantum devices. Studies on small-scale quantum computers based on ions \cite{RefCirac1995}, photons \cite{RefLu2007}, superconduction \cite{RefMakhlin2001}, and other quantum systems \cite{RefBerezovsky2008, RefHanson2008} have been extensively conducted. It may contain tens to thousands of qubits with a gate error rate of less than $10^{-3}$ \cite{RefEndo2021}. Currently, small-qubit quantum computers seem to be built  more practically, compared to  large-scale universal quantum computers. Hence, researchers are considering how several small-scale devices might collaborate to accomplish a task on a large scale.

Distributed quantum computation, a research field in computer science, combines distributed computing and quantum computation. It aims to divide a large problem into many small parts and then distribute them across multiple quantum computers for processing. In this way, each component needs fewer qubits and has a shallower quantum circuit.

There has been great interest in distributed quantum algorithms \cite{RefBuhrman2003, RefYimsiriwattana2004, RefBeals2013, RefLi2017}. Avron et al. \cite{RefAvron2021} proposed a distributed method for constructing Oracles. 
Qiu et al. \cite{RefQiu2022} proposed two distributed Grover's algorithms (parallel and serial) and these algorithms require fewer qubits and have shallower depth of circuits. Also, the distributed algorithms by Qiu et al. \cite{RefQiu2022} can solve the search problem with multiple targets. They  \cite{RefQiu2022} divided a Boolean function $f$ into $2^k$ subfunctions, where $2^k$ represents the number of computing nodes. Theyr determine the number of solutions in each subfunction by  using the quantum counting algorithm, and then ran $(n-k)$-qubit Grover's algorithm for each subfunction. 
Particularly, Qiu et al. \cite{RefQiu2022} proposed an efficient algorithm of constructing quantum circuits for realizing the Oracle corresponding to any Boolean function with conjunctive normal form (CNF).

Tan, Xiao, and Qiu et al. \cite{RefTan2022} presented an interesting and novel distributed Simon's algorithm. Recently, Xiao and Qiu et al. \cite{RefXiao2022} proposed a distributed Shor's algorithm, which can separately estimate patrial bits of $s/r$ for some $s\in\{0,1,\cdots,r-1\}$ by two quantum computers during solving order-finding.

In this paper, we  give a distributed Bernstein-Vazirani algorithm (DBVA) with $t$ computing nodes. After that, we present a distributed exact Grover's algorithm (DEGA) for the case of only one target item in the unordered databases. 
DBVA and DEGA are distributed exact quantum algorithm. 
In comparison to the BV algorithm, the circuit depth of DBVA is not greater than $2^{\text{max}(n_0, n_1, \cdots, n_{t-1})}+3$ instead of $2^{n}+3$.  The circuit depth of DEGA is $8(n~\text{mod}~2)+9$, which is less than the circuit depths of the original Grover's algorithm, $1 + 8\left\lfloor \frac{\pi}{4}\sqrt{2^n} \right\rfloor$. 

MindQuantum (a quantum software) \cite{Refmq_2021} is a general software library supporting the development of applications for quantum computation. In this paper, we provide particular situations of our DBVA and DEGA on MindQuantum to validate the correctness and practicability of our methods. At the first, we explicate the concrete circuits of a 6-qubit DBVA on MindQuantum as an example, which shows how to decompose the original 6-qubit BV problem into two 3-qubit and three 2-qubit tasks. In the next part, we explicate the specific steps of implement $n$-qubit DEGA, where $n\in\{2,3,4,5\}$. In the end, by simulating the 6-qubit DBVA and 5-qubit DEGA running in the depolarized channel, it further illustrates that distributed quantum algorithm has  superiority of resisting noise.

The rest of the paper is organized as follows. We will review BV algorithm and Grover's algorithm in Section \ref{preliminaries} and mainly describe our DBVA and DEGA in Section \ref{DBVA} and Section \ref{DEGA}. Next, the analysis of our algorithms will be presented in Section \ref{analysis}. After that, we provide particular situations of our DBVA and DEGA on MindQuantum in Section \ref{experiment}. Finally, a brief conclusion is given in Section \ref{conclusion}.

\section{Preliminaries}\label{preliminaries}
In this section, we will review  BV algorithm \cite{RefBernstein1993}, Grover's algorithm \cite{RefGrover1997}, and the algorithm by Long \cite{RefLong2001}.

\subsection{Bernstein-Vazirani problem \cite{RefBernstein1993}}
Let Boolean function $f: \{0,1\}^n \rightarrow \{0,1\}$. Consider a function
\begin{eqnarray}\label{bvfdef}
f_s (x)=\langle s\cdot x\rangle = \left(\sum_{i=0}^{n-1}s_i\cdot x_i\right)~\text{mod}~2,
\end{eqnarray}
where $x\in\{0,1\}^n$. Suppose we have an Oracle ($O_{f_s(x)}:\vert x\rangle\rightarrow (-1)^{f_s (x)}\vert x\rangle$) that can return $f_s (x_0)$ for an input $x_0\in\{0,1\}^n$. For BV problem, it aims to identify the hidden string $s\in\{0,1\}^n$.

For BV algorithm, it will just need to query Oracle once (see Figure \ref{figure1}).
\begin{figure}[H]
\centering
\includegraphics[width=3.5in]{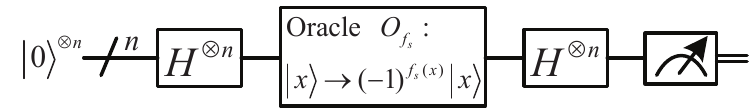}
\caption{\label{figure1} Quantum circuit of the BV algorithm.}
\end{figure}

The evolution of the quantum state in  BV algorithm is shown below:
\begin{eqnarray}\label{state_BV}
  \vert 0\rangle^{\otimes n} &\stackrel{H^{\otimes n}}{\longrightarrow}& \frac{1}{\sqrt{2^n}}\sum_{x\in\{0,1\}^n}\vert x\rangle \nonumber \\
    &\stackrel{O_{f_s(x)}}{\longrightarrow}& \frac{1}{\sqrt{2^n}}\sum_{x\in\{0,1\}^n} (-1)^{f_s (x)}\vert x\rangle = \frac{1}{\sqrt{2^n}}\sum_{x\in\{0,1\}^n} (-1)^{\langle s\cdot x\rangle}\vert x\rangle \nonumber \\
    &=& \frac{1}{\sqrt{2^n}}\sum_{x_0x_1\cdots x_{n-1}\in\{0,1\}^n} (-1)^{\left(\sum_{i=0}^{n-1} s_i \cdot x_i\right)~\text{mod}~2}\vert x_0x_1\cdots x_{n-1}\rangle \nonumber \\
    &\stackrel{H^{\otimes n}}{\longrightarrow}& \vert s_0s_1s_2\cdots s_{n-1}\rangle = \vert s\rangle.
\end{eqnarray}

\subsection{Grover's algorithm \cite{RefGrover1997}}
Suppose that the search task is carried out in an unordered set (or database) with $2^n$ elements. We establish a one-to-one correspondence between the elements in the database and the indexes (integers from 0 to $2^n-1$). We focus on searching the indexes of these elements.

Let Boolean function $f: \{0,1\}^n \rightarrow \{0,1\}$. Define a function for the search problem as follows:
\begin{eqnarray}\label{groverdef}
f(x)=
\begin{cases}
1,x=\tau,\\
0,x\neq \tau ,
\end{cases}
\end{eqnarray}
where $x\in\{0,1\}^n$ and $\tau$ is the target index. In a quantum system, indexes are represented by quantum states $\vert 0\rangle, \vert 1\rangle, \cdots, \vert 2^n-1\rangle$ (or $\vert 00\cdots0\rangle, \vert 00\cdots1\rangle, \cdots, \vert 11\cdots1\rangle$). 

Without losing of generality, we consider a search problem that includes only one target item in the unordered databases throughout this paper. Suppose we have a black box (Oracle  $U_{f(x)}:\vert x\rangle\rightarrow (-1)^{f (x)}\vert x\rangle$, where $x\in\{0,1\}^n$) that can recognize the inputs. The purpose of Grover's algorithm is to find out $\vert \tau\rangle$ with high probability.

Here, we briefly review  Grover's algorithm in Algorithm \ref{Grover's Algorithm}. The whole quantum circuit is shown in Figure \ref{grovercir}.

\begin{breakablealgorithm}
\caption{Grover's Algorithm (includes only one target).}
\label{Grover's Algorithm}
\begin{algorithmic}
\STATE \textbf{Input}: (1) The number of qubits $n$; (2) A function $f(x)$, where $f(x)=0$ for all $x\in\{0,1\}^n$ except $\tau$, for which $f(\tau)=1$.
\STATE \textbf{Initialization}: An Oracle $U_{f(x)}:\vert x\rangle\rightarrow (-1)^{f (x)}\vert x\rangle$, where $x\in\{0,1\}^n$.
\STATE \textbf{Output}: The string $\tau\in\{0,1\}^n$ such that $f(\tau)=1$ with great probability.
\STATE \textbf{Procedure:}
\STATE \textbf{Step 1.} Initialize $n$ qubits $\vert 0\rangle^{\otimes n}$.
\STATE \textbf{Step 2.} Apply $n$ Hadamard gates $H^{\otimes n}$.
\STATE \textbf{Step 3.} Grover operator $G = -H^{\otimes n} U_0 H^{\otimes n}U_{f(x)}$ is executed with $\left\lfloor \frac{\pi}{4}\sqrt{2^n} \right\rfloor$ times, where $U_0=I_n-2(\vert 0\rangle\langle0\vert )^{\otimes n}$ and $U_{f(x)}=I_n - 2\vert \tau\rangle\langle\tau\vert $.
\STATE \textbf{Step 4.} Measure each qubit in the basis $\{\vert 0\rangle, \vert 1\rangle\}$.
\end{algorithmic}
\end{breakablealgorithm}

\begin{figure}[H]
\centering
\includegraphics[width=3.3in]{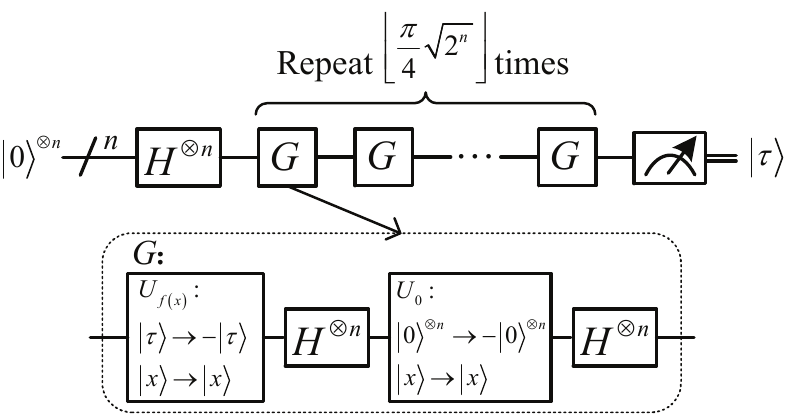}
\caption{\label{grovercir} Quantum circuit of Grover's algorithm.}
\end{figure}

\subsection{The algorithm by Long \cite{RefLong2001}}
In 2001, Long improved Grover's algorithm, which will acquire the goal state with a probability of $100 \%$. Its main idea is to extend Grover operator $G$ to another operator $L$. 

Here, we briefly review this algorithm in Algorithm \ref{Long's Algorithm}. The whole quantum circuit is shown in Figure \ref{longcir}.

\begin{breakablealgorithm}
\caption{The algorithm by Long (includes only one target).}
\label{Long's Algorithm}
\begin{algorithmic}
\STATE \textbf{Input}: (1) The number of qubits $n$; (2) A function $f(x)$, where $f(x)=0$ for all $x\in\{0,1\}^n$ except $\tau$, for which $f(\tau)=1$.
\STATE \textbf{Initialization}: An Oracle $R_{f(x)}:\vert x\rangle\rightarrow e^{i\phi \cdot f (x)}\vert x\rangle$, where $x\in\{0,1\}^n$, $\phi=2\arcsin\left(\sin\left(\frac{\pi}{4J+6}\right) / \sin \theta\right)$, $J=\lfloor(\pi/2-\theta)/(2\theta)\rfloor$, and $\theta= \arcsin {\sqrt{1/2^n}}$.
\STATE \textbf{Output}: The string $\tau\in\{0,1\}^n$ such that $f(\tau)=1$ with a probability of $100 \%$.
\STATE \textbf{Procedure:}
\STATE \textbf{Step 1.} Initialize $n$ qubits $\vert 0\rangle^{\otimes n}$.
\STATE \textbf{Step 2.} Apply $n$ Hadamard gates $H^{\otimes n}$.
\STATE \textbf{Step 3.} Operator $L = -H^{\otimes n} R_0 H^{\otimes n}R_{f(x)}$ is executed with $J+1$ times, where $R_0 = I_n + (e^{i\phi}-1)(\vert 0\rangle\langle0\vert )^{\otimes n}$ and $R_{f(x)} = I_n + (e^{i\phi}-1)\vert \tau\rangle\langle\tau\vert $.
\STATE \textbf{Step 4.} Measure each qubit in the basis $\{\vert 0\rangle, \vert 1\rangle\}$.
\end{algorithmic}
\end{breakablealgorithm}

\begin{figure}[H]
\centering
\includegraphics[width=3.3in]{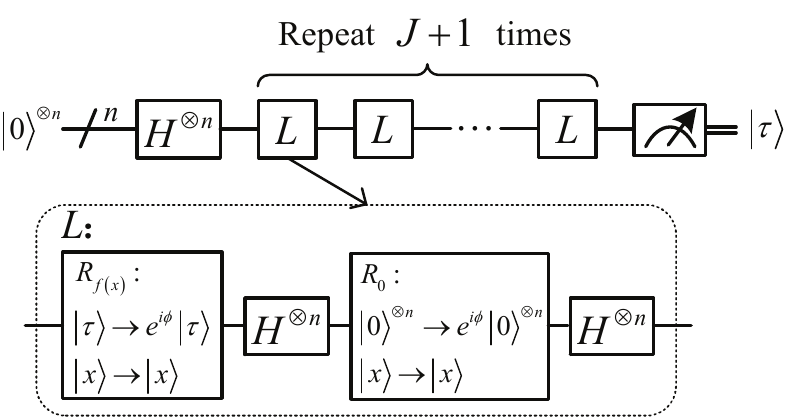}
\caption{\label{longcir} Quantum circuit of the algorithm by Long.}
\end{figure}

\section{Distributed Bernstein-Vazirani algorithm}\label{DBVA}
In this section, we mainly describe our distributed Bernstein-Vazirani algorithm (DBVA) with $t$ computing nodes.

\subsection{Subfunctions}\label{subfunctions}
As mentioned in \cite{RefAvron2021}, an $n$-qubit Boolean function $f: \{0,1\}^n \rightarrow \{0,1\}$ can be split into two subfunctions $f_{\text{even}/\text{odd}}: \{0,1\}^{n-1} \rightarrow \{0,1\}$,
\begin{eqnarray}
f_{\text{even}}(x)&=&f(x_0x_1\cdots x_{i-1}0x_{i+1}\cdots x_{n-2}x_{n-1}),\label{subfunctioneven} \\
f_{\text{odd}}(x)&=&f(x_0x_1\cdots x_{i-1}1x_{i+1}\cdots x_{n-2}x_{n-1}),\label{subfunctionodd}
\end{eqnarray}
where $x=x_0x_1\cdots x_{i-1}x_{i+1}\cdots x_{n-2}x_{n-1}\in\{0,1\}^{n-1}$.

Furthermore, if we choose last $k\in\{1,2,\cdots,n-1\}$ bits to divide $f$, we can divide the original function into $2^k$ subfunctions $f_p: \{0,1\}^{n-k} \rightarrow \{0,1\}$:
\begin{eqnarray}\label{subfunctions_fp}
f_p(x)&=&f(\underbrace{x_0x_1\cdots x_{n-k-1}}_{n-k} \underbrace{y_{p_{0}} y_{p_{1}}\cdots y_{p_{k-1}}}_{k}),
\end{eqnarray}
where $k\in\{1,2,\cdots,n-1\}$, $x=x_0x_1\cdots x_{n-k-1}\in\{0,1\}^{n-k}$, and $y_{p_{0}} y_{p_{1}}\cdots y_{p_{k-1}}\in\{0,1\}^{k}$ is the binary representation of $p\in\{0,1,\cdots,2^k-1\}$.

More generally, $y_{p_{0}} y_{p_{1}}\cdots y_{p_{k-1}}$ can be at any $k$ positions in $x$ of the original function $f$. For instance, we can also obtain another $2^k$ subfunctions $f_q: \{0,1\}^{n-k} \rightarrow \{0,1\}$:
\begin{eqnarray}\label{subfunctions_fq}
f_q(x)&=&f(\underbrace{y_{q_{0}} y_{q_{1}}\cdots y_{q_{k-1}}}_{k}   \underbrace{x_{k} x_{k+1} \cdots x_{n-1}}_{n-k}),
\end{eqnarray}
where $k\in\{1,2,\cdots,n-1\}$, $x=x_{k} x_{k+1} \cdots x_{n-1}\in\{0,1\}^{n-k}$, and $y_{q_{0}} y_{q_{1}}\cdots y_{q_{k-1}}\in\{0,1\}^{k}$ is the binary representation of $q\in\{0,1,\cdots,2^k-1\}$.

Next, we must declare our partition strategy for the original function $f_s (x)$ before providing our DBVA.

Suppose that the hidden string $s\in\{0,1\}^n$ of the original problem we are looking for can be expressed as
\begin{eqnarray}\label{tnodesstrings}
s=s_{n_0}s_{n_1}\cdots s_{n_{t-1}}\in\{0,1\}^{n},
\end{eqnarray}
where $s_{n_j}\in\{0,1\}^{n_j}$, $\sum_{j=0}^{t-1} n_j =n$, and $j \in \{0,1,\cdots,t-1 \}$.

Let the number of qubits of the $j$-th computing node be $n_{j}$, where $j \in \{0,1,\cdots,t-1 \}$.

For the zeroth computing node ($j=0$), we choose last $n_1+n_2+\cdots+n_{t-1}$ bits in $x$ to divide $f$, then we can get $2^{n_1+n_2+\cdots+n_{t-1}}$ subfunctions. We keep the subfunction when $p=0$, and denote it as
\begin{eqnarray}\label{tnodesfn00}
f_{{s_{n_0}}}(m_0) = f(m_0\underbrace{00\cdots00}_{n_1+n_2+\cdots+n_{t-1}}),
\end{eqnarray}
where $m_0\in\{0,1\}^{n_0}$.

For the $j$-th computing node ($j\in\{1,2,\cdots, t-2\}$), we choose first $n_0+n_1+\cdots+n_{j-2}$ and last $n_j+n_{j+1}+\cdots+n_{t-1}$ bits in $x$ to divide $f$, then we can get $2^{n_0+n_1+\cdots+n_{j-2}+n_j+n_{j+1}+\cdots+n_{t-1}}$ subfunctions.  We keep the subfunction when $p=0$, and denote it as
\begin{eqnarray}\label{tnodesf0n10}
f_{{s_{n_{j-1}}}}(m_{j-1})=f(\underbrace{00\cdots00}_{n_0+n_1+\cdots+n_{j-2}} m_{j-1} \underbrace{00\cdots00}_{n_j+n_{j+1}+\cdots+n_{t-1}}),
\end{eqnarray}
where $m_{j-1}\in\{0,1\}^{n_{j-1}}$.

For the last computing node ($j=t-1$), we choose first $n_0 + n_1+\cdots+n_{t-2}$ bits in $x$ to divide $f$, then we can get $2^{n_0 + n_1+\cdots+n_{t-2}}$ subfunctions. We keep the subfunction when $p=0$, and denote it as
\begin{eqnarray}\label{tnodesf0nt-1}
f_{{s_{n_{t-1}}}}(m_{t-1})=f(\underbrace{00\cdots00}_{n_0+n_1+\cdots+n_{t-2}}m_{t-1}),
\end{eqnarray}
where $m_{t-1}\in\{0,1\}^{n_{t-1}}$.

\subsection{DBVA}
Assume that it is easy to obtain the Oracle for each subfunction, which means we can obtain $t$ Oracles $O_{f_{s_{n_j}}(m_j)}$ corresponding to $t$ subfunctions $f_{{s_{n_j}}}(m_j)=\langle s_{n_j}\cdot m_j\rangle: \{0,1\}^{n_j} \rightarrow \{0,1\}$, where
\begin{eqnarray}\label{oraclefsmj}
O_{f_{s_{n_j}}(m_j)}:\vert m_j\rangle\rightarrow (-1)^{f_{{s_{n_j}}}(m_j)} \vert m_j\rangle,
\end{eqnarray}
$s_{n_j}\in\{0,1\}^{n_j}$, $\sum_{j=0}^{t-1} n_j =n$, and $j\in\{0,1,\cdots, t-1\}$.

Next, we provide a detailed introduction to our DBVA in Algorithm \ref{algo1}. The whole quantum circuit is shown in Figure \ref{dbvaalg}.

\begin{breakablealgorithm}
\caption{Distributed Bernstein-Vazirani algorithm (with $t$ computing nodes).}
\label{algo1}
\begin{algorithmic}
\noindent \textbf{Input}: (1) The number of original qubits $n$; (2) A function $f_s (x)=\langle s\cdot x\rangle = (\sum_{i=0}^{n-1}s_i\cdot x_i)~\text{mod}~2$, where $x\in\{0,1\}^n$; (3) The number of computing nodes $t$; (4) The number of qubits of the $j$-th computing node $n_{j}$, where $\sum_{j=0}^{t-1} n_j =n$ and $j \in \{0,1,\cdots,t-1 \}$.

\noindent \textbf{Output}: The string $s=s_{n_0}s_{n_1}\cdots s_{n_{t-1}}\in\{0,1\}^{n}$, where $s_{n_j}\in\{0,1\}^{n_j}$.

\noindent \textbf{Procedure:}

\noindent \textbf{Step 1}: Let $j=0$.

\noindent \textbf{Step 2}: Initialize $n_j$ qubits: $\vert 0\rangle^{\otimes n_j}$.

\noindent \textbf{Step 3}: Apply $n_j$ Hamadard gates and get $\frac{1}{\sqrt{2^{n_j}}}\sum_{m_j\in\{0,1\}^{n_j}}\vert m_j\rangle$.

\noindent \textbf{Step 4}: Generate subfunction $f_{{s_{n_j}}}(m_j)=\langle s_{n_j}\cdot m_j\rangle: \{0,1\}^{n_j} \rightarrow \{0,1\}$ according to function $f_s (x)$ and its corresponding Oracle $O_{f_{s_{n_j}}(m_j)}$.

\noindent \textbf{Step 5}: Query the Oracle $O_{f_{{s_{n_j}}}(m_j)}$ and get $\frac{1}{\sqrt{2^{n_j}}}\sum_{m_j\in\{0,1\}^{n_j}} (-1)^{\langle s_{n_j}\cdot m_j\rangle} \vert m_j\rangle$.

\noindent \textbf{Step 6}: Perform $n_j$ Hamadard gates and get $\vert s_{n_j}\rangle$.

\noindent \textbf{Step 7}: Measure each qubit in the basis $\{\vert 0\rangle, \vert 1\rangle\}$ and obtain $s_{n_j}$.

\noindent \textbf{Step 8}: Let $j=j+1$.

\noindent \textbf{Step 9}: Repeat Step 2 to Step 8 until $j=t$.
\end{algorithmic}
\end{breakablealgorithm}

\begin{figure}[H]
\centering
\includegraphics[width=4.5in]{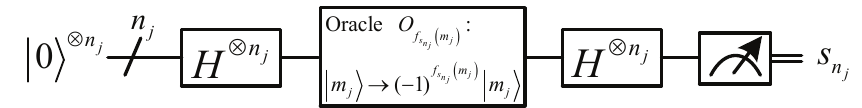}
\caption{\label{dbvaalg} Quantum circuit of DBVA with $t$ computing nodes.}
\end{figure}

The hidden string $s=s_{n_0}s_{n_1}\cdots s_{n_{t-1}}\in\{0,1\}^{n}$ will be exactly acquired after following the above algorithm steps. 
Hence, the entire process of our DBVA with $t$ computing nodes is achieved.

\section{Distributed exact Grover's algorithm}\label{DEGA}
In this section, we mainly describe our distributed exact Grover's algorithm (DEGA), which decomposes the original search problem into $\lfloor n/2 \rfloor$ parts to solve.

\subsection{Subfunctions}\label{subdega}
Similarly, we need to describe our partitioning strategy for the original function before providing our algorithm. We will generate a total of $\lfloor n/2 \rfloor$ subfunctions $g_i (m_i)$ for our DEGA, where $i \in \{0,1,\cdots, \lfloor n/2 \rfloor-1\}$. The specific generation method is as follows.

When $i \in \{0,1,\cdots, \lfloor n/2 \rfloor-2\}$, we choose first $2i$ and last $n-2(i+1)$ bits in $x$ to divide $f$, then we can get $2^{n-2}$ subfunctions $f_{{i},j}: \{0,1\}^2 \rightarrow \{0,1\}$:
\begin{eqnarray}\label{fij}
f_{{i},j}(m_i)=f(\underbrace{y_{j_{0}} y_{j_{1}}\cdots y_{j_{2i-1}}}_{2i}  m_i \underbrace{y_{j_{2i}} y_{j_{2i+1}}\cdots y_{j_{n-3}}}_{n-2(i+1)}),
\end{eqnarray}
where $m_i\in\{0,1\}^{2}$ and $y_{j_{0}} y_{j_{1}}\cdots y_{j_{2i-1}} y_{j_{2i}} y_{j_{2i+1}}\cdots y_{j_{n-3}} \in\{0,1\}^{n-2}$ is the binary representation of $j\in\{0,1,\cdots,$ $2^{n-2}-1\}$. Afterwards, we generate a new function $g_i: \{0,1\}^2 \rightarrow \{0,1\}$ in terms of these subfunctions,
\begin{eqnarray}\label{gi}
g_{i}(m_i)=\text{OR} (f_{{i},0}(m_i), f_{{i},1}(m_i), \cdots, f_{{i},2^{n-2}-1}(m_i) ),
\end{eqnarray}
where $m_i\in\{0,1\}^{2}$. Besides,
\begin{eqnarray}\label{orfunction}
\text{OR} (x) =
\begin{cases}
1,\vert x\vert \ge 1,\\
0,\vert x\vert =0,
\end{cases}
\end{eqnarray}
where $\vert x\vert $ is the Hamming weight of $x\in\{0,1\}^{n}$ (its number of 1's). It means we have acquired $\lfloor n/2 \rfloor-1$ subfunctions $g_i: \{0,1\}^2 \rightarrow \{0,1\}$, where $i \in \{0,1,\cdots, \lfloor n/2 \rfloor-2\}$.

When $i =\lfloor n/2 \rfloor-1$, we choose first $2i$ bits in $x$ to divide $f$, then we can get $2^{2i}$ subfunctions $f_{{i},j}: \{0,1\}^{n-2i} \rightarrow \{0,1\}$:
\begin{eqnarray}\label{flast}
f_{{i},j}(m_i)=f(\underbrace{y_{j_{0}} y_{j_{1}}\cdots y_{j_{2i-1}}}_{2i}  \underbrace{m_i}_{n-2i} ),
\end{eqnarray}
where $m_i\in\{0,1\}^{n-2i}$ and $y_{j_{0}} y_{j_{1}}\cdots y_{j_{2i-1}} \in\{0,1\}^{2i}$ is the binary representation of $j\in\{0,1,\cdots,$ $2^{2i}-1\}$. Afterwards, we also generate a new function $g_i: \{0,1\}^{n-2i} \rightarrow \{0,1\}$ by means of these subfunctions,
\begin{eqnarray}\label{glast}
g_{i}(m_i)=\text{OR} (f_{{i},0}(m_i), f_{{i},1}(m_i), \cdots, f_{{i},2^{2i}-1}(m_i) ),
\end{eqnarray}
where $m_i\in\{0,1\}^{n-2i}$. Specifically, when $n$ is an even number, we have
\begin{eqnarray}\label{even}
n - 2i = n- 2(\lfloor n/2 \rfloor-1)=n-2\left(\frac{n}{2} - 1\right)=2.
\end{eqnarray}
When $n$ is an odd number, we have
\begin{eqnarray}\label{odd}
n - 2i = n - 2(\lfloor n/2 \rfloor-1)=n -2\left(\frac{n-1}{2} - 1\right)=3.
\end{eqnarray}

\subsection{DEGA}
So far, we have successfully generated a total of $\lfloor n/2 \rfloor$ subfunctions $g_i (m_i)$ according to the original function $f(x)$ and $n$, where $i \in\{0,1,\cdots,$ $\lfloor n/2 \rfloor-1\}$. Additionally, we assume that obtaining the Oracle for each subfunction is simple. In other words, we can have $\lfloor n/2 \rfloor-1$ Oracles
\begin{eqnarray}\label{utaui}
U_{g_i (x)}:\vert x\rangle\rightarrow (-1)^{g_i (x)}\vert x\rangle,
\end{eqnarray}
where $x\in\{0,1\}^2$, $i \in\{0,1,\cdots,$ $\lfloor n/2 \rfloor-2\}$, and $\tau_i$ is the target string of $g_i (x)$ such that $g_i (\tau_i)=1$.

For the last subfunctions $g_i: \{0,1\}^{n-2i} \rightarrow \{0,1\}$, we need to first determine the parity of $n$. If $n$ is an even number, we can also have the Oracle as in Eq.~\eqref{utaui}, where $i =\lfloor n/2 \rfloor-1$. Otherwise, we can have another Oracle
\begin{eqnarray}\label{rtaui}
R_{g_i (x)}:\vert x\rangle\rightarrow e^{i\phi \cdot g_i (x)}\vert x\rangle,
\end{eqnarray}
where $x\in\{0,1\}^3$, $i =\lfloor n/2 \rfloor-1$, $\tau_i$ is the target string of $g_i (x)$ such that $g_i (\tau_i)=1$, $\phi=2\arcsin\left(\sin\left(\frac{\pi}{4J+6}\right) / \sin \theta\right)$, $J=\lfloor(\pi/2-\theta)/(2\theta)\rfloor$, and $\theta= \arcsin \left({\sqrt{1/2^3}}\right)$.

Next, we provide a detailed introduction to our DEGA in Algorithm \ref{algodega}. The whole quantum circuit is shown in Figure \ref{degacir}.

\begin{breakablealgorithm}
\caption{Distributed exact Grover's algorithm (includes only one target).}
\label{algodega}
\begin{algorithmic}
\STATE \textbf{Input}: (1) The number of qubits $n \ge 2$; (2) A function $f(x)$, where $f(x)=0$ for all $x\in\{0,1\}^n$ except $\tau$, for which $f(\tau)=1$; (3) $\lfloor n/2 \rfloor$ subfunctions $g_i (x)$ as in Eq.~\eqref{gi} and Eq.~\eqref{glast} generated according to $f(x)$ and $n$, where $i \in\{0,1,\cdots,$ $\lfloor n/2 \rfloor-1\}$.
\STATE \textbf{Output}: The target index string $\tau = \tau _{0}\tau_{1}\cdots\tau_{\lfloor n/2 \rfloor-1} \in\{0,1\}^n$ such that $f(\tau)=1$ with a probability of $100 \%$.
\STATE \textbf{Procedure:}
\STATE \textbf{Step 1.} Initialize $n$ qubits $\vert 0\rangle^{\otimes n}$.
\STATE \textbf{Step 2.} Apply $n$ Hadamard gates $H^{\otimes n}$.
\STATE \textbf{Step 3.} Let $i=0$.
\STATE \textbf{Step 4.} Execute Grover operator $G_i = -H^{\otimes 2} U_0 H^{\otimes 2}U_{g_i (x)}$ once, where $U_0=I_2-2(\vert 0\rangle\langle0\vert )^{\otimes 2}$ and $U_{g_i (x)}=I_2 - 2\vert \tau_i\rangle\langle\tau_i\vert $ act on the $(2i)$-th and $(2i+1)$-th qubits, and $\tau_i$ is the target string of $g_i (x)$ such that $g_i (\tau_i)=1$.
\STATE \textbf{Step 5.} Let $i=i+1$.
\STATE \textbf{Step 6.} Repeat Step 4 to Step 5 until $i=\lfloor n/2 \rfloor-1$.
\STATE \textbf{Step 7.} Determine the parity of $n$.
\STATE \textbf{Step 8.} If $n$ is an even number, execute Grover operator $G_i = -H^{\otimes 2} U_0 H^{\otimes 2}U_{g_i (x)}$ once, where $U_0=I_2-2(\vert 0\rangle\langle0\vert )^{\otimes 2}$ and $U_{g_i (x)}=I_2 - 2\vert \tau_i\rangle\langle\tau_i\vert $ act on the $(2i)$-th and $(2i+1)$-th qubits, and $\tau_i$ is the target string of $g_i (x)$ such that $g_i (\tau_i)=1$.
\STATE \textbf{Step 9.} If $n$ is an odd number, execute operator $L_i = -H^{\otimes 3} R_0 H^{\otimes 3}R_{g_i (x)}$ twice, where $R_0 = I_3 + (e^{i\phi}-1)(\vert 0\rangle\langle0\vert )^{\otimes 3}$ and $R_{g_i (x)} = I_3 + (e^{i\phi}-1)\vert \tau_i\rangle\langle\tau_i\vert $ act on the $(2i)$-th, $(2i+1)$-th, and $(2i+2)$-th qubits, and $\tau_i$ is the target string of $g_i (x)$ such that $g_i (\tau_i)=1$.
\STATE \textbf{Step 10.} Measure each qubit in the basis $\{\vert 0\rangle, \vert 1\rangle\}$.
\end{algorithmic}
\end{breakablealgorithm}

The target index string $\tau = \tau _{0}\tau_{1}\cdots\tau_{\lfloor n/2 \rfloor-1} \in\{0,1\}^n$ will be successfully searched with a probability of $100 \%$ after following the above algorithm steps.

Therefore, the entire process of our DEGA is accomplished.

\begin{figure}[H]
\centering
\includegraphics[width=6.4in]{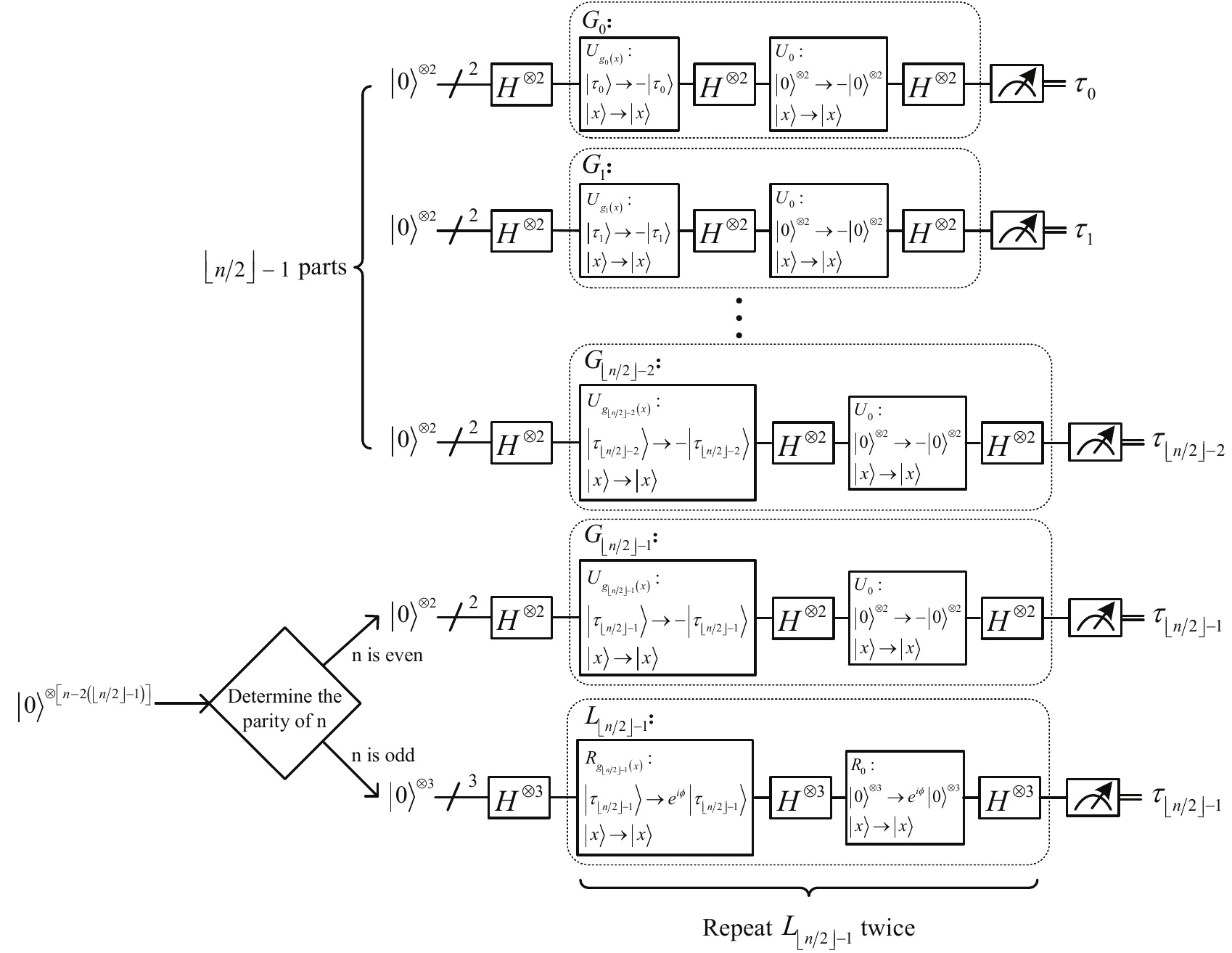}
\caption{\label{degacir} Quantum circuit of DEGA.}
\end{figure}

\section{Analysis}\label{analysis}
In this section, the analysis of our DBVA and DEGA is presented.



\subsection{The degree of $f_s (x)$}\label{degree}
We demonstrate the degree of $f_s (x)$ as in Eq.~\eqref{bvfdef} by the following theorem.

\begin{theorem}
The degree $f_s (x)$ is $\vert s\vert $, where $\vert s\vert $ means the Hamming weight of $s$ (its number of 1's).
\label{degreeth}
\end{theorem}

\begin{proof}
For the function considering in BV problem:
\begin{eqnarray}\label{bvfdefdeg}
f_s (x)=\langle s\cdot x\rangle = \left(\sum_{i=0}^{n-1}s_i\cdot x_i\right)~\text{mod}~2,
\end{eqnarray}
where $x\in\{0,1\}^n$. Once the string $s = s_0s_1s_2\cdots s_{n-1}\in\{0,1\}^n$ is determined, then the multilinear polynomial of $f_s (x)$ can be expressed as
\begin{eqnarray}\label{multilinear}
p_s(x)= \left(\sum_{i\in\{j\vert  s_j=1\}} s_i\cdot x_i\right)~\text{mod}~2= \left(\sum_{i\in\{j\vert  s_j=1\}} x_i\right)~\text{mod}~2,
\end{eqnarray}
where $x\in\{0,1\}^n$.  Consider the following function,
\begin{eqnarray}\label{parity}
\text{PARITY}_n(z) = \left(\sum_{i=0}^{n-1} z_i\right)~\text{mod}~2,
\end{eqnarray}
where $z = z_0z_1z_2\cdots z_{n-1}\in\{0,1\}^n$. Furthermore, the degree of $\text{PARITY}_n(z)$ has been proved to be $n$ \cite{RefBeals1998}. Intuitively, the degree of $f_s (x)$ is $\vert \{j\vert  s_j=1\}\vert =\vert s\vert $. Note that $\vert s\vert $ means the Hamming weight of $s$ (its number of 1's).

Therefore, the degree $f_s (x)$ is $\vert s\vert $.
\end{proof}

\subsection{Correctness of DBVA}\label{correctnessdbva}
In this subsection, we focus on verifying the correctness of our DBVA. We demonstrate the correctness of Algorithm \ref{algo1} by the following Theorem \ref{tm1correctnessdbva}.

\begin{theorem}
(\textbf{Correctness of DBVA})
For an $n$-qubit BV problem, suppose that the hidden string $s\in\{0,1\}^n$  be expressed as $s=s_{n_0}s_{n_1}\cdots s_{n_{t-1}}\in\{0,1\}^{n}$, where $s_{n_j}\in\{0,1\}^{n_j}$ and $\sum_{j=0}^{t-1} n_j =n$. We can obtain $s_{j}$ due to Algorithm \ref{algo1}, where $j\in\{0,1,\cdots, t-1\}$.
\label{tm1correctnessdbva}
\end{theorem}

\begin{proof}
It is straightforward.
\end{proof}

\subsection{Selection of subfunctions}\label{selection of subfunctions}
In the previous section, we introduced our partition strategy of the original function. In this subsection, we analyze the selection of subfunctions at the 0-th computing node in DBVA, and similar analysis results can be obtained at other computing nodes.

We choose last $k=n_1+n_2+\cdots+n_{t-1}$ bits in $x$ to divide $f$, then we can get $2^{k}$ subfunctions $f_p: \{0,1\}^{n-k} \rightarrow \{0,1\}$:
\begin{eqnarray}
f_p(w)&=&f(\underbrace{w_0w_1\cdots w_{n-k-1}}_{n-k} \underbrace{y_{p_{0}} y_{p_{1}}\cdots y_{p_{k-1}}}_{k}),
\end{eqnarray}
where $w=w_0w_1\cdots w_{n-k-1}\in\{0,1\}^{n-k}$, and $y_{p_{0}} y_{p_{1}}\cdots y_{p_{k-1}}\in\{0,1\}^{k}$ is the binary representation of $p\in\{0,1,\cdots,2^k-1\}$. We keep the subfunction when $p=0$, and denote it as
\begin{eqnarray}
f_{{s_{n_0}}}(m_0) = f(m_0\underbrace{00\cdots00}_{n_1+n_2+\cdots+n_{t-1}}),
\end{eqnarray}
where $m_0\in\{0,1\}^{n_0}$. 

In fact, no matter which subfunction is selected, the final result will not be affected. 
For any $p\in\{0,1,\cdots,2^k-1\}$, the subfunction $f_p(w)$ will be determined accordingly, the final substring is always $s_0s_1\dots s_{n-k} \in \{0,1\}^{n-k}$.

%

The similar analysis results can be obtained at other computing nodes.

\subsection{The depth of circuir}\label{circuitdepthdbva}
In order to compare the circuit depth, we give the following definition.

\begin{definition} \label{defdepth}
(\textbf{Depth of circuit}) The depth of circuit is defined as the longest path from the input to the output, moving forward in time along qubit wires.
\end{definition}

The depth of the circuit on the left, for example, is 1, whereas the right is 11. (see Figure \ref{depthexample})

\begin{figure}[H]
\centering
\includegraphics[width=4in]{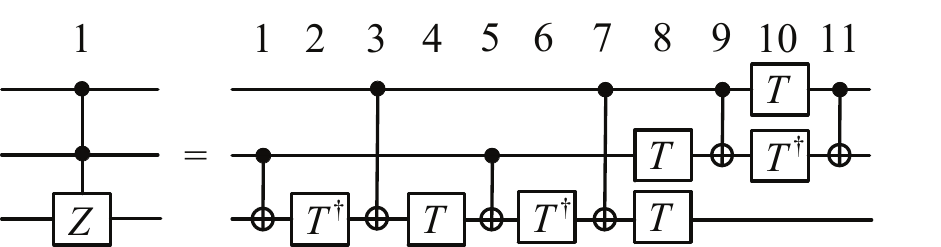}
\caption{\label{depthexample} Equivalent quantum circuit of Toffoli gate (or $C^2 Z$ gate), where $T$ represent $\pi / 8$ gate and $T^{\dagger}$ represents the conjugate transpose of $T$.}
\end{figure}


In the process of building the circuit of  BV algorithm, we will run into an issue of how to construct the Oracle required by the algorithm. Due to the algorithm for constructing quantum circuits to realizing Oracle by Qiu et al. \cite{RefQiu2022}, we here have the construction for the Oracles immediately as follows.

Given a function $f: \{0,1\}^n \rightarrow \{0,1\}$. Suppose there are $k$ inputs $x^{(i)}=(x_{0}x_{1}\cdots x_{n-1})^{(i)}$ satisfying $f(x^{(i)})=1$, where $i\in\{0,1,\cdots,k-1\}$. The Oracle construction is as follows:

\noindent \textbf{Step 1.} Let $i=0$.

\noindent \textbf{Step 2.} For $x^{(i)}$, execute Pauli $X$ gate on $j$-th qubit if $(x_{j})^{(i)}=0$, otherwise execute identity gate $I$, where $j\in\{0,1,\cdots,n-1\}$.

\noindent \textbf{Step 3.} Apply $C^{n-1}Z$ gate once.

\noindent \textbf{Step 4.} Execute the same operation as in Step 2.

\noindent \textbf{Step 5.} Let $i=i+1$.

\noindent \textbf{Step 6.} Repeat Step 2 to Step 5 until $i=k$.

\noindent \textbf{Step 7.} Summarize all quantum gates in the above steps.

Following the above steps, we can easily construct the quantum circuit corresponding to Oracle $O_f$.

For example, a Boolean function $f: \{0,1\}^3 \rightarrow \{0,1\}$ is defined as follows:
\begin{equation}\label{example}
f(x)=\begin{cases}1,x = 010, 101,\\
0,\text{otherwise}.
\end{cases}
\end{equation}
According to the above method, its corresponding Oracle $O_f$ quantum circuit is shown below (see Figure \ref{figure2}).

\begin{figure}[H]
\centering
\includegraphics[width=2.5in]{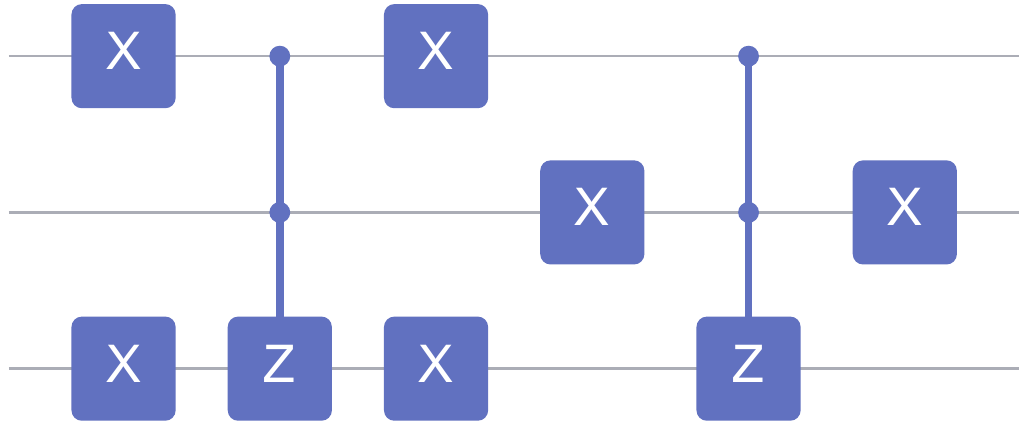}
\caption{\label{figure2} Quantum circuit corresponding to Oracle $O_f$ in the example.}
\end{figure}


Once we figure out how to build the circuit realizing Oracle, we can get the circuit depth of BV algorithm.  We need the following lemma.

\begin{lemma} \label{targetnum}
For any $s = s_0s_1s_2\dots s_{n-1}\in\{0,1\}^n/0^n$, we have $\vert \{x\vert  f_s (x)=1\}\vert =2^{n-1}$, where $f_s (x)$ as in Eq.~\eqref{bvfdef}.
\end{lemma}

\begin{proof}
First of all, $f_s (x)=\langle s\cdot x\rangle= \left(\sum_{i=0}^{n-1}s_i\cdot x_i\right)~\text{mod}~2: \{0,1\}^{n} \rightarrow \{0,1\}$, where $s\in\{0,1\}^n$. Suppose $k\neq0$ bits in $s$ are 1, and the rest are 0. Without loss of generality, for the purpose of analysis, $s$ can be rearranged so that the bits with the value of 1 are all on the left, as shown below
\begin{eqnarray}
s &=& \underbrace{s_{i_0} s_{i_1} \cdots s_{i_{k-1}}}_{k} \underbrace{s_{j_0} s_{j_1} \cdots s_{j_{n-k-1}}}_{n-k} \nonumber \\
&=& \underbrace{11\cdots11}_{k} \underbrace{00\cdots00}_{n-k}.
\end{eqnarray}
Then we have
\begin{eqnarray}
f_s (x) &=& \left(\sum_{i=0}^{k-1} 1 \cdot x_i + \sum_{i=k}^{n-1} 0 \cdot x_i \right)~\text{mod}~2 \nonumber \\
        &=& \left(\sum_{i=0}^{k-1} x_i \right)~\text{mod}~2. \label{kparity}
\end{eqnarray}
It means that once we determine the first $k$ bits in $x\in\{0,1\}^n$, no matter what the value of the last $n-k$ bits in $x$ is, it will not affect the final function value. Therefore, we can regard Eq.~\eqref{kparity} as a parity function: $\{0,1\}^{k} \rightarrow \{0,1\}$:
\begin{eqnarray}
\text{PARITY}_k(z) = \left(\sum_{i=0}^{k-1} z_i\right)~\text{mod}~2.
\end{eqnarray}
Obviously, since $\text{PARITY}_k(z)$ is a balanced function (half of the inputs will output 0, and the other will output 1), there are $2^{k-1}$ strings such that $\text{PARITY}_k(z)=1$. Finally, we add the remaining $n-k$ bits, which means there are $2^{k-1}\cdot 2^{n-k}=2^{n-1}$ strings making $f_s (x)=1$. In other words, we have $\vert \{x\vert  f_s (x)=1\}\vert =2^{n-1}$.
\end{proof}

We can build a complete quantum circuit of  BV algorithm by the Oracle construction method introduced above. According to Lemma \ref{targetnum}, there are $2^{n-1}$ target states that need to flip the phase. Furthermore, we should execute a quantum circuit with a circuit depth of 3 to realize phase flipping of a target state. In addition, $H^{\otimes n}$ needs to be executed twice. So the circuit depth of the BV algorithm satisfies
\begin{eqnarray}
dep(\text{BV}) &&=
\begin{cases}
3 \cdot 2^{n-1}+2 , & \text{target states do not contain}~\vert 1\rangle^{\otimes n}, \nonumber \\
3 \cdot 2^{n-1}-2+2, & \text{target states contain}~\vert 1\rangle^{\otimes n},
\end{cases}\\
&&=
\begin{cases}
3 \cdot 2^{n-1}+2 , & \text{target states do not contain}~\vert 1\rangle^{\otimes n}, \nonumber \\
3 \cdot 2^{n-1}, & \text{target states contain}~\vert 1\rangle^{\otimes n},
\end{cases}\\
&&\leq 3 \cdot 2^{n-1}+2 \label{BVbeforeoptidepth}
\end{eqnarray}
When the target states contain $\vert 1\rangle^{\otimes n}$, it is not necessary to execute Step 2 and Step 4 of the Oracle construction circuit during flipping the phase of $\vert 1\rangle^{\otimes n}$. So the circuit depth can be reduced by 2.

Additionally, since $X$ gate can be executed in parallel on different qubits and $X^2=I$, the circuit can be further optimized. For example, the circuit depth of Figure \ref{figure2} is 6 before optimization and 5 after optimization. Thus, the circuit depth of  BV algorithm after optimization satisfies
\begin{eqnarray}
dep(\text{BV, optimized}) &&=
\begin{cases}
2 \cdot 2^{n-1} +1+2 , & \text{target states do not contain}~\vert 1\rangle^{\otimes n},\nonumber \\
2 \cdot 2^{n-1} +1-1+2, & \text{target states contain}~\vert 1\rangle^{\otimes n},
\end{cases}\\&&=
\begin{cases}
2^n +3 , & \text{target states do not contain}~\vert 1\rangle^{\otimes n},\nonumber \\
2^n +2, & \text{target states contain}~\vert 1\rangle^{\otimes n},
\end{cases}\\
&&\leq 2^n +3. \label{BVafterdepth}
\end{eqnarray}

Similarly, we can get the circuit depth of DBVA by the following theorem.

\begin{theorem}\label{DBVAdepth}
The circuit depth of Algorithm \ref{algo1} is not greater than $2^{\text{max}(n_0, n_1, \cdots, n_{t-1})}+3$.
\end{theorem}

\begin{proof}
Since each computing node will execute a quantum circuit of small-scale BV algorithm, we just need to consider the computing node with the largest number when calculating the circuit depth. Therefore, we can know that the circuit depth of DBVA after optimization satisfies
\begin{eqnarray}
dep(\text{DBVA, optimized}) \leq 2^{\text{max}(n_0, n_1, \cdots, n_{t-1})}+3.
\end{eqnarray}
When the target states do not contain $\vert 1\rangle^{\otimes \text{max}(n_0, n_1, \cdots, n_{t-1})}$, the equal sign is true. When the target states contain $\vert 1\rangle^{\otimes \text{max}(n_0, n_1, \cdots, n_{t-1})}$, the circuit depth can be $2^{\text{max}(n_0, n_1, \cdots, n_{t-1})}+2$.
\end{proof}

\subsection{Correctness of DEGA}\label{correctnessdega}
To verify the correctness of our DEGA, it is equivalent to proving that we can obtain the target index string $\tau \in\{0,1\}^n$ by Algorithm \ref{algodega} with a probability of $100 \%$. Firstly, we give the following theorem.

\begin{theorem}\label{proofsubfunction}
Each subfunction $g_i (x)$ (as in Eq.~\eqref{gi} and Eq.~\eqref{glast}) has its corresponding target string $\tau_i$ such that $g_i (\tau_i)=1$, which satisfies
\begin{eqnarray}\label{targetstring}
\tau = \tau _{0}\tau_{1}\cdots\tau_{\lfloor n/2 \rfloor-1},
\end{eqnarray}
where $i \in\{0,1,\cdots,$ $\lfloor n/2 \rfloor-1\}$ and $\tau\in\{0,1\}^n$ is the target string of the original function $f(x)$.
\end{theorem}

\begin{proof} The method of proof is simple, so we omit the details.
\end{proof}

So far, we have proved that the target string corresponding to each subfunction $g_i (m_i)$ is $\tau_i$, which satisfies $\tau = \tau _{0}\tau_{1}\cdots\tau_{\lfloor n/2 \rfloor-1}$, where $i \in\{0,1,\cdots,$ $\lfloor n/2 \rfloor-1\}$.



Next, we demonstrate the correctness of Algorithm \ref{algodega} by the following Theorem \ref{proofcorrectness}.

\begin{theorem}\label{proofcorrectness}
(\textbf{Correctness of DEGA})
For an $n$-qubit search problem that includes only one target index string, suppose that the goal can be expressed as $\tau = x_{0} x_{1} \cdots x_{n-1}\in\{0,1\}^n$ of the original function $f(x)$. Then we can obtain $\tau = \tau _{0}\tau_{1}\cdots\tau_{\lfloor n/2 \rfloor-1}$ by Algorithm \ref{algodega} with a probability of $100 \%$, where $\tau_i$ is the target string of $g_i (x)$ and satisfies
\begin{eqnarray}
\tau_i =
\begin{cases}
x_{2i} x_{2i+1},          &i \in\{0,1,\cdots,$ $\lfloor n/2 \rfloor-2\},\\
x_{2i} x_{2i+1} \cdots x_{n-1}=
\begin{cases}
x_{2i} x_{2i+1} = x_{n-2} x_{n-1},          &n~\text{is even},\\
x_{2i} x_{2i+1}x_{2i+2} = x_{n-3} x_{n-2} x_{n-1}, &n~\text{is odd},
\end{cases}, &i =\lfloor n/2 \rfloor-1.
\end{cases}
\end{eqnarray}
\end{theorem}

\begin{proof}
From subsection \ref{subdega}, we have declared how to generate $\lfloor n/2 \rfloor$ subfunctions $g_i (x)$ as in Eq.~\eqref{gi} and Eq.~\eqref{glast} according to $f(x)$ and $n$, which decomposes the original search problem into $\lfloor n/2 \rfloor$ parts to solve, where $i \in\{0,1,\cdots,$ $\lfloor n/2 \rfloor-1\}$.

Furthermore, suppose that the goal can be expressed as $\tau = x_{0} x_{1} \cdots x_{n-1}\in\{0,1\}^n$ of the original function $f(x)$. Theorem \ref{proofsubfunction} has already explained each subfunction $g_i (x)$ has its corresponding target string $\tau_i$ such that $g_i (\tau_i)=1$, which satisfies
\begin{eqnarray}
\tau = \tau _{0}\tau_{1}\cdots\tau_{\lfloor n/2 \rfloor-1}.
\end{eqnarray}

Last but not least, we will accurately obtain the solution of each part. For $i \in\{0,1,\cdots,$ $\lfloor n/2 \rfloor-2\}$, by running the 2-qubit Grover's algorithm, we can precisely determine $\tau_i = x_{2i} x_{2i+1}\in\{0,1\}^2$. For $i =\lfloor n/2 \rfloor-1$, by running the 2-qubit Grover's algorithm when $n$ is even or the 3-qubit algorithm by Long when $n$ is odd, we can exactly obtain
\begin{eqnarray}
\tau_i = x_{2i} x_{2i+1} \cdots x_{n-1} =
\begin{cases}
x_{2i} x_{2i+1} = x_{n-2} x_{n-1} \in\{0,1\}^2,          &n~\text{is even},\\
x_{2i} x_{2i+1} x_{2i+2}= x_{n-3} x_{n-2} x_{n-1}\in\{0,1\}^3, &n~\text{is odd}.
\end{cases}
\end{eqnarray}

In conclusion, we can obtain $\tau = \tau _{0}\tau_{1}\cdots\tau_{\lfloor n/2 \rfloor-1}$ by Algorithm \ref{algodega} with a probability of $100 \%$, where $\tau_i$ is the target string of $g_i (x)$, where $i \in\{0,1,\cdots,$ $\lfloor n/2 \rfloor-1\}$.
\end{proof}

\subsection{Comparison}\label{comparisondega}


In order to compare the circuit depth, we will introduce a general construction for the Oracles in Grover's algorithm.

$U_{f(x)}=I_n - 2\vert \tau\rangle\langle\tau\vert$ is employed in the Grover operator $G = -H^{\otimes n} U_0 H^{\otimes n}U_{f(x)}$ to flip the phase of the target state $\vert \tau\rangle = \vert x_{0} x_{1} \cdots x_{n-1}\rangle$, where $\tau = x_{0} x_{1} \cdots x_{n-1}\in\{0,1\}^n$. Thus, the specific construction circuit of $U_{f(x)}$ (see Figure \ref{oracleOf}) is as follows:
\begin{eqnarray}
O_{U_{f(x)}} = \left(\bigotimes_{i=0} ^{n-1} X^{1-x_i}\right) \left(C^{n-1}Z\right) \left(\bigotimes_{i=0} ^{n-1} X^{1-x_i}\right),
\end{eqnarray}
where $X$ and $Z$ represent the Pauli-$X$ and Pauli-$Z$ gates, and $C^{n-1} Z$ will only flip the phase of $\vert 1\rangle ^{\otimes n}$,
\begin{eqnarray}
\left(C^{n-1} Z\right) \vert x \rangle =
\begin{cases}
-\vert x \rangle, & \vert x \rangle = \vert 1\rangle ^{\otimes n},\\
\vert x \rangle, & \text{otherwise},
\end{cases}
\end{eqnarray}
where $x = x_{0} x_{1} \cdots x_{n-1}\in\{0,1\}^n$. Besides, it should be noted that
\begin{eqnarray}
X^{1-x_i} =
\begin{cases}
X, & x_i = 0,\\
I, & x_i = 1,
\end{cases}
\end{eqnarray}
where $i \in\{0,1,\cdots,n-1\}$.
In particular, the construction circuit of $U_0=I_n-2(\vert 0\rangle\langle0\vert )^{\otimes n}$ (see Figure \ref{oracle00}) is
\begin{eqnarray}
O_{U_0} = \left( X ^{\otimes n} \right) \left(C^{n-1}Z\right) \left( X ^{\otimes n} \right).
\end{eqnarray}

\begin{figure}[H]
\centering
\begin{minipage}{0.33\textwidth}
\centering
\includegraphics[width=2.1in]{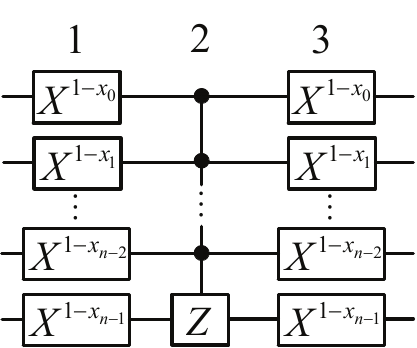}
\caption{\label{oracleOf}The quantum circuit corresponding to the Oracle $O_{U_{f(x)}}$.}
\end{minipage}
\qquad
\begin{minipage}{0.33\textwidth}
\centering
\includegraphics[width=1.9in]{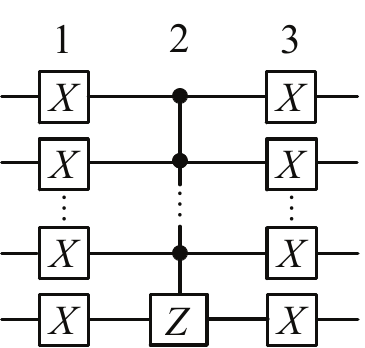}
\caption{\label{oracle00}The quantum circuit corresponding to the Oracle $O_{U_0}$.}
\end{minipage}
\end{figure}
It can  be found that the quantum circuit in Figure \ref{oracleOf} or Figure \ref{oracle00} has a circuit depth of 3.

Similarly, we can also give the specific construction circuit of $R_{f(x)} = I_n + (e^{i\phi}-1)\vert \tau\rangle\langle\tau\vert $ with $L = -H^{\otimes n} R_0 H^{\otimes n}R_{f(x)}$ (see Figure \ref{oracleRf}) as follows:
\begin{eqnarray}
O_{R_{f(x)}} = \left(\bigotimes_{i=0} ^{n-1} X^{1-x_i}\right)\left( C^{n-1}R(\phi)\right) \left(\bigotimes_{i=0} ^{n-1} X^{1-x_i}\right),
\end{eqnarray}
where $R(\phi)$ represents the rotation gate,
\begin{equation}
 R(\phi) = \left(
  \begin{array}{cc}
    1 & 0 \\
    0 & e^{i\phi}
  \end{array}
\right),
\end{equation}
and $C^{n-1}R(\phi)$ will only rotate the phase of $\vert 1\rangle ^{\otimes n}$,
\begin{eqnarray}
\left(C^{n-1}R(\phi)\right) \vert x \rangle =
\begin{cases}
e^{i\phi}\vert x \rangle, & \vert x \rangle = \vert 1\rangle ^{\otimes n},\\
\vert x \rangle, & \text{otherwise},
\end{cases}
\end{eqnarray}
where $x = x_{0} x_{1} \cdots x_{n-1}\in\{0,1\}^n$. In particular, the construction circuit of $R_0 = I_n + (e^{i\phi}-1)(\vert 0\rangle\langle0\vert )^{\otimes n}$  is
\begin{eqnarray}
O_{R_0} = \left( X ^{\otimes n} \right) \left(C^{n-1}R(\phi)\right) \left( X ^{\otimes n} \right).
\end{eqnarray}

It can be found that the quantum circuits in Figure \ref{oracleRf} or Figure \ref{oracleR0} also have a circuit depth of 3.


After identifying the specific quantum gates in  Grover's algorithm and the algorithm by Long, we might calculate the circuit depth of these two methods.

\begin{figure}[H]
\centering
\begin{minipage}{0.33\textwidth}
\centering
\includegraphics[width=2.1in]{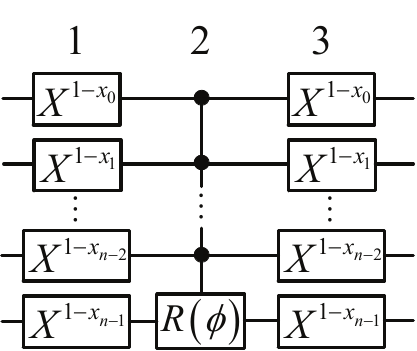}
\caption{\label{oracleRf}The quantum circuit corresponding to the Oracle $O_{R_{f(x)}}$.}
\end{minipage}
\qquad
\begin{minipage}{0.33\textwidth}
\centering
\includegraphics[width=1.9in]{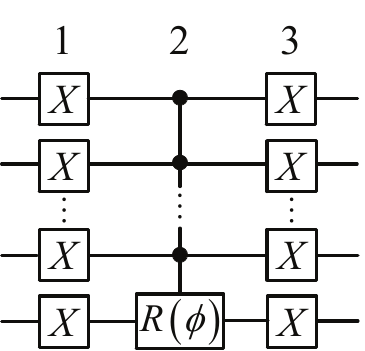}
\caption{\label{oracleR0}The quantum circuit corresponding to the Oracle $O_{R_0}$.}
\end{minipage}
\end{figure}

\begin{theorem}
The circuit depths for  Grover's algorithm and the algorithm by Long are $1 + 8\left\lfloor \frac{\pi}{4}\sqrt{2^n} \right\rfloor$ and $9 + 8\left( \left\lfloor \frac{\pi}{4}\sqrt{2^n} - \frac{1}{2}\right\rfloor \right)$, respectively.
\label{GroverLonddepth}
\end{theorem}

\begin{proof}
In the Grover's algorithm, Grover operator $G = -H^{\otimes n} U_0 H^{\otimes n}U_{f(x)}$ is executed with $\left\lfloor \frac{\pi}{4}\sqrt{2^n} \right\rfloor$ times, where $dep(U_f)=dep(U_0)=3$ and $dep(H^{\otimes n})=1$. Thus, the circuit depth of  Grover's algorithm is
\begin{eqnarray}
dep(\text{Grover}) &=& 1 + \left\lfloor \frac{\pi}{4}\sqrt{2^n} \right\rfloor * (3+1+3+1) \nonumber \\
                   &=& 1 + 8\left\lfloor \frac{\pi}{4}\sqrt{2^n} \right\rfloor.
\end{eqnarray}

Similarly, since the operator $L = -H^{\otimes n} R_0 H^{\otimes n}R_{f(x)}$ in the algorithm by Long is executed with $J+1$ times, where $J=\lfloor(\pi/2-\theta)/(2\theta)\rfloor$, $\theta= \arcsin \left(\sqrt{1/2^n}\right)$, and $dep(R_f)=dep(R_0)=3$, the circuit depth of the algorithm by Long is
\begin{eqnarray}
dep(\text{Long}) &=& 1 + \left( \left\lfloor \frac{\pi}{4}\sqrt{2^n} - \frac{1}{2}\right\rfloor + 1 \right) * (3+1+3+1) \nonumber \\
                 &=& 9 + 8\left( \left\lfloor \frac{\pi}{4}\sqrt{2^n} - \frac{1}{2}\right\rfloor \right).
\end{eqnarray}
\end{proof}


Next, we give the circuit depth of our DEGA by the following theorem.

\begin{theorem}
The circuit depth of Algorithm \ref{algodega} is $8(n~\text{mod}~2)+9$.
\label{DEGAdepth}
\end{theorem}

\begin{proof}
When $n$ is an even number, our DEGA will run $\lfloor n/2 \rfloor$ parts 2-qubit Grover's algorithm, and each part only needs to iterate the Grover operator $G$ once, so the circuit depth is
\begin{eqnarray}
dep(\text{DEGA}, n~\text{is even}) = 1 + 1 * (3+1+3+1) = 9.
\end{eqnarray}
When $n$ is an odd number, our DEGA will run $\lfloor n/2 \rfloor-1$ parts 2-qubit Grover's algorithm and 1 part 3-qubit  algorithm by Long, and the last part needs to iterate the operator $L$ twice, so the circuit depth is
\begin{eqnarray}
dep(\text{DEGA}, n~\text{is odd}) = 1 + 2 * (3+1+3+1) = 17.
\end{eqnarray}
In other words, the actual depth of our circuit is
\begin{eqnarray}
dep(\text{DEGA}) = 8(n~\text{mod}~2)+9 =
\begin{cases}
9,          &n~\text{is even},\\
17, &n~\text{is odd}.
\end{cases}
\end{eqnarray}
\end{proof}

Therefore, the actual circuit depth will be shallower compared with the original Grover's and the algorithm by Long. The circuit depth of our DEGA only depends on the parity of $n$, and it is not deepened as $n$ increases.

Next, we will contrast our DEGA with the existing distributed Grover's algorithms.


In \cite{RefQiu2022}, Qiu  et al. presented the serial and parallel distributed Grover's algorithms and divided the original function $f$ into $2^k$ subfunctions $f_p: \{0,1\}^{n-k} \rightarrow \{0,1\}$ as in Eq.~\eqref{subfunctions_fp}, where $2^k$ represents the number of computing nodes. In particular, the method in \cite {RefQiu2022} can solve the search problem having multiple objectives, but in the present paper we must realize the search problem only having unique objective, otherwise we cannot solve it here. The total number of qubits used in their parallel scheme by Qiu et al. is $2^k(n-k)$.

\section{Experiment}\label{experiment}
In this section, we provide particular situations of our DBVA and DEGA on MindQuantum. 




\subsection{The Bernstein-Vazirani algorithm, the hidden string $s=001011$}\label{BVMQ}
Let Boolean function $f: \{0,1\}^6 \rightarrow \{0,1\}$. Consider a function $f_s (x)=\langle s\cdot x\rangle = (\sum_{i=0}^{5}s_i\cdot x_i)~\text{mod}~2$, where $x\in\{0,1\}^6$. Suppose the hidden string we want to find is $s=001011$. The function values corresponding to all input strings are shown in TABLE \ref{examplefsx}. 
Then we can build a complete quantum circuit of the BV algorithm (see Figure \ref{fullcir}).

\begin{table}[h]
\centering
\scalebox{0.85}{
\begin{tabular}{ccc||ccc||ccc||ccc}
\toprule
$i$ & $x^{(i)}$ & $f_s(x^{(i)})$ & $i$ & $x^{(i)}$ & $f_s(x^{(i)})$ & $i$ & $x^{(i)}$ & $f_s(x^{(i)})$ & $i$ & $x^{(i)}$ & $f_s(x^{(i)})$ \\
\midrule
0 & 000000 & 0 & 16 & 010000 & 0 & 32 & 100000 & 0 & 48 & 110000 & 0 \\ 
1 & 000001 & 1 & 17 & 010001 & 1 & 33 & 100001 & 1 & 49 & 110001 & 1 \\ 
2 & 000010 & 1 & 18 & 010010 & 1 & 34 & 100010 & 1 & 50 & 110010 & 1 \\ 
3 & 000011 & 0 & 19 & 010011 & 0 & 35 & 100011 & 0 & 51 & 110011 & 0 \\ 
4 & 000100 & 0 & 20 & 010100 & 0 & 36 & 100100 & 0 & 52 & 110100 & 0 \\ 
5 & 000101 & 1 & 21 & 010101 & 1 & 37 & 100101 & 1 & 53 & 110101 & 1 \\ 
6 & 000110 & 1 & 22 & 010110 & 1 & 38 & 100110 & 1 & 54 & 110110 & 1 \\ 
7 & 000111 & 0 & 23 & 010111 & 0 & 39 & 100111 & 0 & 55 & 110111 & 0 \\ 
8 & 001000 & 1 & 24 & 011000 & 1 & 40 & 101000 & 1 & 56 & 111000 & 1 \\ 
9 & 001001 & 0 & 25 & 011001 & 0 & 41 & 101001 & 0 & 57 & 111001 & 0 \\ 
10 & 001010 & 0 & 26 & 011010 & 0 & 42 & 101010 & 0 & 58 & 111010 & 0 \\
11 & 001011 & 1 & 27 & 011011 & 1 & 43 & 101011 & 1 & 59 & 111011 & 1 \\
12 & 001100 & 1 & 28 & 011100 & 1 & 44 & 101100 & 1 & 60 & 111100 & 1 \\ 
13 & 001101 & 0 & 29 & 011101 & 0 & 45 & 101101 & 0 & 61 & 111101 & 0 \\ 
14 & 001110 & 0 & 30 & 011110 & 0 & 46 & 101110 & 0 & 62 & 111110 & 0 \\ 
15 & 001111 & 1 & 31 & 011111 & 1 & 47 & 101111 & 1 & 63 & 111111 & 1 \\ 
\bottomrule
\end{tabular}}
\caption{Function value of original function $f_s (x)$.}
\label{examplefsx}
\end{table}


\begin{figure}[H]
\centering
\includegraphics[width=6.5in]{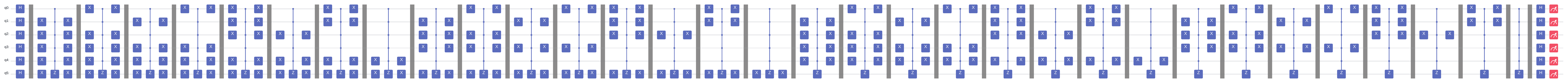}
\caption{\label{fullcir} The quantum circuit of  BV algorithm.}
\end{figure}

Furthermore, 
the above circuit can be further optimized (see Figure \ref{fullciropti}).
\begin{figure}[H]
\centering
\includegraphics[width=6.5in]{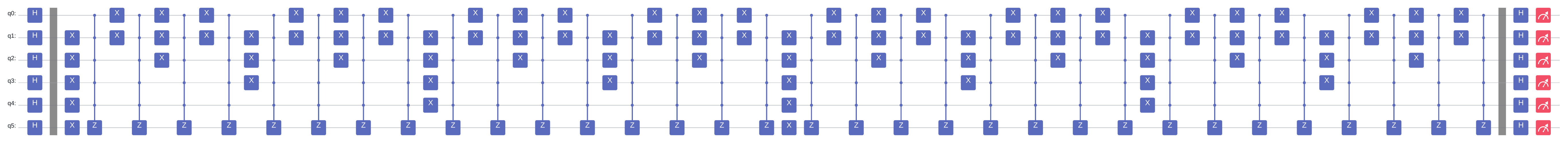}
\caption{\label{fullciropti} The quantum circuit of the BV algorithm after optimization.}
\end{figure}

The total number of quantum gates before optimization is 236, and after optimization is 130. The circuit depth is reduced from 96 to 66 after optimization.


Sampling the circuit 10,000 times, the results can be found in Figure \ref{measure001011}.
\begin{figure}[H]
\centering
\includegraphics[width=3in]{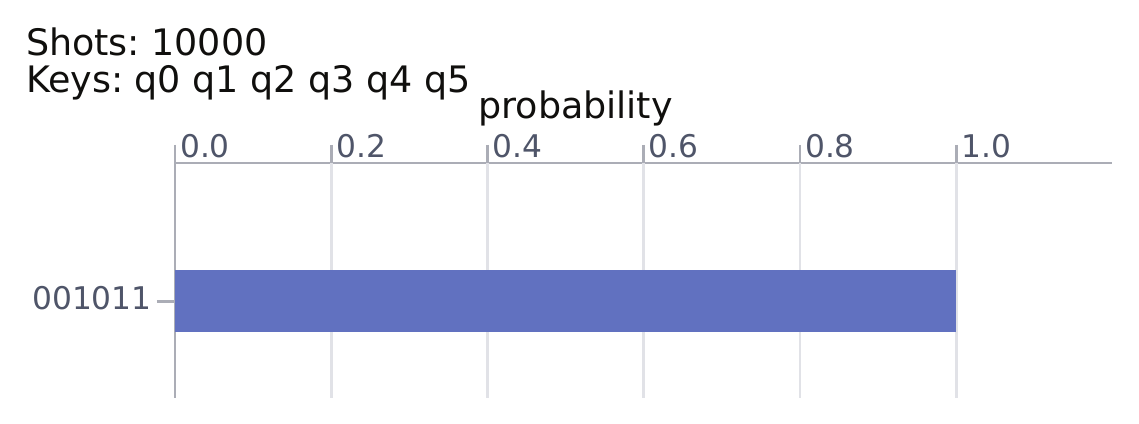}
\caption{\label{measure001011} The sampling results of BV algorithm after optimization.}
\end{figure}

The sampling results demonstrate that $s=001011$ is obtained after measurement. 


\subsection{The DBVA with two computing nodes, the hidden string $s=001011$}\label{experiment2}
The next step is to examine the quantum circuit of our DBVA with two computing nodes. Here, computing nodes $t=2$, and the number of qubits per computing node $n_0 = n_1 =3$. 

Firstly, generate two subfunctions $f_{s_{n_0}}(y)$ and $f_{s_{n_1}}(y)$ according to function $f_s (x)$:
\begin{eqnarray}\label{f3pq}
f_{s_{n_0}}(y)&=&f_s(y000),\\
f_{s_{n_1}}(y)&=&f_s(000y),
\end{eqnarray}
where $y\in\{0,1\}^{3}$. Each subfunction value for all input strings is shown in TABLE \ref{examplesub}.


\begin{table}[H]
\centering
\scalebox{0.85}{
\begin{tabular}{cccc}
\toprule
$i$  & $y^{(i)}$ & $f_{s_{n_0}}\left(y^{(i)}\right)$ & $f_{s_{n_1}}\left(y^{(i)}\right)$\\
\midrule
   0       & 000     & 0             & 0        \\
   1       & 001     & 1             & 1       \\ 
   2       & 010     & 0             & 1         \\
   3       & 011     & 1             & 0        \\
   4       & 100     & 0             & 0         \\
   5       & 101     & 1             & 1            \\
   6       & 110     & 0             & 1           \\ 
   7       & 111     & 1             & 0           \\
\bottomrule
\end{tabular}}
\caption{Function values of subfunctions $f_{s_{n_0}}(y)=f_s(y000)$ and $f_{s_{n_1}}(y)=f_s(000y)$.}
\label{examplesub}
\end{table}


We build the complete quantum circuit of our DBVA with two computing nodes (see Figure \ref{dbva2cir}). Note that since the reading order on MindQuantum is from right to left, which is the opposite of our reading order from left to right, the 3-qubit BV algorithm corresponding to $f_{s_{n_1}}(y)$ is located on the top of the circuit.

Furthermore, the optimized quantum circuit can be obtained (see Figure \ref{dbva2ciropti}).

The total number of quantum gates before optimization is 40, and after optimization is 36.  The circuit depth is reduced from 14 to 11 after optimization.

By sampling the circuit 10,000 times, the results can be found in Figure \ref{001011measureparallel11}.
\begin{figure}[H]
\centering
\includegraphics[width=4in]{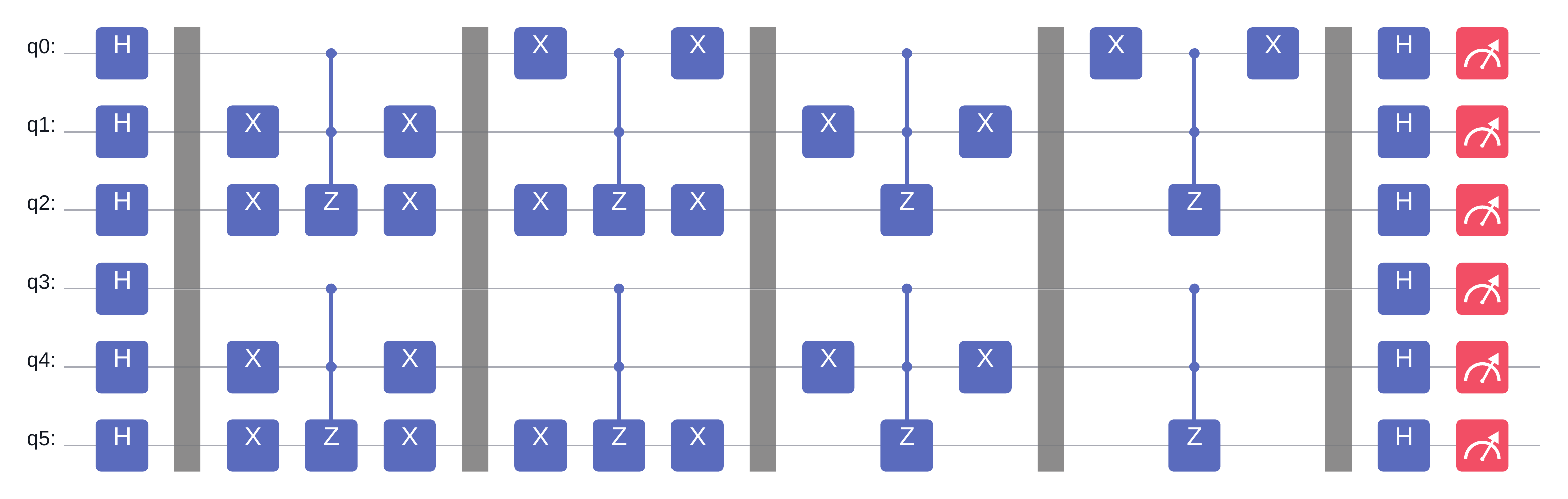}
\caption{\label{dbva2cir} The quantum circuit of DBVA with two computing nodes.}
\end{figure}

\begin{figure}[H]
\centering
\includegraphics[width=4in]{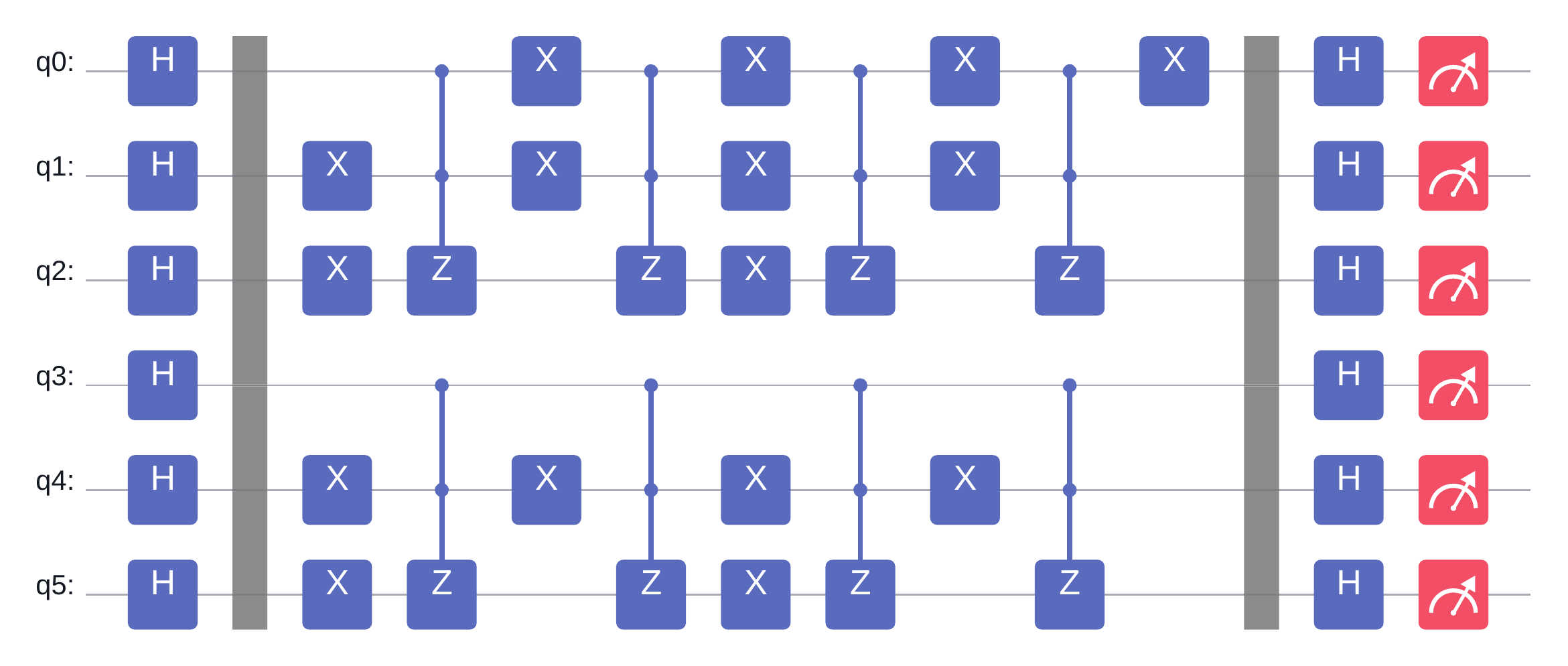}
\caption{\label{dbva2ciropti} The quantum circuit of DBVA with two computing nodes after optimization.}
\end{figure}

\begin{figure}[H]
\centering
\includegraphics[width=3in]{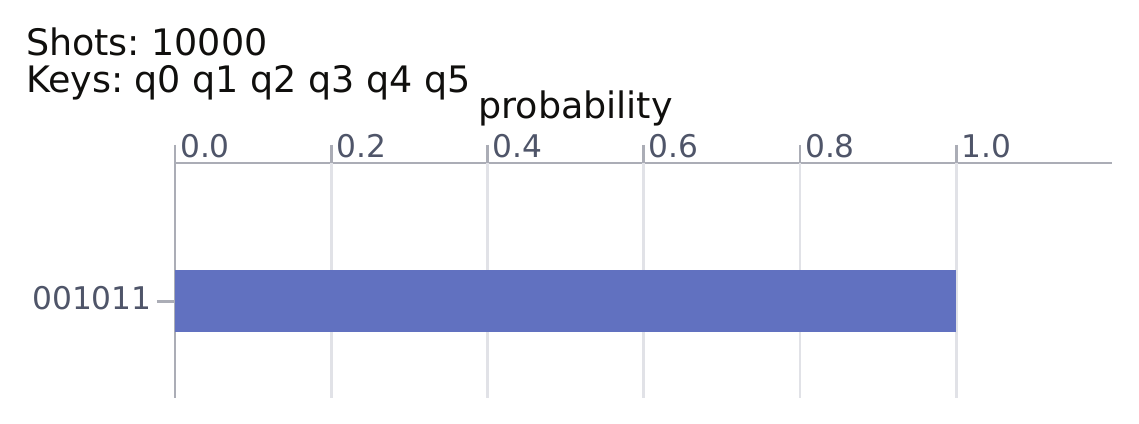}
\caption{\label{001011measureparallel11} The sampling results of DBVA with two computing nodes after optimization.}
\end{figure}

As can be seen from the sampling results, we definitely obtain the hidden string $s=s_{n_0}s_{n_1}=001011$ after our DBVA with two computing nodes, which is the same as the result of the original BV algorithm. 


\subsection{The DBVA with $t=3$ computing nodes, the hidden string $s=001011$}\label{experiment3}
Further, we consider the experiment of our DBVA with $t=3$ computing nodes. Here, computing nodes $t=3$, and the number of qubits per computing node $n_0 = n_1 =n_2 =2$.

Firstly, generate three subfunctions $f_{s_{n_0}}(m)$, $f_{s_{n_1}}(m)$, and $f_{s_{n_2}}(m)$ according to function $f_s (x)$:
\begin{eqnarray}\label{f2pq}
f_{s_{n_0}}(m)&=&f(m0000),\\
f_{s_{n_1}}(m)&=&f(00m00),\\
f_{s_{n_2}}(m)&=&f(0000m),
\end{eqnarray}
where $m\in\{0,1\}^{2}$. Each subfunction value for all input strings is shown in TABLE \ref{examplesubsub}.


\begin{table}[h]
  \centering
\scalebox{0.85}{
  \begin{tabular}{ccccc}
    \toprule
      $i$ & $m^{(i)}$ & $f_{s_{n_0}}\left(m^{(i)}\right)$  & $f_{s_{n_1}}\left(m^{(i)}\right)$ & $f_{s_{n_2}}\left(m^{(i)}\right)$\\
    \midrule
    0 & 00 & 0 & 0 & 0\\
    1 & 01 & 0 & 0 & 1 \\
    2 & 10 & 0 & 1 & 1 \\
    3 & 11 & 0 & 1 & 0\\
    \bottomrule
  \end{tabular}}
\caption{Function values of subfunctions $f_{s_{n_0}}=f(m0000)$, $f_{s_{n_1}}(m)=f(00m00)$, $f_{s_{n_2}}(m)=f(0000m)$.}
  \label{examplesubsub}
\end{table}

We can easily draw the whole quantum circuit of our DBVA with $t=3$ computing nodes (see Figure \ref{dbva3ciropti}). The 2-qubit BV algorithms corresponding to $f_{s_{n_2}}(y)$, $f_{s_{n_1}}(y)$ and $f_{s_{n_0}}(y)$ are located at the top, middle and bottom of the circuit respectively.

The number of quantum gates is 22 and the circuit depth is 7. By sampling the circuit 10,000 times, the results can be found in Figure \ref{001011measureparallel}.


\begin{figure}[H]
\centering
\includegraphics[width=3.8in]{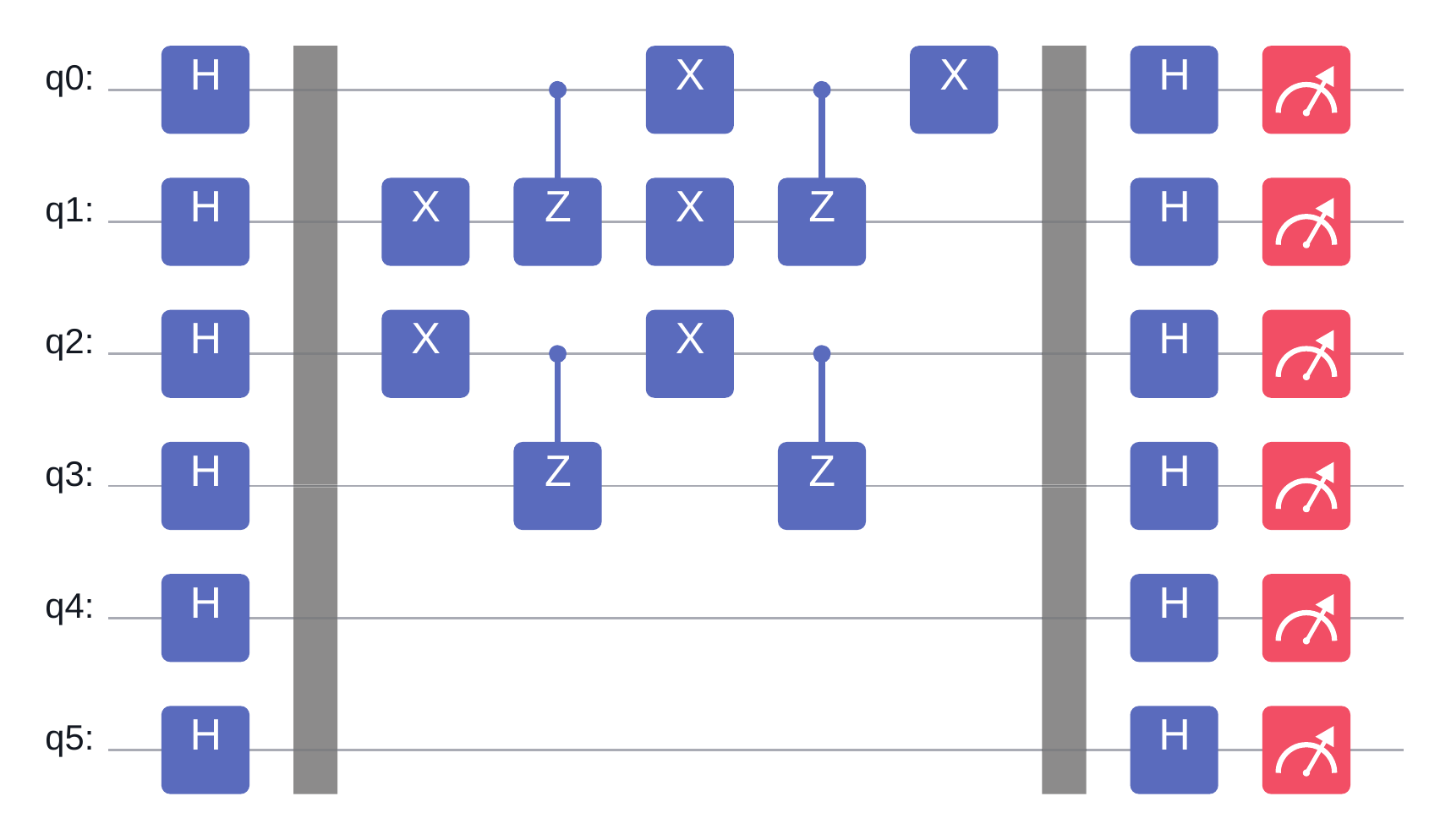}
\caption{\label{dbva3ciropti} The quantum circuit of DBVA with $t=3$ computing nodes.}
\centering
\includegraphics[width=3in]{001011measureparallel.pdf}
\caption{\label{001011measureparallel} The sampling results of DBVA with $t=3$ computing nodes after optimization.}
\end{figure}

As can be seen from the sampling results, we definitely obtain the hidden string $s=s_{n_0}s_{n_1}s_{n_2}=001011$ after our DBVA with $t=3$ computing nodes, which is the same as the result of the original BV algorithm. 


It can be found that the final result of DBVA with two or three computing nodes is consistent with the result of the original BV algorithm. However, DBVA requires fewer quantum gates and has a shallower circuit depth (see TABLE \ref{comparetable}). With the increase of the number of computing nodes (from 2 to 3), the number of quantum gates and the circuit depth are further reduced.


\begin{table}[h]
  \centering
\scalebox{0.9}{
  \begin{tabular}{lccccc}
    \toprule
    &BV & optimized BV & DBVA-2 & optimized DBVA-2 & optimized DBVA-3 \\ \midrule
    \text{1. The number of computing nodes} &1&1&2&2&3 \\
    \text{2. Figure number} &\text{Figure} \ref{fullcir}&\text{Figure} \ref{fullciropti}&\text{Figure} \ref{dbva2cir}&\text{Figure} \ref{dbva2ciropti}&\text{Figure} \ref{dbva3ciropti} \\ 
    \text{3. Final result}&001011&001011&001011&001011&001011 \\ 
    \text{4. The number of quantum gates}&236&130&40&36&22 \\ 
    \text{5. The circuit depth}&96&66&14&11&7 \\ 
    \bottomrule
  \end{tabular}}
  \caption{A simple comparison between our DBVA and the original BV algorithm.}
\label{comparetable}
\end{table}


\subsection{The 2-qubit DEGA, the target string $\tau = 01$}\label{dega2example}
Let Boolean function $f: \{0,1\}^2 \rightarrow \{0,1\}$. Suppose
\begin{eqnarray}
f(x)=
\begin{cases}
1,x=01,\\
0,x\neq 01 ,
\end{cases}
\end{eqnarray}
where $x\in\{0,1\}^2$ and $\tau = 01$ is the target string. When $n=2$, our DEGA is a 2-qubit Grover's algorithm. Thus, we can build its circuit (see Figure \ref{2bitgrover}). 
By sampling the circuit 10,000 times, the results can be found in Figure \ref{measure01}. The total number of quantum gates is 14 and the circuit depth is 9.

\begin{figure}[H]
\centering
\includegraphics[width=4in]{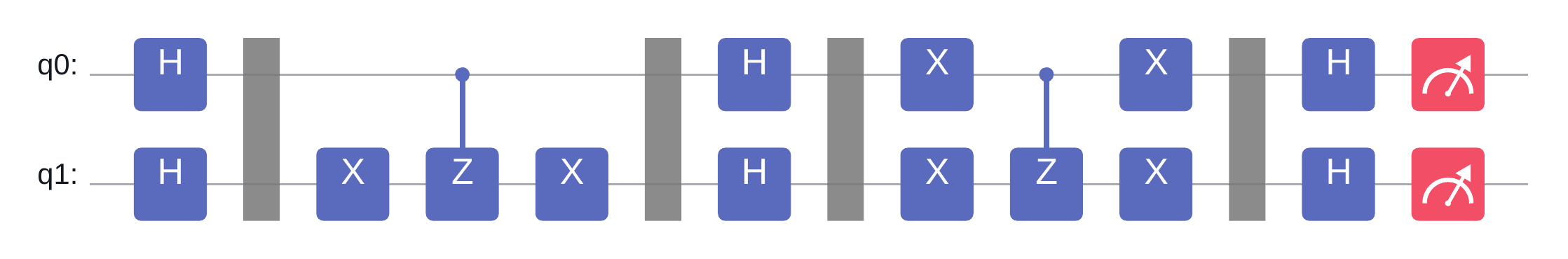}
\caption{\label{2bitgrover} The quantum circuit of the 2-qubit Grover's algorithm (target string $\tau = 01$).}
\centering
\includegraphics[width=3in]{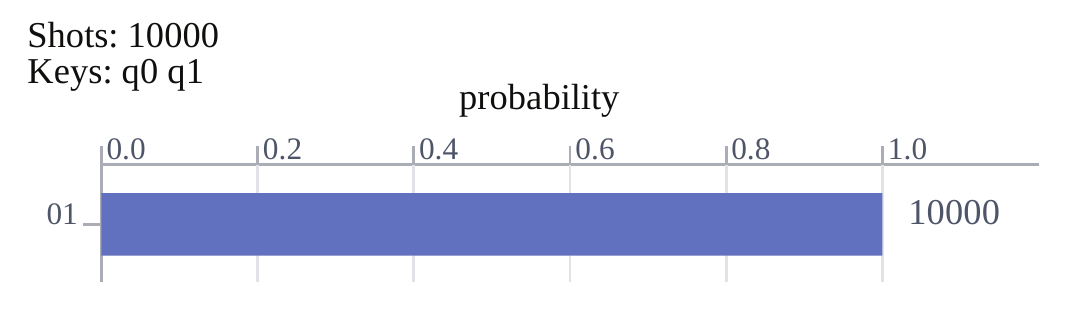}
\caption{\label{measure01} The sampling results of the 2-qubit Grover's algorithm (target string $\tau = 01$).}
\end{figure}

The sampling results demonstrate that $\tau=01$ is obtained exactly after measurement. 

\subsection{The 3-qubit DEGA, the target string $\tau = 101$}\label{dega3example}
Let Boolean function $f: \{0,1\}^3 \rightarrow \{0,1\}$. Suppose
\begin{eqnarray}
f(x)=
\begin{cases}
1,x=101,\\
0,x\neq 101 ,
\end{cases}
\end{eqnarray}
where $x\in\{0,1\}^3$ and $\tau = 101$ is the target string. When $n=3$, our DEGA is a 3-qubit algorithm by Long. Thus, we can build its circuit (see Figure \ref{3bitlong}). 
By sampling the circuit 10,000 times, the results can be found in Figure \ref{measure101}. Note that PhaseShift (called the PS gate) is a single-qubit gate, whose matric is as follows
\begin{eqnarray}
\text{PS($\phi$)} =
\left(
\begin{array}{cc}
1 & 0 \\
0 & e^{i\phi}
\end{array}
\right).
\end{eqnarray}

\begin{figure}[H]
\centering
\includegraphics[width=5.2in]{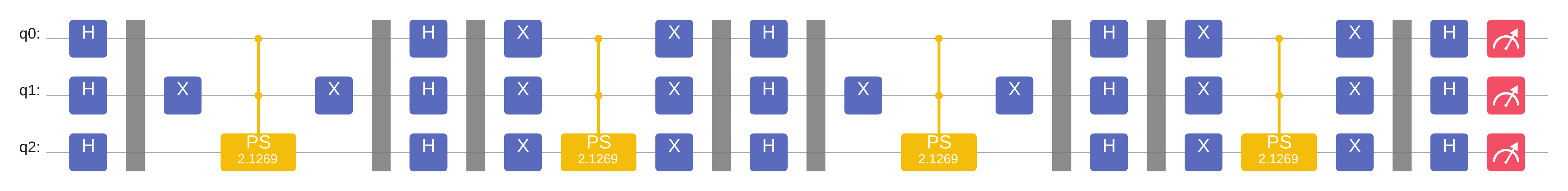}
\caption{\label{3bitlong} The quantum circuit of the 3-qubit algorithm  by Long (target string $\tau = 101$). The parameter in the circuit is $\phi=2\arcsin\left(\sin\left(\frac{\pi}{4J+6}\right) / \sin \theta\right) =2.1268800471555034\approx 2.1269$, where $J=\lfloor(\pi/2-\theta)/(2\theta)\rfloor$ and $\theta= \arcsin {\sqrt{1/2^3}}$.}
\centering
\includegraphics[width=3in]{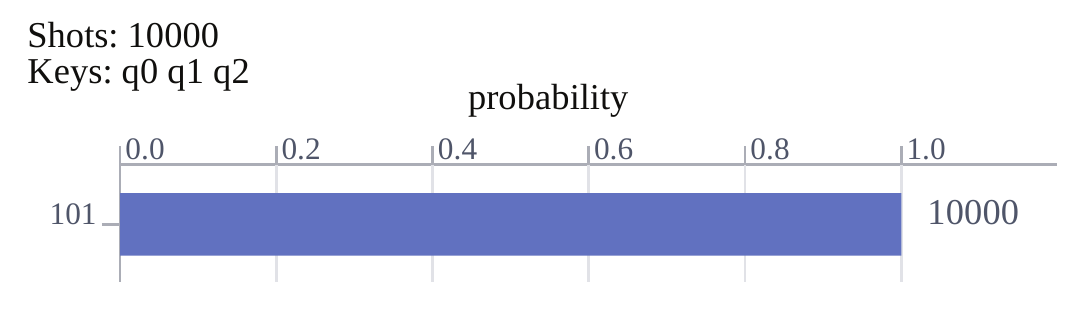}
\caption{\label{measure101} The sampling results of the 3-qubit algorithm by Long (target string $\tau = 101$).}
\end{figure}

The total number of quantum gates is 35 and the circuit depth is 17.

The sampling results demonstrate that $\tau=101$ is obtained exactly after measurement. 

\subsection{The 4-qubit DEGA, the target string $\tau = 1001$}\label{dega4example}
Let Boolean function $f: \{0,1\}^4 \rightarrow \{0,1\}$. Suppose
\begin{eqnarray}
f(x)=
\begin{cases}
1,x=1001,\\
0,x\neq 1001 ,
\end{cases}
\end{eqnarray}
where $x\in\{0,1\}^4$ and $\tau = 1001$ is the target string. When $n=4$, we can decomposes the original search problem into $\lfloor n/2 \rfloor=2$ parts to solve. To be specific, we choose last $2$ bits in $x$ to divide $f$, then we can get $2^2=4$ subfunctions $f_{{0},j}: \{0,1\}^2 \rightarrow \{0,1\}$:
\begin{eqnarray}
f_{{0},0}(m_0)&=&f( m_0 00), \\
f_{{0},1}(m_0)&=&f( m_0 01),\\
f_{{0},2}(m_0)&=&f( m_0 10),\\
f_{{0},3}(m_0)&=&f( m_0 11),
\end{eqnarray}
where $m_0\in\{0,1\}^{2}$ and $j\in\{0,1,2,3\}$. Obviously, $f_{{0},0}(m_0)=f_{{0},2}(m_0)=f_{{0},3}(m_0)\equiv0$, and
\begin{eqnarray}
f_{{0},1}(m_0)=
\begin{cases}
1,m_0=10,\\
0,m_0\neq 10,
\end{cases}
\end{eqnarray}
where $m_0\in\{0,1\}^2$ and $10$ is the target substring.

Afterwards, we generate a new function $g_0: \{0,1\}^2 \rightarrow \{0,1\}$ in terms of above four subfunctions,
\begin{eqnarray}
g_{0}(m_0)=\text{OR} \left(f_{{0},0}(m_0), f_{{0},1}(m_0), f_{{0},2}(m_0), f_{{0},3}(m_0)\right)=
\begin{cases}
1,m_0=10,\\
0,m_0\neq 10,
\end{cases}
\end{eqnarray}
where $m_0\in\{0,1\}^{2}$.

Similarly, we choose first $2$ bits in $x$ to divide $f$, then we can get $2^2=4$ subfunctions $f_{{1},j}: \{0,1\}^2 \rightarrow \{0,1\}$:
\begin{eqnarray}
f_{{1},0}(m_1)&=&f( 00 m_1 ), \\
f_{{1},1}(m_1)&=&f( 01 m_1 ),\\
f_{{1},2}(m_1)&=&f( 10 m_1 ),\\
f_{{1},3}(m_1)&=&f( 11 m_1 ),
\end{eqnarray}
where $m_1\in\{0,1\}^{2}$ and $j\in\{0,1,2,3\}$. Obviously, $f_{{1},0}(m_1)=f_{{1},1}(m_1)=f_{{1},3}(m_1)\equiv0$, and
\begin{eqnarray}
f_{{1},2}(m_1)=
\begin{cases}
1,m_1=01,\\
0,m_1\neq 01,
\end{cases}
\end{eqnarray}
where $m_1\in\{0,1\}^2$ and $01$ is the target substring.

Afterwards, we generate a new function $g_1: \{0,1\}^2 \rightarrow \{0,1\}$ in terms of above four subfunctions,
\begin{eqnarray}
g_{1}(m_1)=\text{OR} \left(f_{{1},0}(m_1), f_{{1},1}(m_1), f_{{1},2}(m_1), f_{{1},3}(m_1)\right)=
\begin{cases}
1,m_1=01,\\
0,m_1\neq 01,
\end{cases}
\end{eqnarray}
where $m_1\in\{0,1\}^{2}$.

Next, we can build the completed circuit (see Figure \ref{4bitdega}). The 2-qubit Grover's algorithm corresponding to $g_{1}(m_1)$ is located on the top of the circuit.
By sampling the circuit 10,000 times, the results can be found in Figure \ref{measure1001}.

\begin{figure}[H]
\centering
\includegraphics[width=4in]{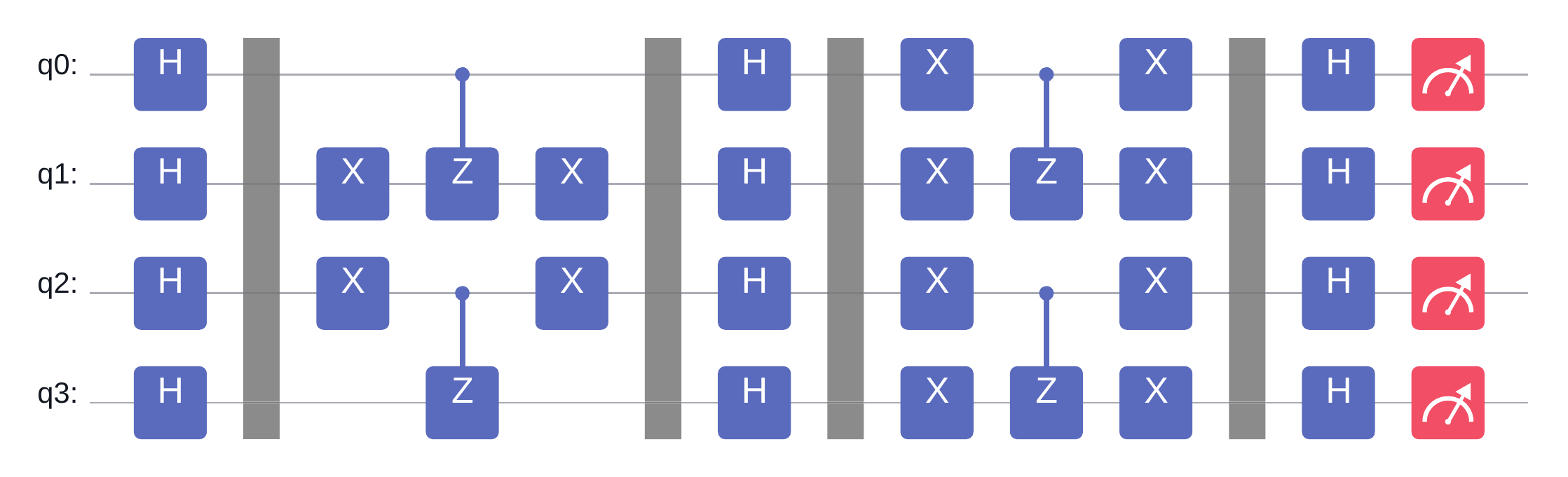}
\caption{\label{4bitdega} The quantum circuit of the 4-qubit DEGA (target string $\tau = 1001$).}
\end{figure}

\begin{figure}[H]
\centering
\includegraphics[width=3in]{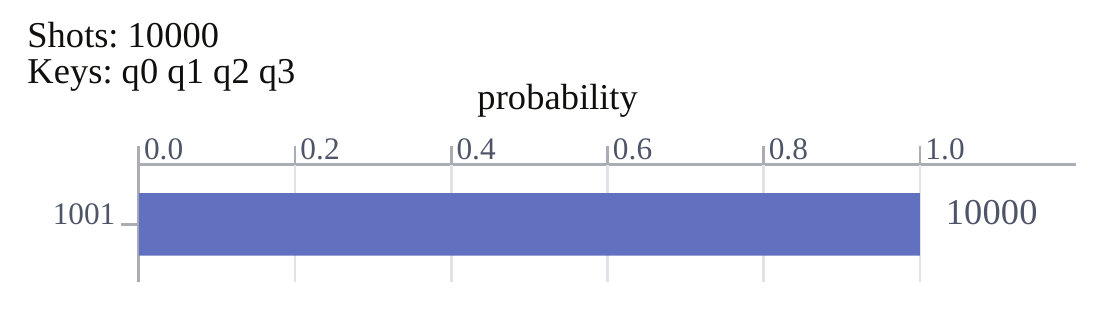}
\caption{\label{measure1001} The sampling results of the 4-qubit DEGA (target string $\tau = 1001$).}
\end{figure}

The total number of quantum gates is 28 and the circuit depth is 9.

The sampling results demonstrate that $\tau=1001$ is obtained exactly after measurement. 

If the Grover's algorithm or the algorithm by Long is employed to search 1001, we can build the completed circuits (see Figure \ref{grover1001} and Figure \ref{long1001}) respectively. 
By sampling the circuits 10,000 times respectively, the results can be found in Figure \ref{grovermeasure1001} and Figure \ref{longmeasure1001}. The number of quantum gates required by two algorithms is 70, and the circuit depth is 25.

\begin{figure}[H]
\centering
\includegraphics[width=4.5in]{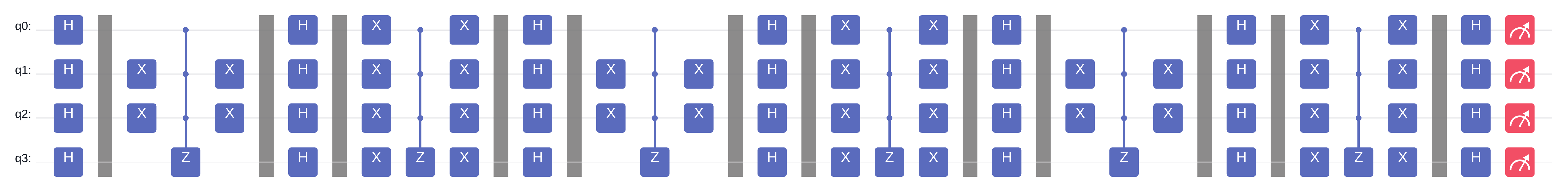}
\caption{\label{grover1001} The quantum circuit of the 4-qubit Grover's algorithm (target string $\tau = 1001$).}
%
\centering
\includegraphics[width=4.5in]{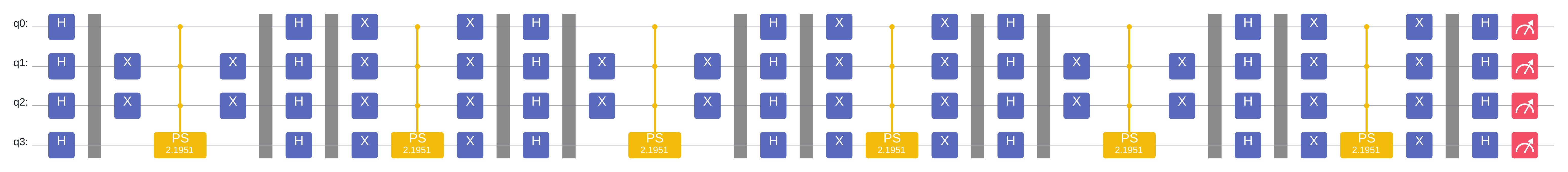}
\caption{\label{long1001} The quantum circuit of the 4-qubit algorithm by Long (target string $\tau = 1001$). The parameter in the circuit is $\phi=2\arcsin\left(\sin\left(\frac{\pi}{4J+6}\right) / \sin \theta\right) =2.195057699090115\approx 2.1951$, where $J=\lfloor(\pi/2-\theta)/(2\theta)\rfloor$ and $\theta= \arcsin {\sqrt{1/2^4}}$.}
\end{figure}

\begin{figure}[H]
\centering
\begin{minipage}{0.49\textwidth}
\centering
\includegraphics[width=2in]{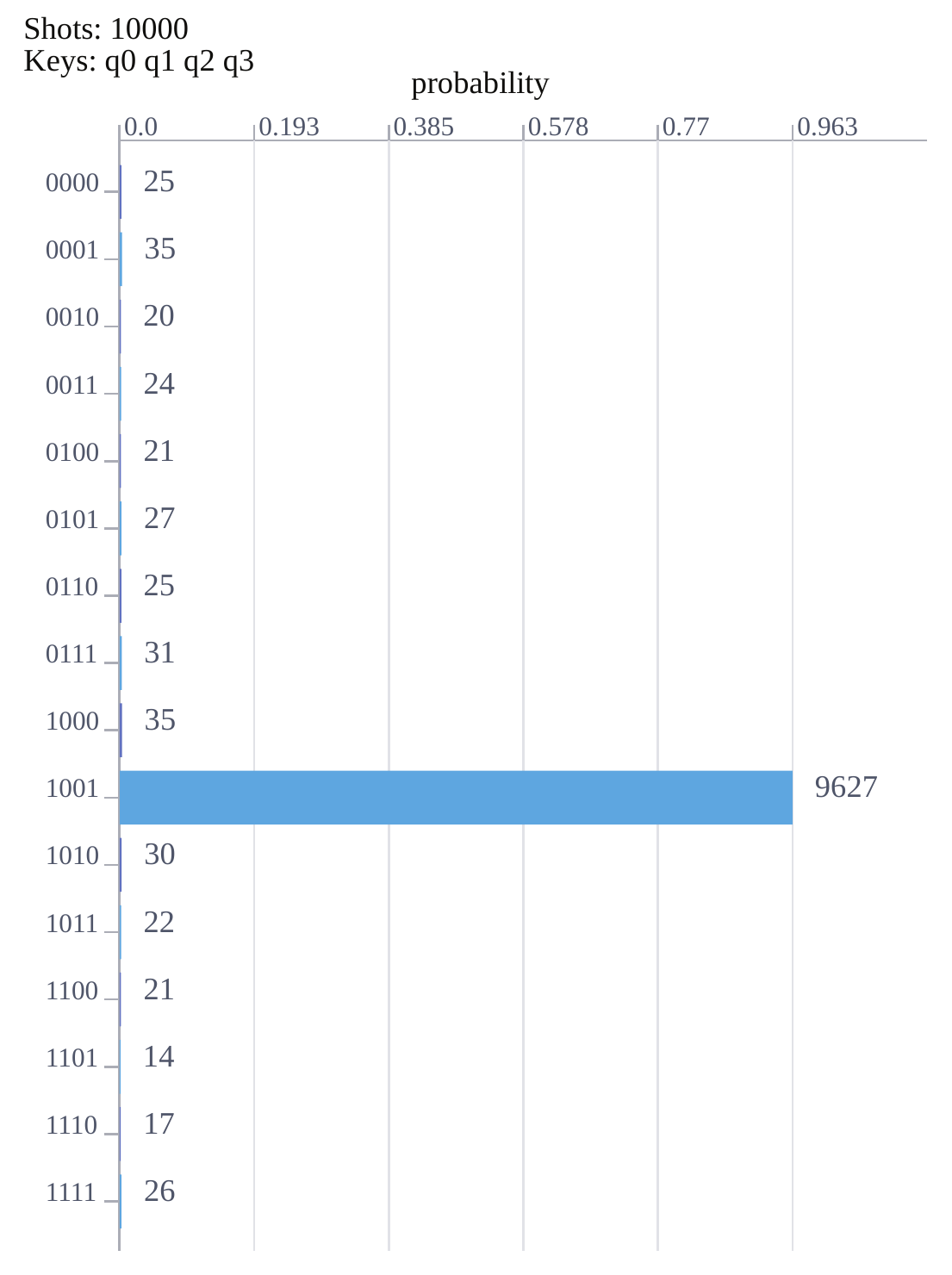}
\caption{\label{grovermeasure1001} The sampling results of the 4-qubit Grover's algorithm (target string $\tau = 1001$).}
\end{minipage}
\begin{minipage}{0.49\textwidth}
\centering
\includegraphics[width=2.5in]{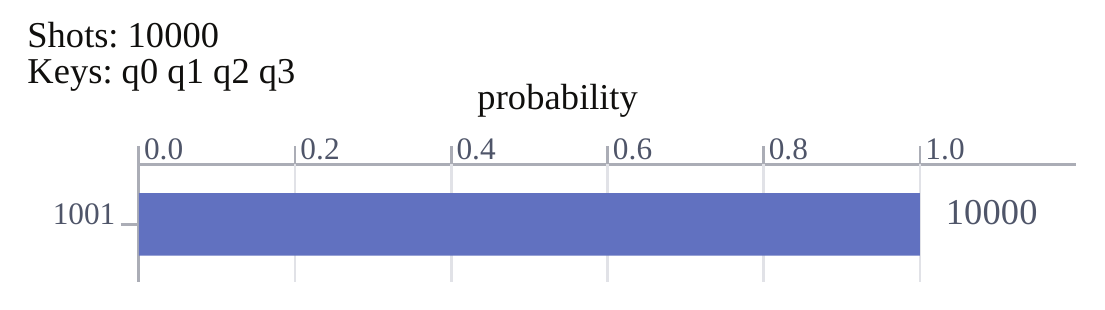}
\caption{\label{longmeasure1001} The sampling results of the 4-qubit algorithm by Long (target string $\tau = 1001$).}
\end{minipage}
\end{figure}

It can be found that the probability of the Grover's algorithm acquiring target string $\tau=1001$ is 0.9627, which is not exact, while the algorithm by Long can accurately obtain $\tau=1001$.

\subsection{The 5-qubit DEGA, the target string $\tau = 01001$}\label{dega5example}
Let Boolean function $f: \{0,1\}^5 \rightarrow \{0,1\}$. Suppose
\begin{eqnarray}
f(x)=
\begin{cases}
1,x=01001,\\
0,x\neq 01001 ,
\end{cases}
\end{eqnarray}
where $x\in\{0,1\}^5$ and $\tau = 01001$ is the target string. When $n=5$, we can decomposes the original search problem into $\lfloor n/2 \rfloor=2$ parts to solve. To be specific, we choose last $2$ bits in $x$ to divide $f$, then we can get $2^2=4$ subfunctions $f_{{0},j}: \{0,1\}^3 \rightarrow \{0,1\}$:
\begin{eqnarray}
f_{{0},0}(m_0)&=&f( m_0 00), \\
f_{{0},1}(m_0)&=&f( m_0 01),\\
f_{{0},2}(m_0)&=&f( m_0 10),\\
f_{{0},3}(m_0)&=&f( m_0 11),
\end{eqnarray}
where $m_0\in\{0,1\}^{3}$ and $j\in\{0,1,2,3\}$. Obviously, $f_{{0},0}(m_0)=f_{{0},2}(m_0)=f_{{0},3}(m_0)\equiv0$, and
\begin{eqnarray}
f_{{0},1}(m_0)=
\begin{cases}
1,m_0=010,\\
0,m_0\neq 010,
\end{cases}
\end{eqnarray}
where $m_0\in\{0,1\}^3$ and $010$ is the target substring.

Afterwards, we generate a new function $g_0: \{0,1\}^3 \rightarrow \{0,1\}$ in terms of above four subfunctions,
\begin{eqnarray}
g_{0}(m_0)=\text{OR} \left(f_{{0},0}(m_0), f_{{0},1}(m_0), f_{{0},2}(m_0), f_{{0},3}(m_0)\right)=
\begin{cases}
1,m_0=010,\\
0,m_0\neq 010,
\end{cases}
\end{eqnarray}
where $m_0\in\{0,1\}^{3}$.

Similarly, we choose first $3$ bits in $x$ to divide $f$, then we can get $2^3=8$ subfunctions $f_{{1},j}: \{0,1\}^2 \rightarrow \{0,1\}$:
\begin{eqnarray}
f_{{1},0}(m_1)&=f( 000 m_1 ), f_{{1},4}(m_1)&=f( 100 m_1 ),\\
f_{{1},1}(m_1)&=f( 001 m_1 ), f_{{1},5}(m_1)&=f( 101 m_1 ),\\
f_{{1},2}(m_1)&=f( 010 m_1 ), f_{{1},6}(m_1)&=f( 110 m_1 ),\\
f_{{1},3}(m_1)&=f( 011 m_1 ), f_{{1},7}(m_1)&=f( 111 m_1 ),
\end{eqnarray}
where $m_1\in\{0,1\}^{2}$ and $j\in\{0,1, \cdots,7\}$. Obviously, $f_{{1},0}(m_1)=f_{{1},1}(m_1)=f_{{1},3}(m_1)=f_{{1},4}(m_1)=f_{{1},5}(m_1)=f_{{1},6}(m_1)=f_{{1},7}(m_1)\equiv0$, and
\begin{eqnarray}
f_{{1},2}(m_1)=
\begin{cases}
1,m_1=01,\\
0,m_1\neq 01,
\end{cases}
\end{eqnarray}
where $m_1\in\{0,1\}^2$ and $01$ is the target substring.

Afterwards, we generate a new function $g_1: \{0,1\}^2 \rightarrow \{0,1\}$ in terms of above eight subfunctions,
\begin{eqnarray}
g_{1}(m_1)=\text{OR} \left(f_{{1},0}(m_1), f_{{1},1}(m_1), \cdots , f_{{1},7}(m_1)\right)=
\begin{cases}
1,m_1=01,\\
0,m_1\neq 01,
\end{cases}
\end{eqnarray}
where $m_1\in\{0,1\}^{2}$.

Next, we can build the completed circuit (see Figure \ref{5bitdega}). By sampling the circuit 10,000 times, the results can be found in Figure \ref{measure01001}. The total number of quantum gates is 53 and the circuit depth is 17.



\begin{figure}[H]
\centering
\includegraphics[width=5in]{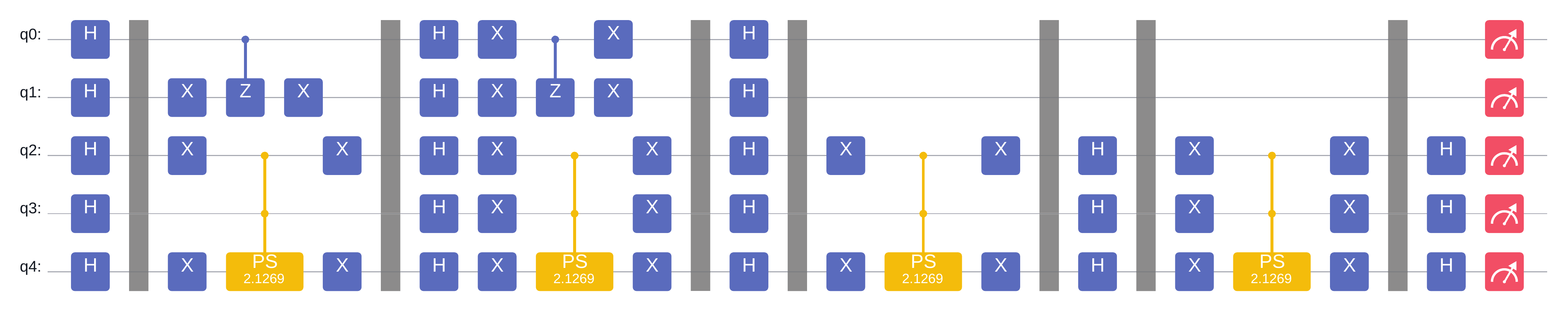}
\caption{\label{5bitdega} The quantum circuit of the 5-qubit DEGA (target string $\tau = 01001$). The parameter in the circuit is $\phi=2\arcsin\left(\sin\left(\frac{\pi}{4J+6}\right) / \sin \theta\right) =2.1268800471555034\approx 2.1269$, where $J=\lfloor(\pi/2-\theta)/(2\theta)\rfloor$ and $\theta= \arcsin {\sqrt{1/2^3}}$.}
\centering
\includegraphics[width=3in]{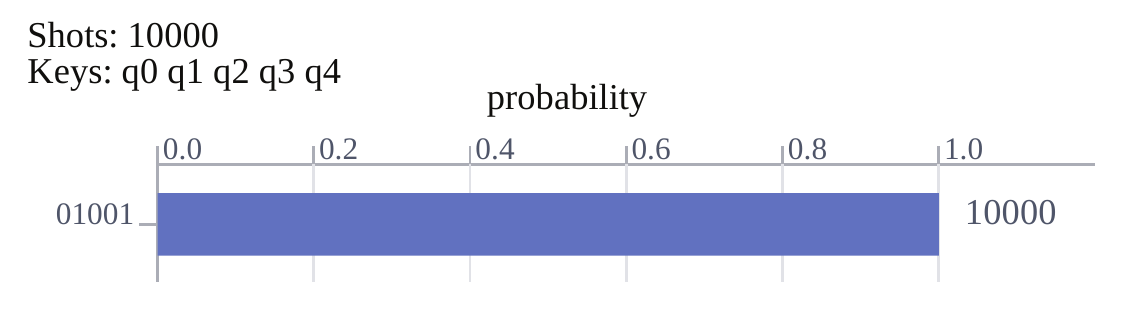}
\caption{\label{measure01001} The sampling results of the 5-qubit DEGA (target string $\tau = 01001$).}
\end{figure}

The sampling results demonstrate that $\tau=01001$ is obtained exactly after measurement. 

If  Grover's algorithm or the algorithm by Long is employed to search $\tau = 01001$, we can build the completed circuits (see Figure \ref{grover01001} and Figure \ref{long01001}) respectively. 
By sampling the circuits 10,000 times respectively, the results can be found in Figure \ref{grovermeasure01001} and Figure \ref{longmeasure01001}.

\begin{figure}[H]
\centering
\includegraphics[width=6in]{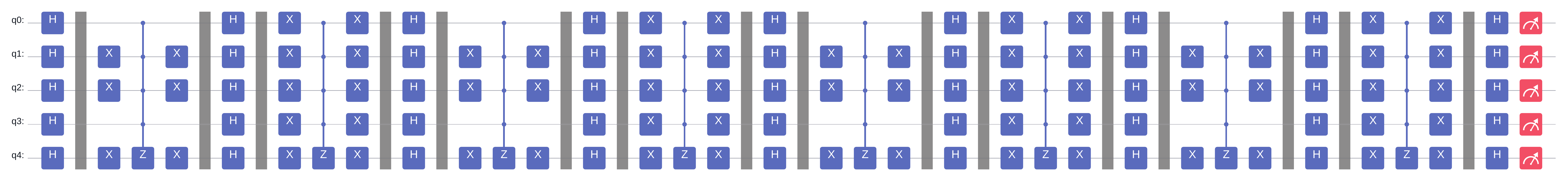}
\caption{\label{grover01001} The quantum circuit of the 5-qubit Grover's algorithm (target string $\tau = 01001$).}
\end{figure}

\begin{figure}[H]
\centering
\includegraphics[width=6in]{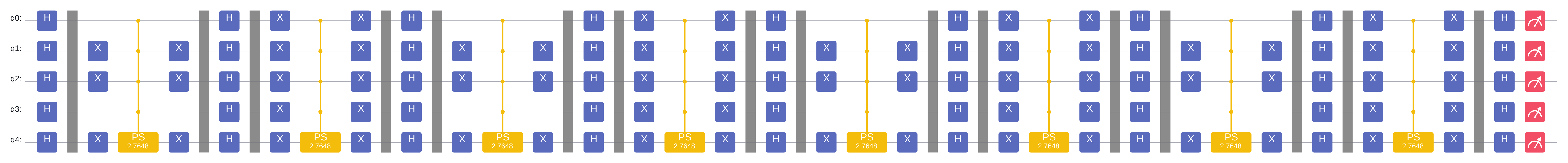}
\caption{\label{long01001} The quantum circuit of the 5-qubit algorithm by Long (target string $\tau = 01001$). The parameter in the circuit is $\phi=2\arcsin\left(\sin\left(\frac{\pi}{4J+6}\right) / \sin \theta\right) =2.764763603060391\approx 2.7648$, where $J=\lfloor(\pi/2-\theta)/(2\theta)\rfloor$ and $\theta= \arcsin {\sqrt{1/2^5}}$.}
\end{figure}

\begin{figure}[H]
\centering
\begin{minipage}{0.49\textwidth}
\centering
\includegraphics[width=2.6in]{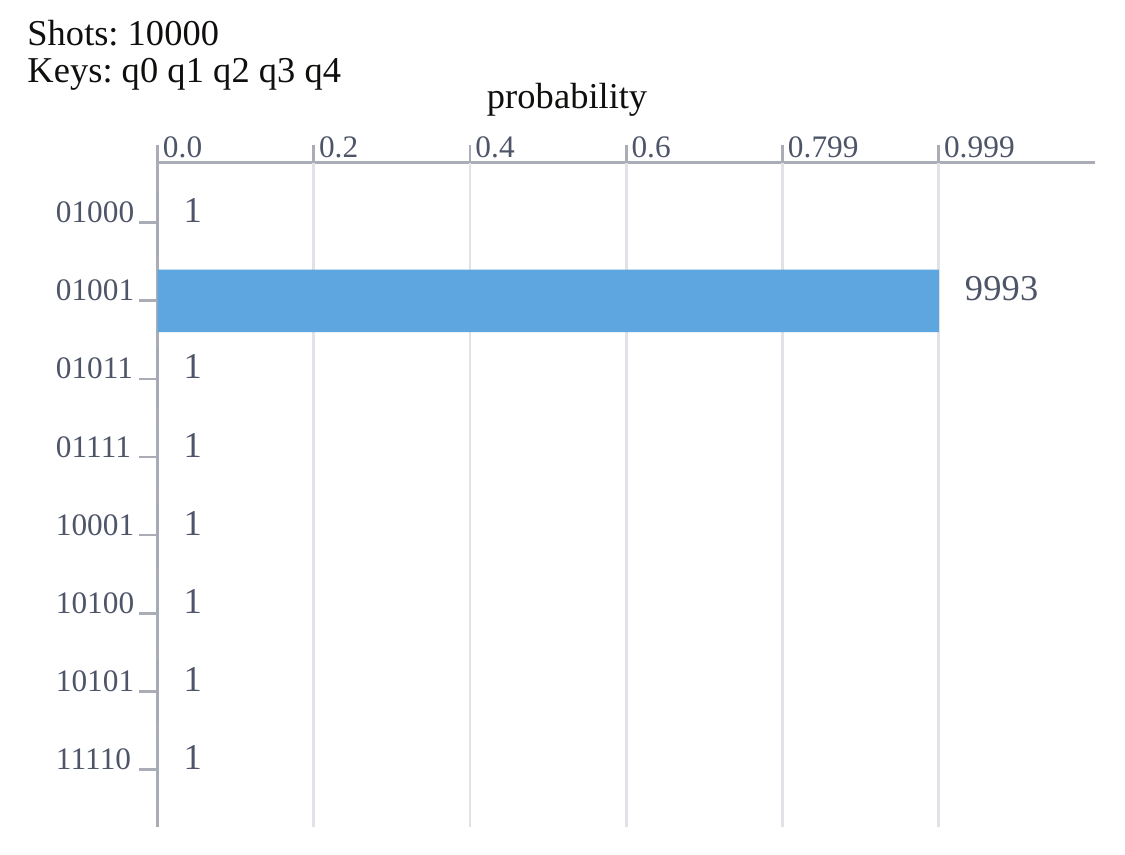}
\caption{\label{grovermeasure01001} The sampling results of the 5-qubit Grover's algorithm (target string $\tau = 01001$).}
\end{minipage}
\begin{minipage}{0.49\textwidth}
\centering
\includegraphics[width=2.6in]{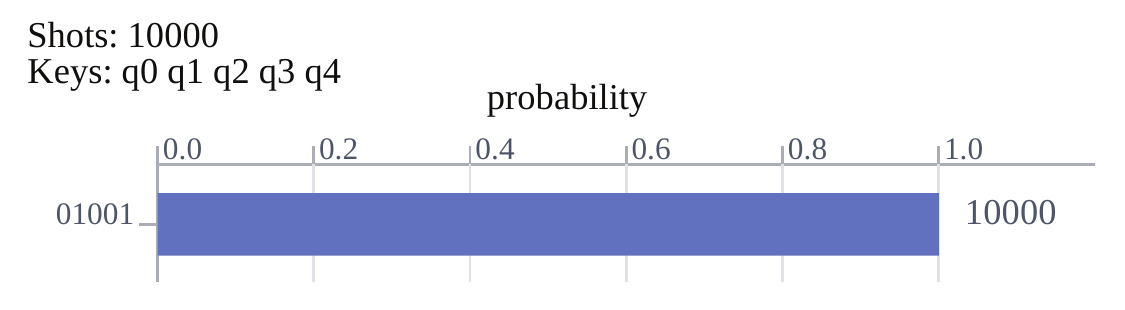}
\caption{\label{longmeasure01001} The sampling results of the 5-qubit algorithm by Long (target string $\tau = 01001$).}
\end{minipage}
\end{figure}

The number of quantum gates required by two algorithms is 117, and the circuit depths are 33.

The algorithm by Long and our DEGA can achieve exact search, while  Grover's algorithm can obtain target strings with high probability (see Figure \ref{probability}). 

\begin{figure}[H]
\centering
\includegraphics[width=3in]{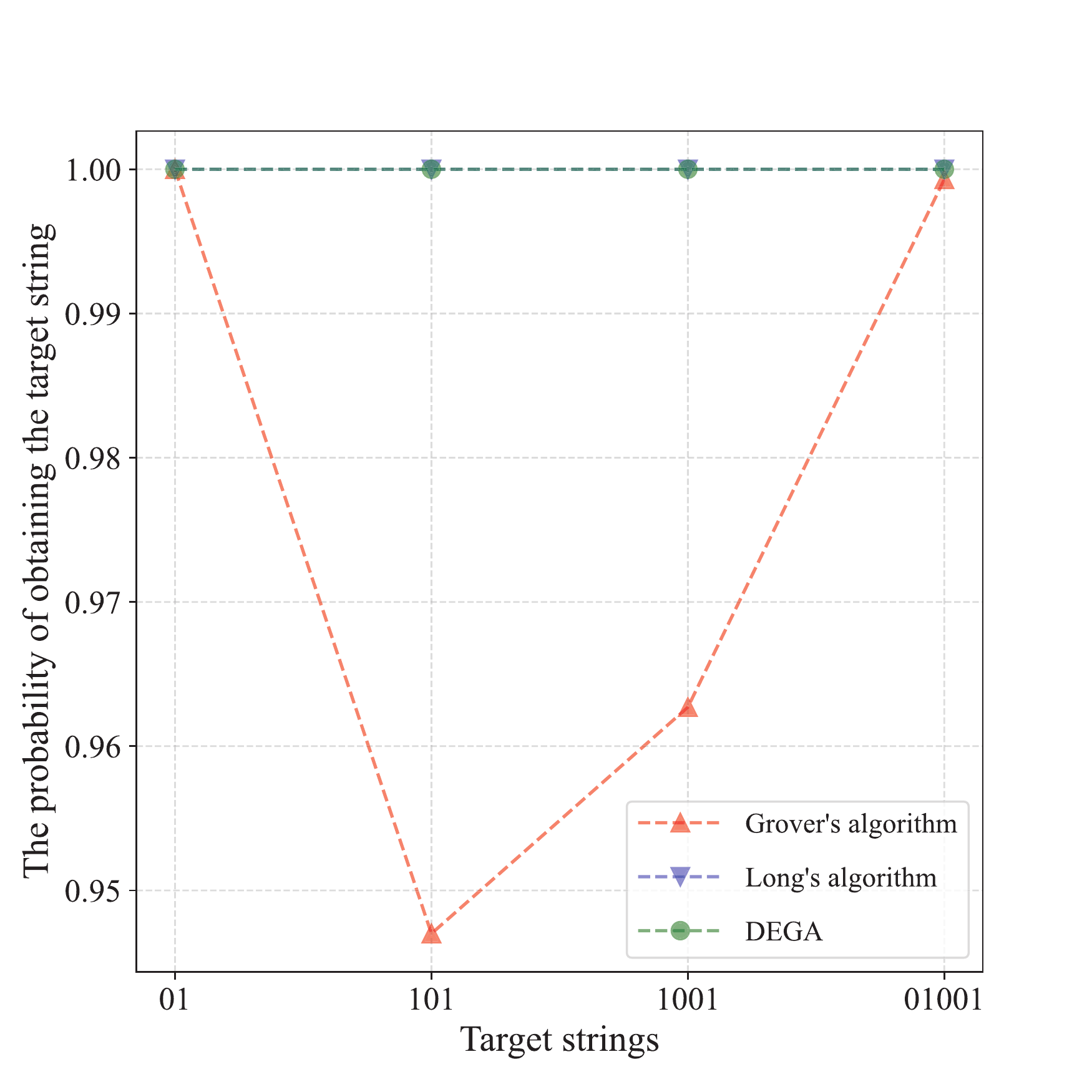}
\caption{\label{probability} The probability of obtaining the target string by different algorithms.}
\end{figure}

In addition, the number of quantum gates and circuit depth required by Grover's and the algorithm by Long are the same, which will be deepened as $n$ increases. However, our DEGA requires fewer quantum gates (see Figure \ref{quantumgatesnumber}) and shallower circuit depth (Figure \ref{circuitdepth}). The circuit depth of our DEGA only depends on the parity of $n$, and it is not deepened as $n$ increases.


\begin{figure}[H]
\centering
\begin{minipage}{0.49\textwidth}
\centering
\includegraphics[width=3in]{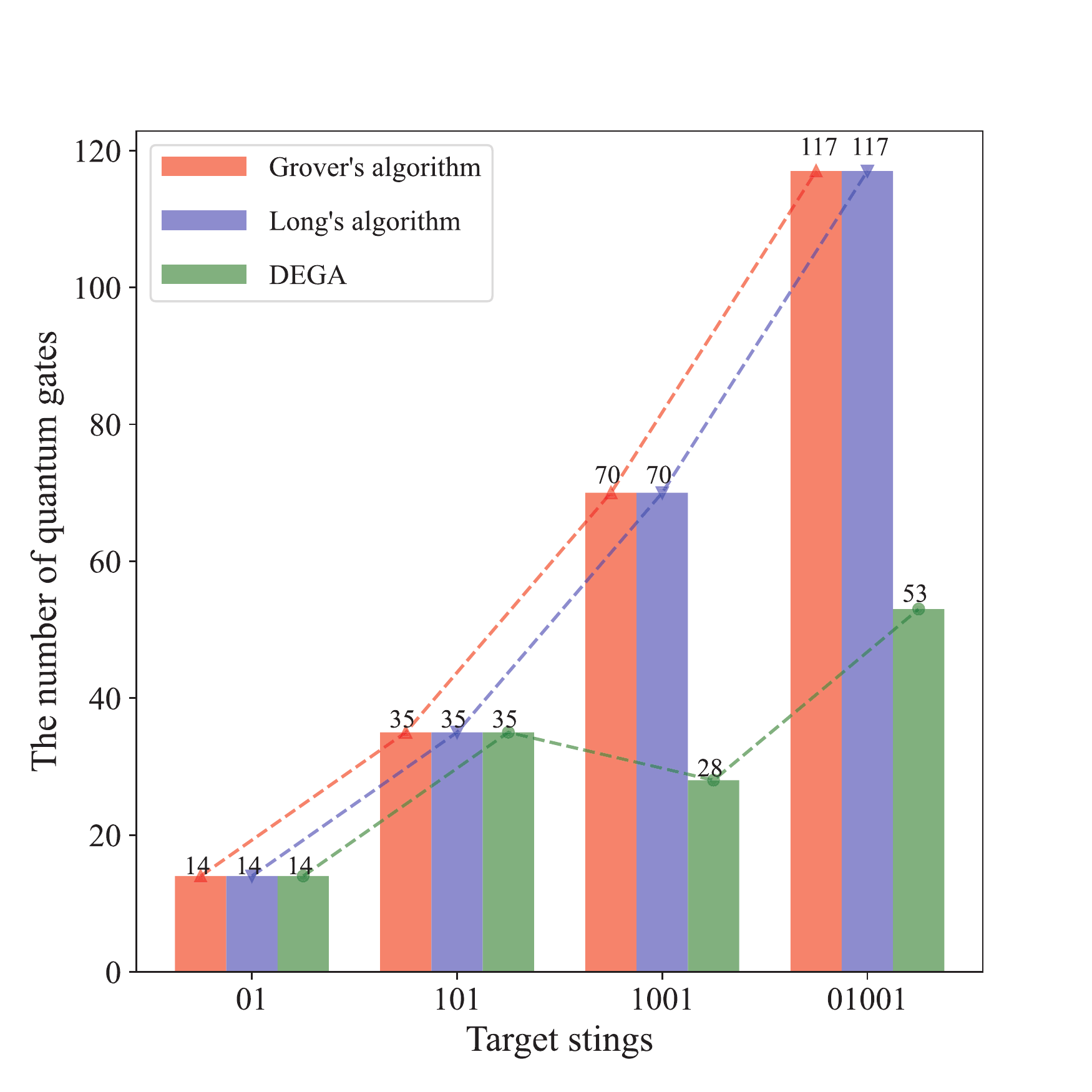}
\caption{\label{quantumgatesnumber} The number of quantum gates required by different algorithms..}
\end{minipage}
\begin{minipage}{0.49\textwidth}
\centering
\includegraphics[width=3in]{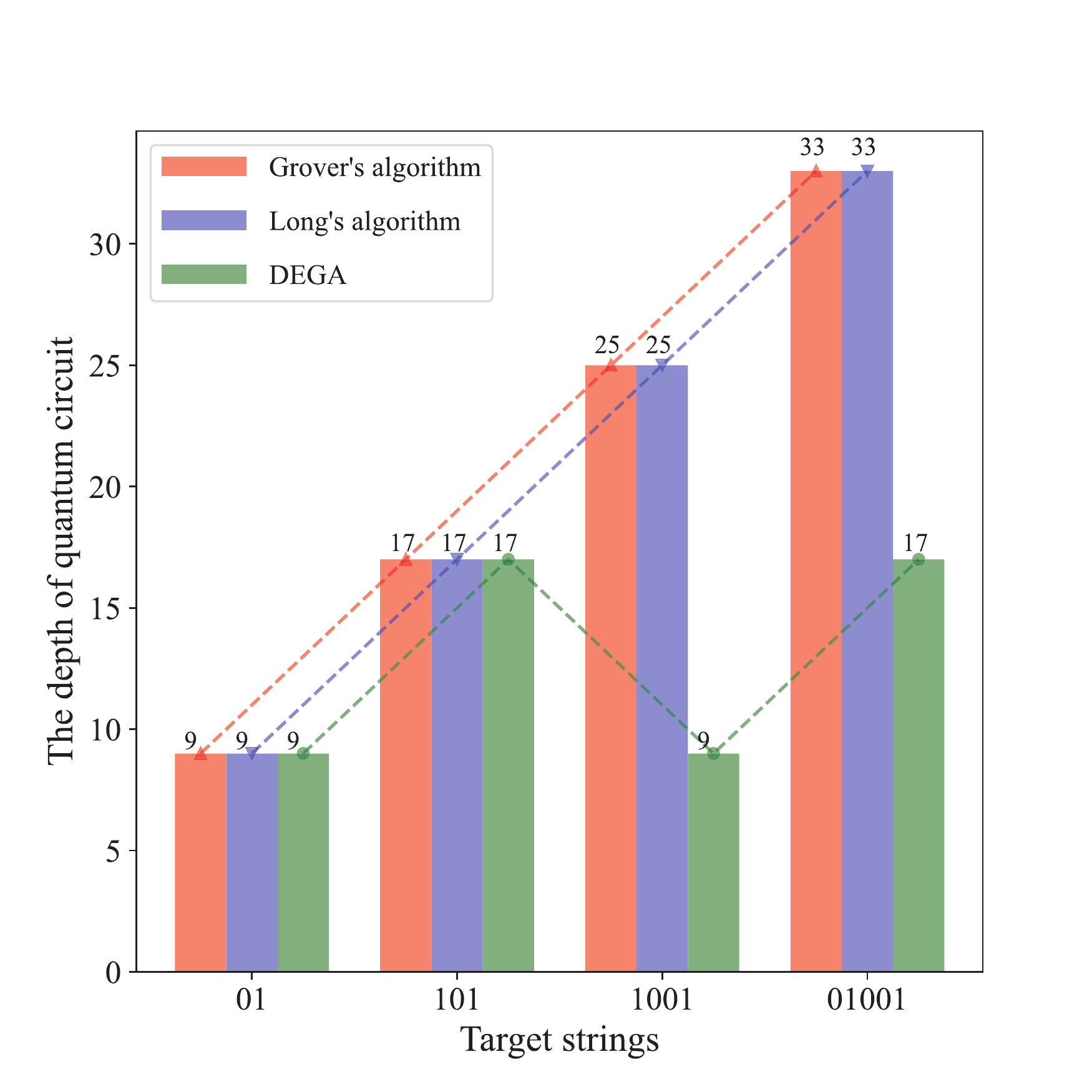}
\caption{\label{circuitdepth} The depth of quantum circuit of different algorithms.}
\end{minipage}
\end{figure}

\subsection{The depolarization channel}\label{channel}
In the above experiment, the quantum circuits we built ignore the influence of noise. However, to run quantum algorithms on real quantum computers, there is a certain error in the operation of each type of quantum gate. In this subsection, we will simulate the 6-qubit DBVA and 5-qubit DEGA running in the depolarization channel that is a crucial type of quantum noise, it further illustrates that distributed quantum algorithm has the superiority of resisting noise.

In the depolarization channel \cite{RefNielsen2002}, a single qubit will be depolarized with probability $p$. which means it is replaced by complete mixed state $I/2$. Also, it will remain unchanged with the probability of $1-p$. After this noise, the state of the quantum system  becomes
\begin{eqnarray}\label{noise2}
\varepsilon(\rho) = (1-p)\rho + p\left(\frac{I}{2}\right),
\end{eqnarray}
where $\rho$ represents the initial quantum state. For arbitrary state $\rho$, we have
\begin{eqnarray}\label{I2}
\frac{I}{2} = \frac{\rho+X\rho X+ Y\rho Y+ Z\rho Z}{4},
\end{eqnarray}
where $X,Y$ and $Z$ are Pauli operators. Inserting Eq.~\eqref{I2} into Eq.~\eqref{noise2} gives
\begin{eqnarray}\label{noise2com}
\varepsilon(\rho) = \left(1-\frac{3p}{4}\right)\rho + \frac{p}{4}\left( X\rho X+ Y\rho Y+ Z\rho Z \right).
\end{eqnarray}
For convenience, the depolarization channel is sometimes expressed as
\begin{eqnarray}\label{noise2com2}
\varepsilon(\rho) = \left(1-p\right)\rho + \frac{p}{3}\left( X\rho X+ Y\rho Y+ Z\rho Z \right),
\end{eqnarray}
which means the state $\rho$ will remain unchanged with the probability $1-p$, and apply the Pauli operators $X,Y$ or $Z$ with probability $p/3$ for each operators.

Furthermore, the depolarizing channel can be extended to a $d$-dimensional ($d > 2$) quantum system. Similarly, the depolarizing channel of a $d$-dimensional quantum system again replaces the quantum system with tcomplete mixed state $I/d$ with probability $p$, and remains unchanged with the probability of $1-p$. The corresponding quantum operation is
\begin{eqnarray}\label{noised}
\varepsilon(\rho) = (1-p)\rho + p\left(\frac{I}{d}\right).
\end{eqnarray}

\subsubsection{6-qubit DBVA, the hidden string $s=001011$}
In this subsection, we simulate the quantum algorithms running in the depolarization channel by adding noise behind each quantum gate. That is, we build the quantum circuits of Figure \ref{fullcir}, Figure \ref{fullciropti}, Figure \ref{dbva2ciropti} and Figure \ref{dbva3ciropti} in the depolarization channel as shown in Figure \ref{cirbvnoisecom} to Figure \ref{cirdbva3noise}. 
Set the noise parameter $p=0.03$.

Sample the above four circuits in the depolarizing channel 10,000 times, respectively. In order to better reproduce the experimental results, we set the random seeds of simulator and sampling both being 42, the noise parameter $p$ being $0.03$. At last, the results can be found in Figure \ref{cirbvnoiseoutput} to Figure \ref{cirdbva2noiseoutput}, respectively.

\begin{figure}[H]
\centering
\begin{minipage}{1\textwidth}
\centering
\includegraphics[width=1\textwidth]{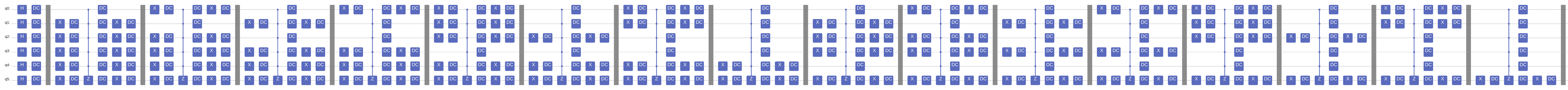}
\end{minipage}
\\
\begin{minipage}{1\textwidth}
\centering
\includegraphics[width=1\textwidth]{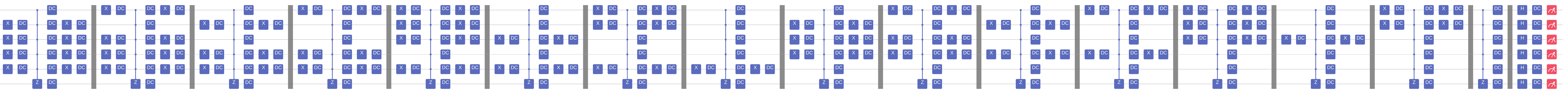}
\caption{\label{cirbvnoisecom} The quantum circuit of the BV algorithm in the depolarizing channel. Since circuit is too long, divide the whole quantum circuit into upper and lower parts. A gray rectangle represents a barrier gate, which is used to separate quantum gates. And DC represents the depolarization channel.}
\end{minipage}
\end{figure}

\begin{figure}[H]
\centering
\includegraphics[width=1\textwidth]{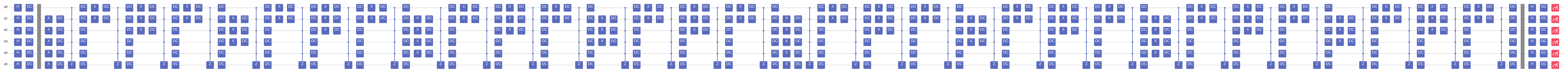}
\caption{\label{cirbvoptinoise.eps} The quantum circuit of the BV algorithm after optimization in the depolarizing channel.}
\end{figure}

\begin{figure}[H]
\centering
\includegraphics[width=0.9\textwidth]{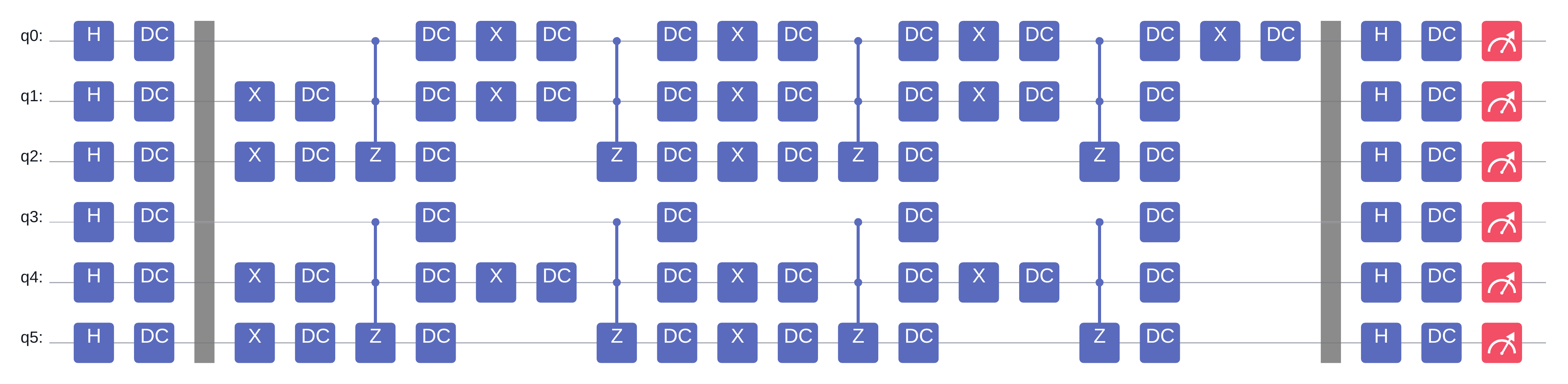}
\caption{\label{cirdbva2noise} The quantum circuit of DBVA with two computing nodes after optimization in the depolarizing channel.}
\end{figure}
\begin{figure}[H]
\centering
\includegraphics[width=0.8\textwidth]{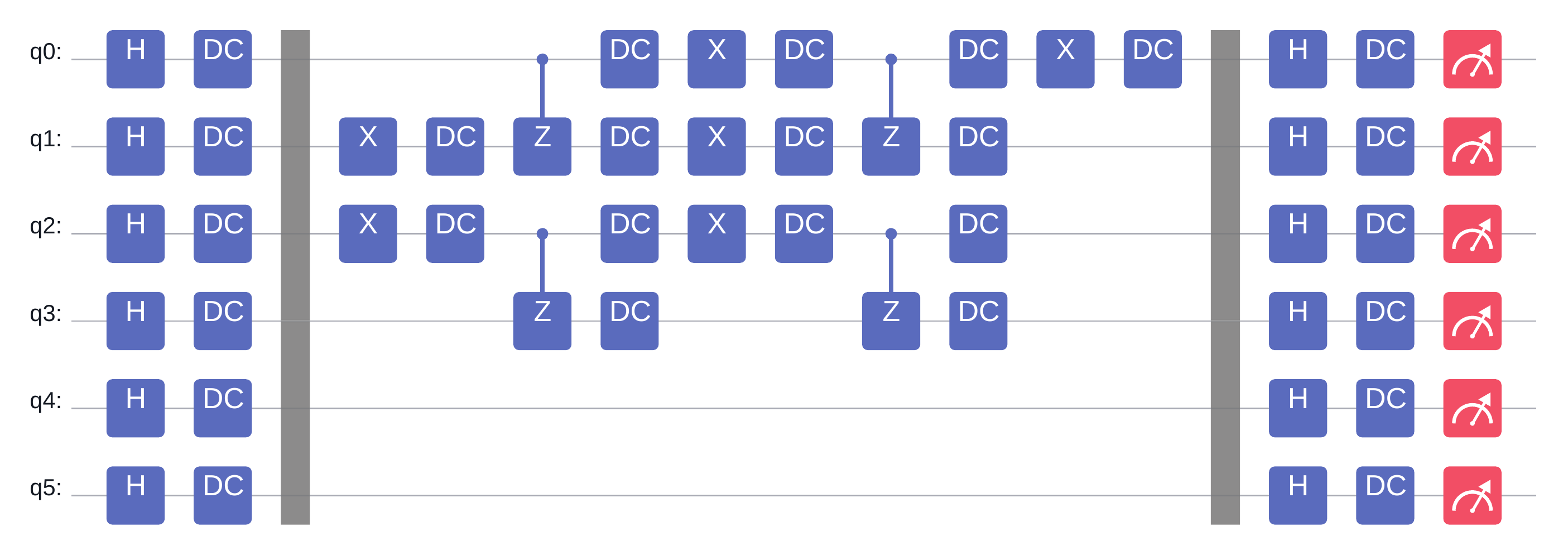}
\caption{\label{cirdbva3noise} The quantum circuit of DBVA with $t=3$ computing nodes after optimization in the depolarizing channel.}
\end{figure}



\begin{figure}[h]
\centering
\begin{minipage}{0.24\textwidth}
\centering
\includegraphics[width=1in]{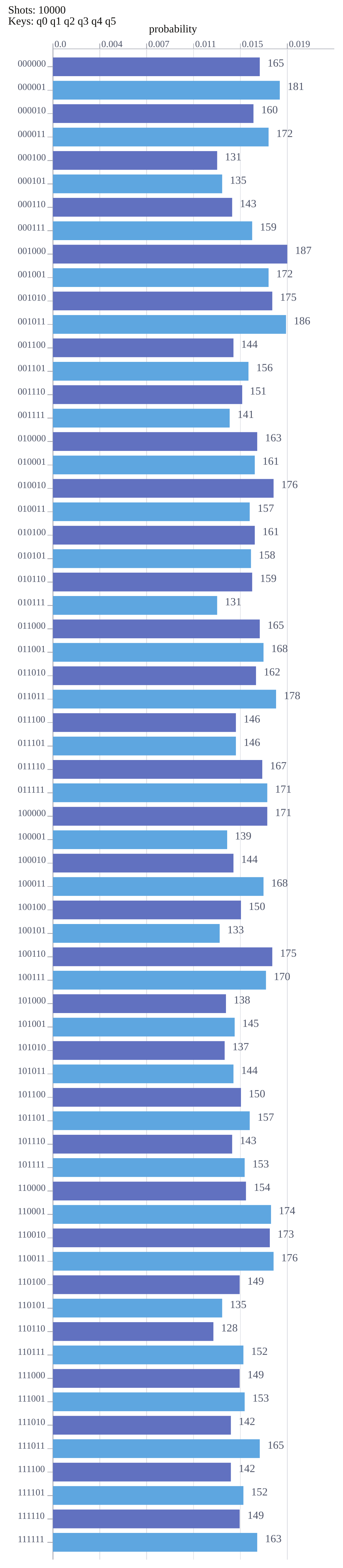}
\caption{\label{cirbvnoiseoutput}The sampling results of the BV algorithm in the depolarizing channel. The noise parameter $p=0.03$.}
\end{minipage}
\begin{minipage}{0.24\textwidth}
\centering
\includegraphics[width=1.02in]{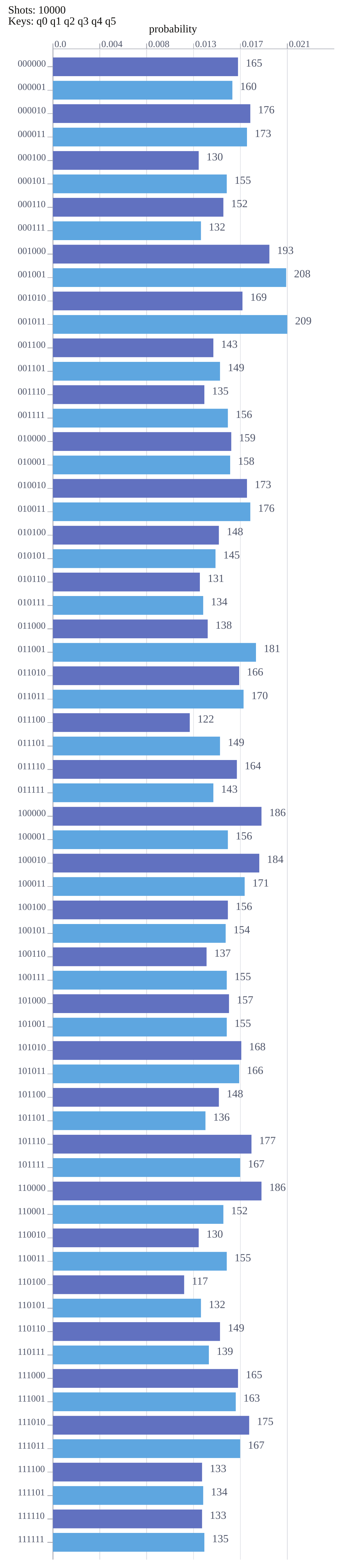}
\caption{\label{cirbvoptinoiseoutput}The sampling results of the BV algorithm after optimization in the depolarizing channel. The noise parameter $p=0.03$.}
\end{minipage}
\begin{minipage}{0.24\textwidth}
\centering
\includegraphics[width=1.05in]{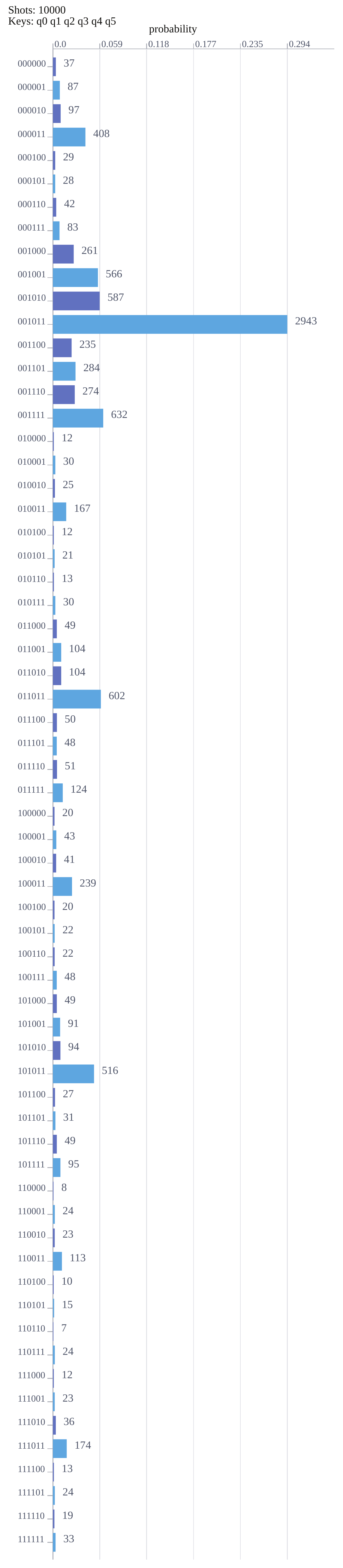}
\caption{\label{cirdbva2noiseoutput}The sampling results of DBVA with two computing nodes after optimization in the depolarizing channel. The noise parameter $p=0.03$.}
\end{minipage}
\begin{minipage}{0.24\textwidth}
\centering
\includegraphics[width=1.38in]{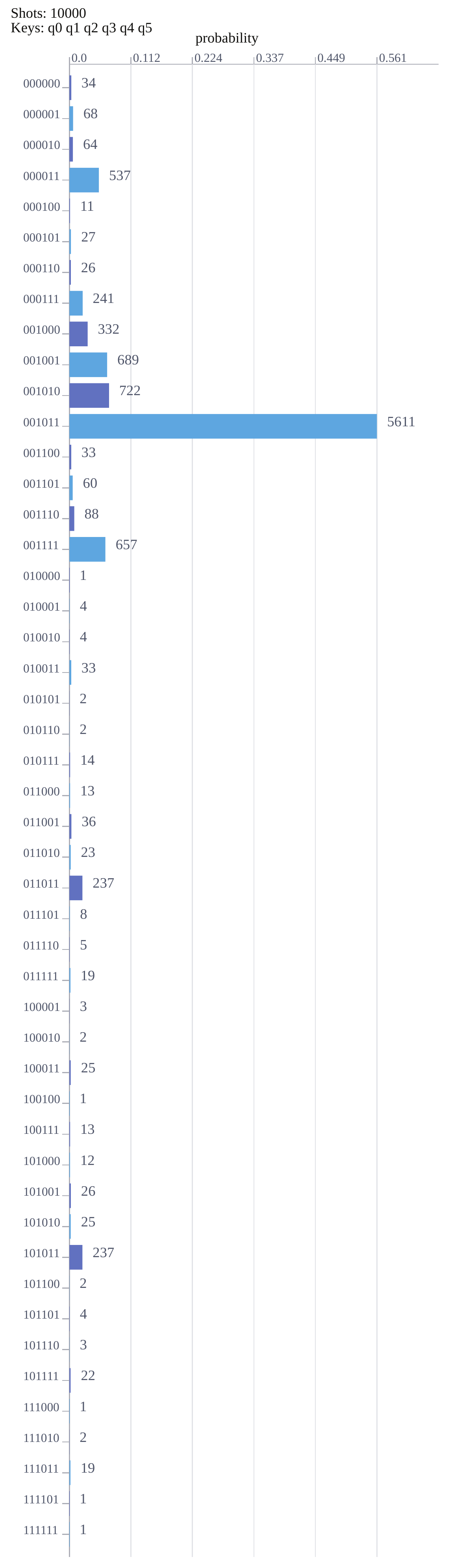}
\caption{\label{cirdbva2noiseoutput}The sampling results of DBVA with $t=3$ computing nodes after optimization in the depolarizing channel. The noise parameter $p=0.03$.}
\end{minipage}
\end{figure}

The probability of obtaining the hidden string $s=001011$ are 0.0186, 0.0209, 0.2943 and 0.5611, respectively. 
However, the desired results were drown out by the channel noise (see Figure \ref{cirbvnoiseoutput} and Figure \ref{cirbvoptinoiseoutput}). Furthermore, the probability in our DBVA are higher than in the optimized BV algorithm (0.2943 $>$ 0.0209, 0.5611 $>$ 0.0209). 
When $t$ is greater, the probability is higher (0.5611 $>$ 0.2943).

We already know that there is a certain error in the operation of each type of quantum gate in the depolarizing channel. In other words, the more quantum gates are, the greater the cumulative noise of the quantum circuit is, and the lower the probability of obtaining the target string is. Here, the number of quantum gates required by our DBVA are obviously less than that required by the BV algorithm. 
At present, to run quantum algorithms on real quantum computers, it is usually necessary to decompose the multi-qubit quantum gates into single-qubit gates and double-qubit gates. If we further consider the decomposition of $C^{n-1}Z$ gates into single-qubit gates and double-qubit gates, the difference between the number of quantum gates required by the BV algorithm and our DBVA will be larger, and the fidelity of theoretical and experimental final states will be lower. However, we will not go into too much details here.

As the number of qubits increases, the depth of quantum circuit or noise parameter $p$ increases, the errors caused by the noise will further accumulate. 
Fortunately, our DBVA needs fewer quantum gates, resulting in shallower quantum circuits. Hence, our DBVA has higher noise resistance.

\subsubsection{5-qubit DEGA, the target string $\tau = 01001$}
We can rebuild the quantum circuits of the Grover's algorithm, the algorithm by Long and our DEGA in the depolarization channel (see Figure \ref{grover01001noise} to Figure \ref{5bitdeganoise}). 
Set the noise parameter $p=0.01$ and DC represents the depolarization channel.

\begin{figure}[H]
\centering
\includegraphics[width=6.2in]{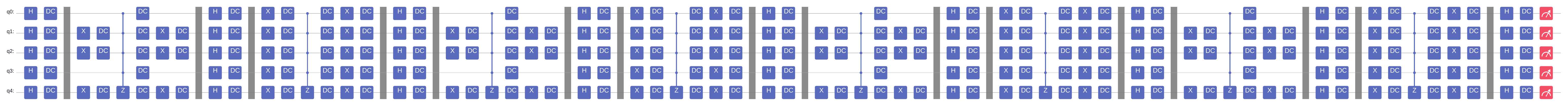}
\caption{\label{grover01001noise} The quantum circuit of the 5-qubit Grover's algorithm (target string $\tau = 01001$) in the depolarizing channel.}
\end{figure}

\begin{figure}[H]
\centering
\includegraphics[width=6.2in]{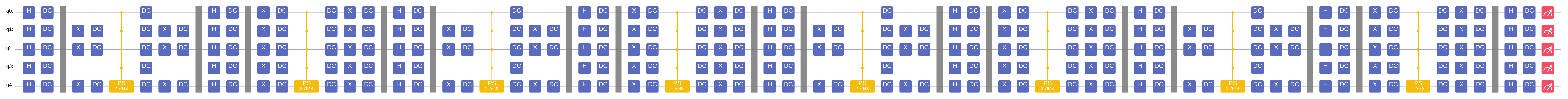}
\caption{\label{long01001noise} The quantum circuit of the 5-qubit algorithm by Long (target string $\tau = 01001$) in the depolarizing channel. The parameter in the circuit is $\phi=2\arcsin\left(\sin\left(\frac{\pi}{4J+6}\right) / \sin \theta\right) =2.764763603060391\approx 2.7648$, where $J=\lfloor(\pi/2-\theta)/(2\theta)\rfloor$ and $\theta= \arcsin {\sqrt{1/2^5}}$.}
\end{figure}

\begin{figure}[H]
\centering
\includegraphics[width=6in]{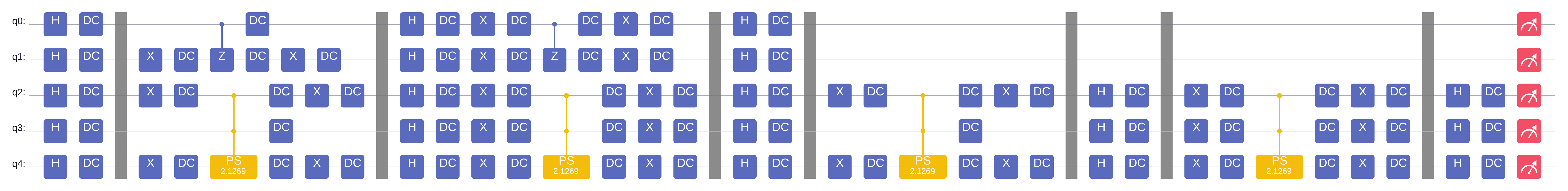}
\caption{\label{5bitdeganoise} The quantum circuit of the 5-qubit DEGA (target string $\tau = 01001$) in the depolarizing channel. The parameter in the circuit is $\phi=2\arcsin\left(\sin\left(\frac{\pi}{4J+6}\right) / \sin \theta\right) =2.1268800471555034\approx 2.1269$, where $J=\lfloor(\pi/2-\theta)/(2\theta)\rfloor$ and $\theta= \arcsin {\sqrt{1/2^3}}$.}
\end{figure}

Sample the above three circuits 10,000 times, respectively. In order to better reproduce the experimental results, we set the random seeds of simulator and sampling both being 42. At last, the sampling results can be found in Figure \ref{measure01001grovernoise} to Figure \ref{measure01001noise}, respectively. It can be seen that due to the existence of noise, it will not achieve exact search.

\begin{figure}[H]
\centering
\begin{minipage}{0.32\textwidth}
\centering
\includegraphics[width=0.7\textwidth]{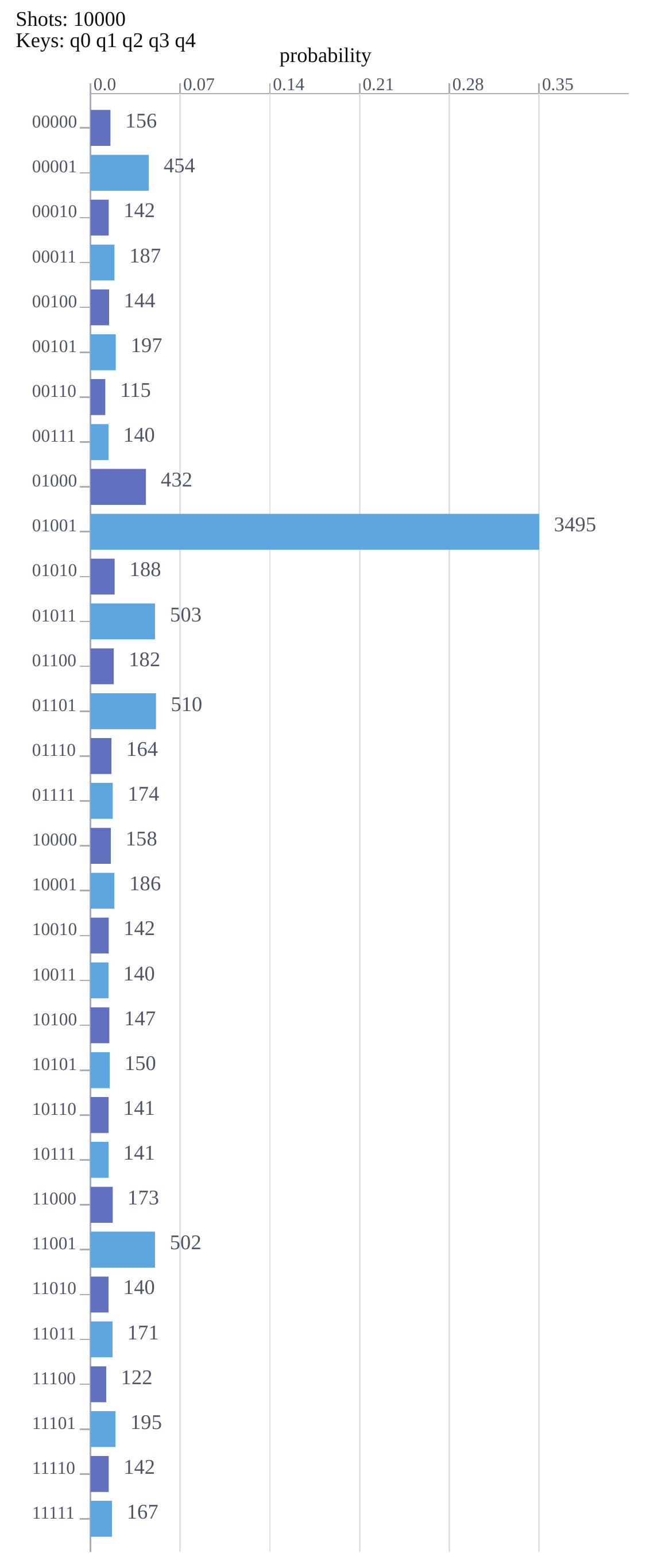}
\caption{\label{measure01001grovernoise} The sampling results of 5-qubit Grover's algorithm in the depolarizing channel. The noise parameter $p=0.01$.}
\end{minipage}
\begin{minipage}{0.32\textwidth}
\centering
\includegraphics[width=0.7\textwidth]{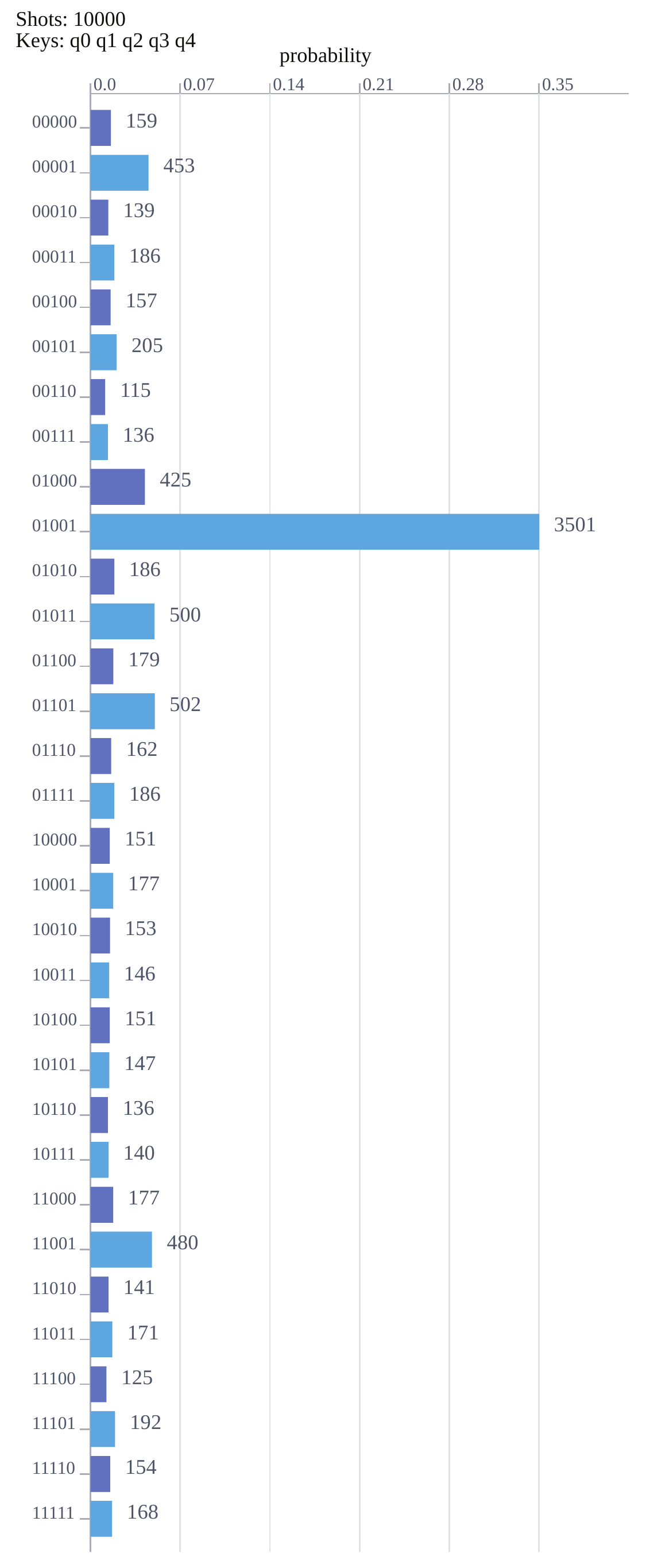}
\caption{\label{measure01001longnoise} The sampling results of 5-qubit the algorithm by Long in the depolarizing channel. The noise parameter $p=0.01$.}
\end{minipage}
\begin{minipage}{0.32\textwidth}
\centering
\includegraphics[width=0.7\textwidth]{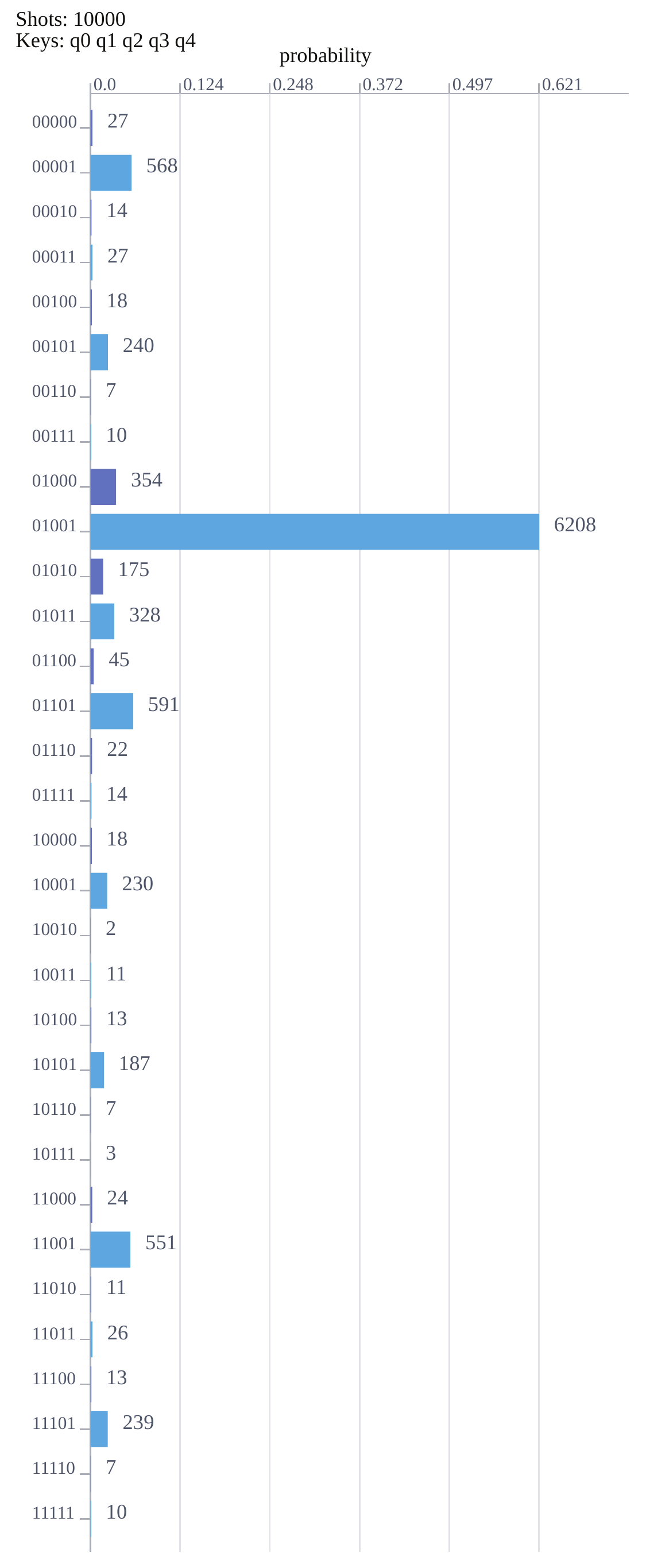}
\caption{\label{measure01001noise} The sampling results of 5-qubit DEGA in the depolarizing channel. The noise parameter $p=0.01$.}
\end{minipage}
\end{figure}

Furthermore, we set the noise parameter $p$ as $0, 0.01,\cdots, 0.09$. The statistical chart as shown in Figure \ref{probability01001}.

\begin{figure}[H]
\centering
\includegraphics[width=3in]{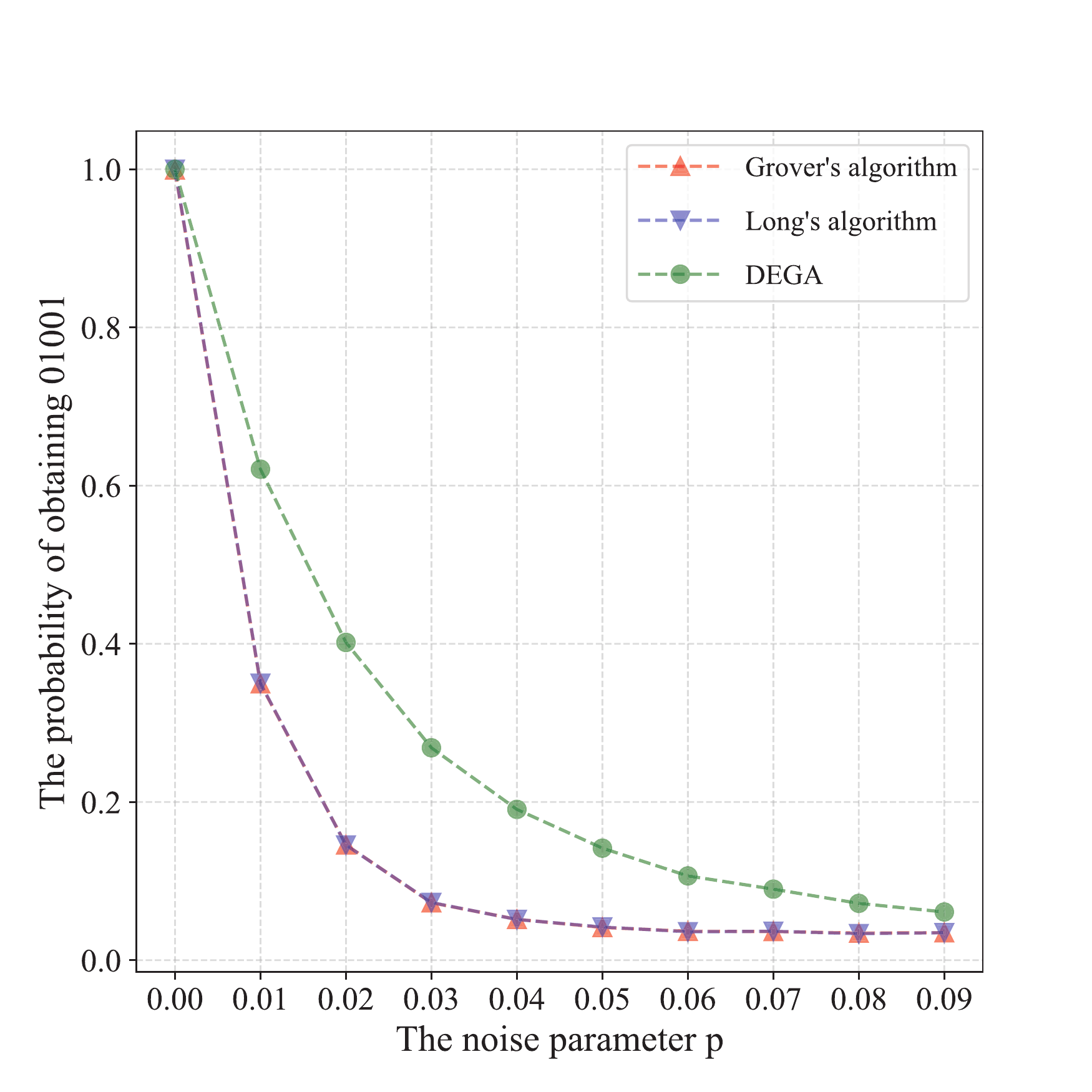}
\caption{\label{probability01001} The probability of obtaining $\tau=01001$ by different algorithms.}
\end{figure}

With the increase of noise parameter $p$, the probability of obtaining $\tau=01001$ decreases. In addition, 
the probabilities of getting target string by Grover's algorithm and the algorithm by Long are close. 
The probability of 
our DEGA is higher than that of Grover's algorithm and the algorithm by Long.

When $p=0.07$, the sampling results can be found in Figure \ref{measure01001grovernoise007} to Figure \ref{measure01001deganoise007}, respectively. 
The desired results are drown out by the channel noise in Grover's algorithm and the algorithm by Long. However, $\tau=01001$ can still be obtained with a higher probability than other states in our DEGA. Even if $p=0.09$, our DEGA still works (see Figure \ref{measure01001grovernoise009} to Figure \ref{measure01001deganoise009}). It can be found that although the number of qubits of our DEGA is the same as that of Grover's algorithm and the algorithm by Long, our circuit depth is shallower, so that our scheme is more resistant to the depolarization channel noise.

In short, by simulating the algorithms running in the depolarized channel, it further illustrates that distributed quantum algorithms have the superiority of resisting noise.

\begin{figure}[H]
\centering
\begin{minipage}{0.32\textwidth}
\centering
\includegraphics[width=1.45in]{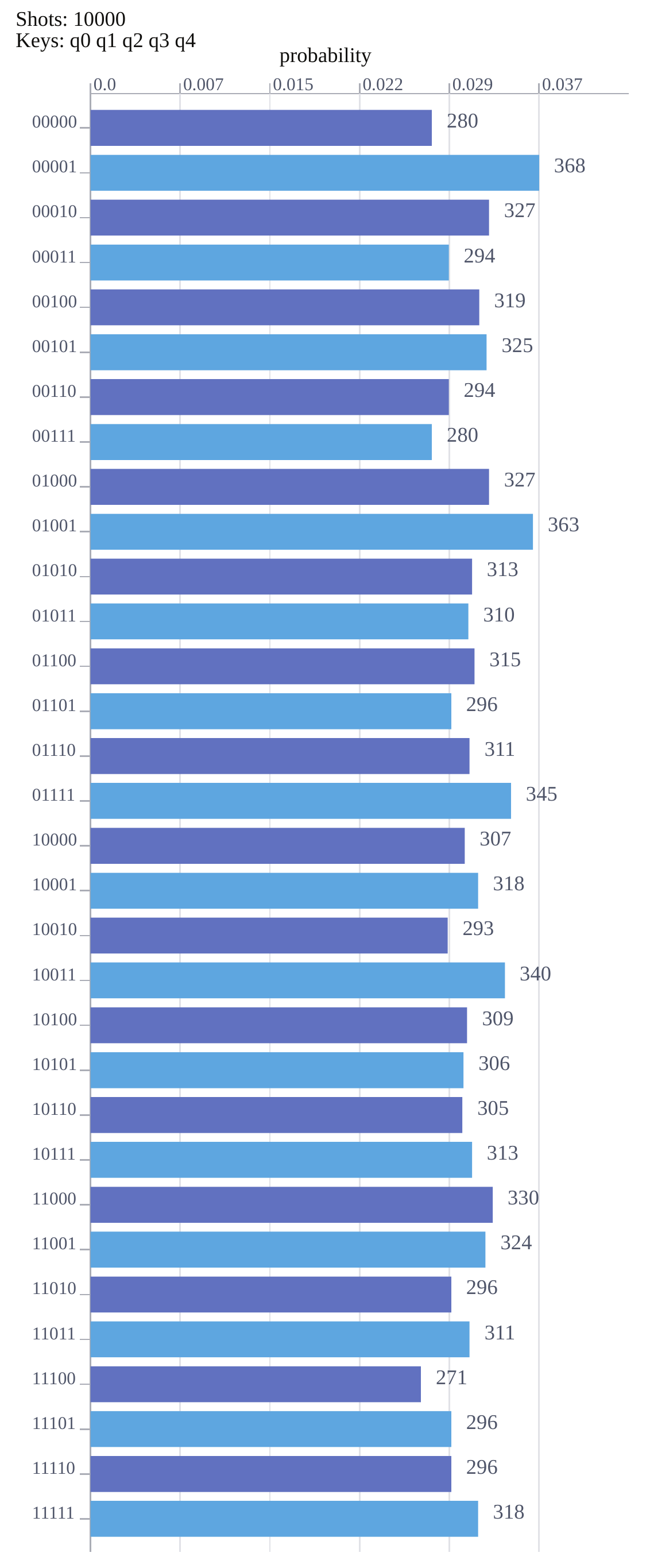}
\caption{\label{measure01001grovernoise007} The sampling results of the 5-qubit Grover's algorithm in the depolarizing channel. The noise parameter $p=0.07$.}
\end{minipage}
\begin{minipage}{0.32\textwidth}
\centering
\includegraphics[width=1.45in]{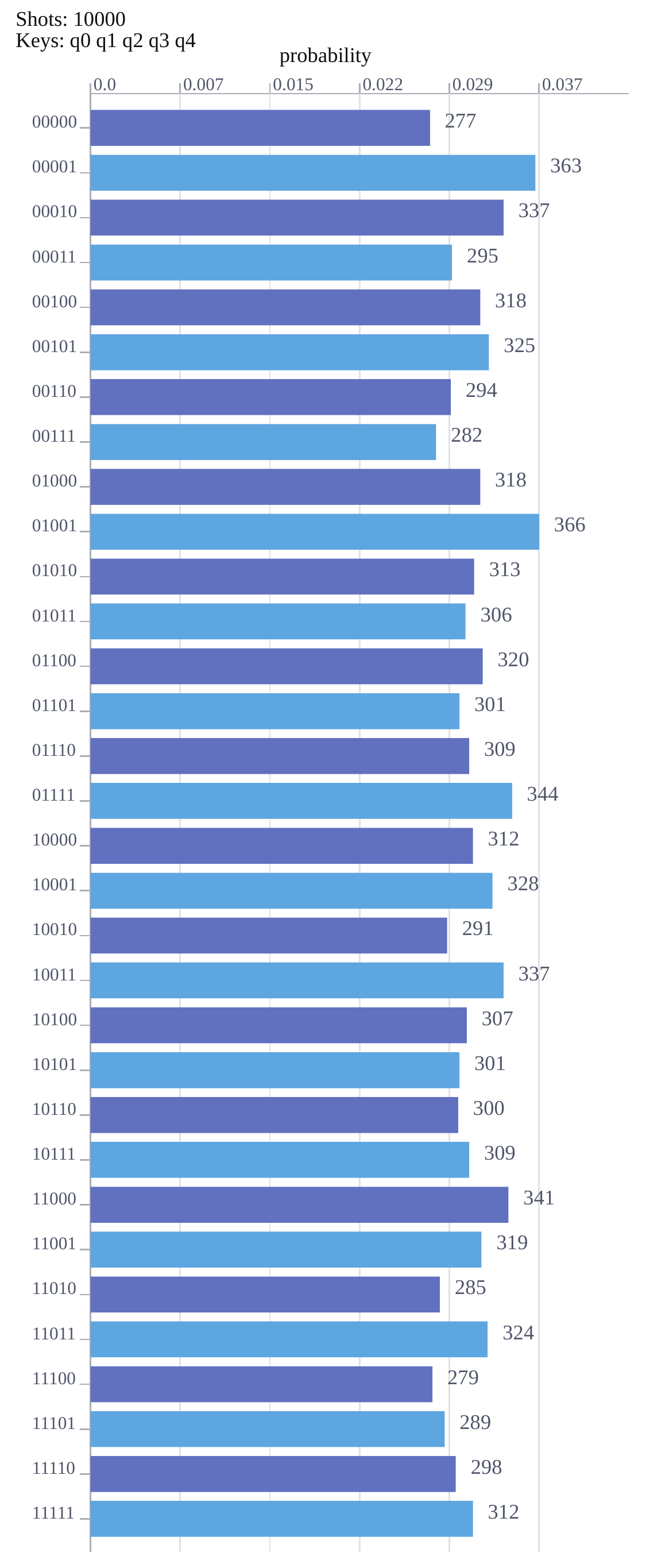}
\caption{\label{measure01001longnoise007} The sampling results of the 5-qubit algorithm by Long in the depolarizing channel. The noise parameter $p=0.07$.}
\end{minipage}
\begin{minipage}{0.32\textwidth}
\centering
\includegraphics[width=1.45in]{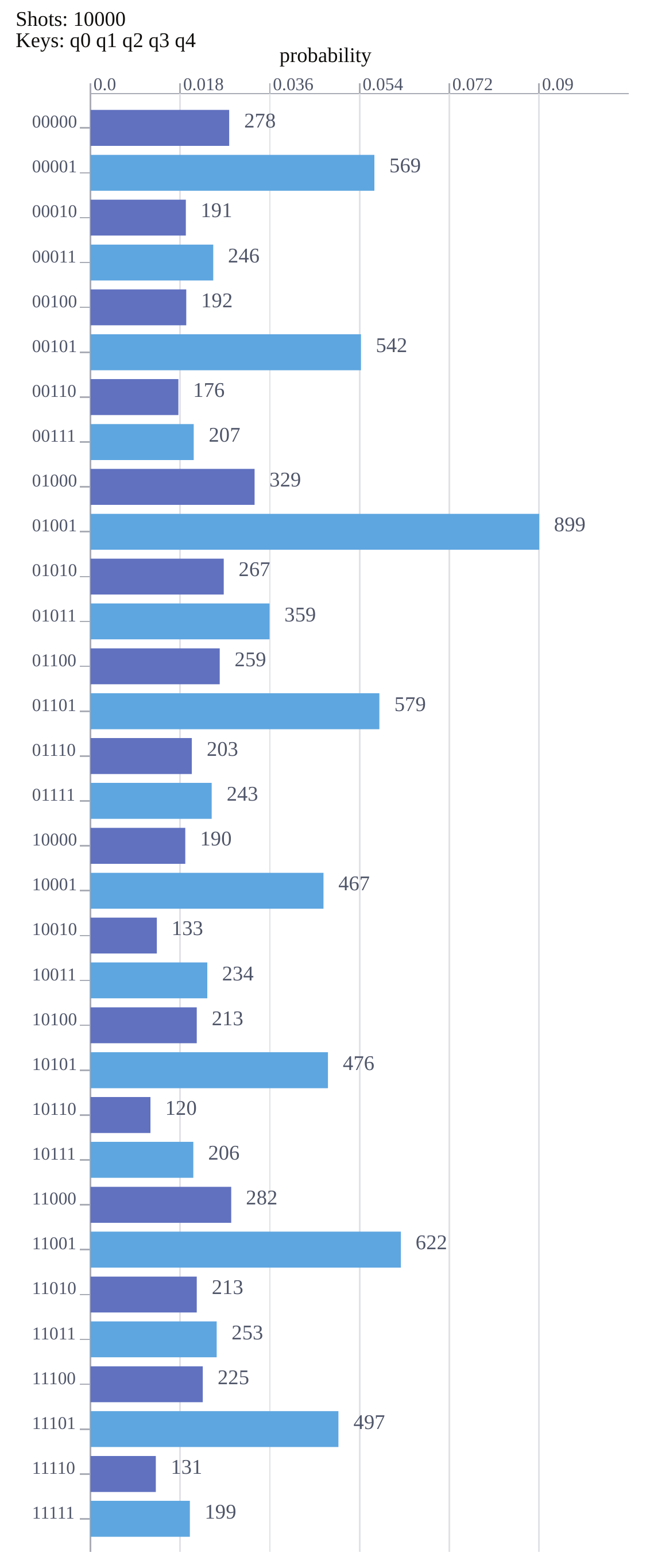}
\caption{\label{measure01001deganoise007} The sampling results of the 5-qubit DEGA in the depolarizing channel. The noise parameter $p=0.07$.}
\end{minipage}
\centering
\begin{minipage}{0.32\textwidth}
\centering
\includegraphics[width=1.45in]{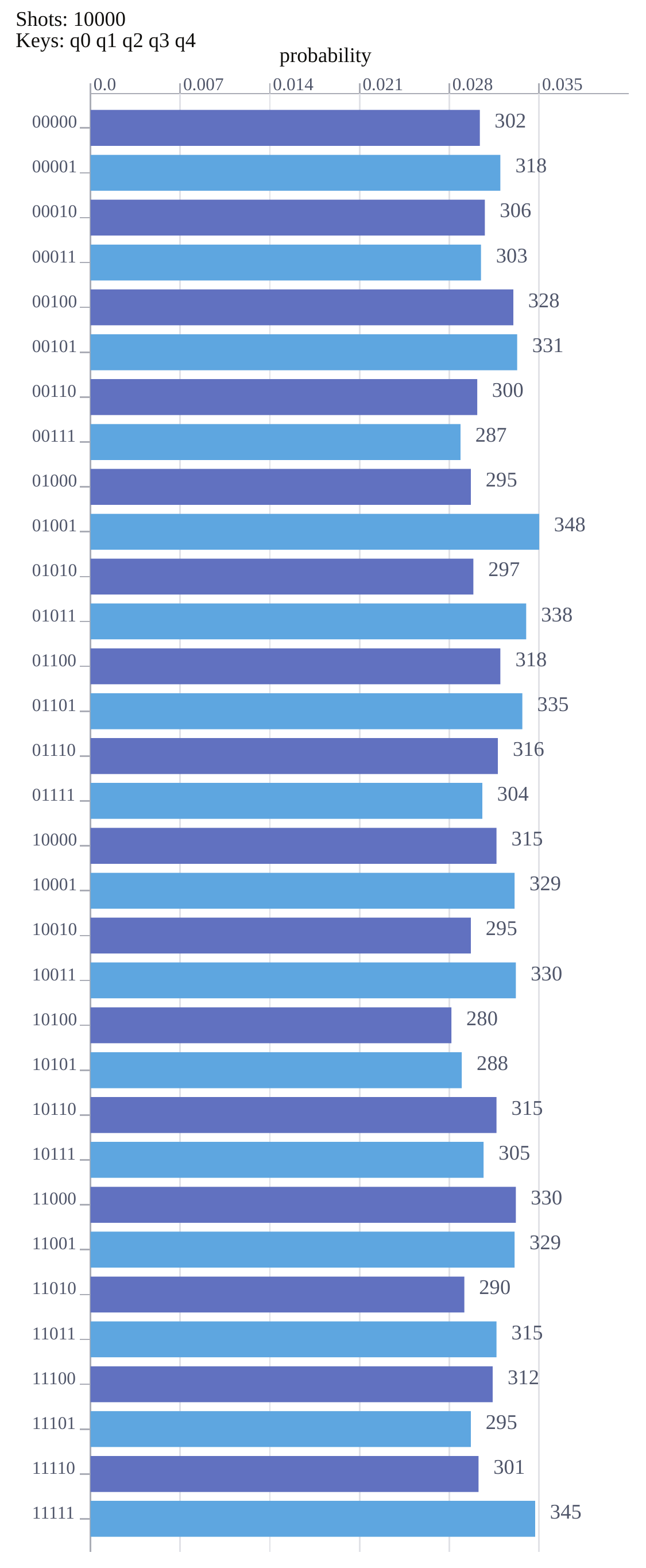}
\caption{\label{measure01001grovernoise009} The sampling results of the 5-qubit Grover's algorithm in the depolarizing channel. The noise parameter $p=0.09$.}
\end{minipage}
\begin{minipage}{0.32\textwidth}
\centering
\includegraphics[width=1.45in]{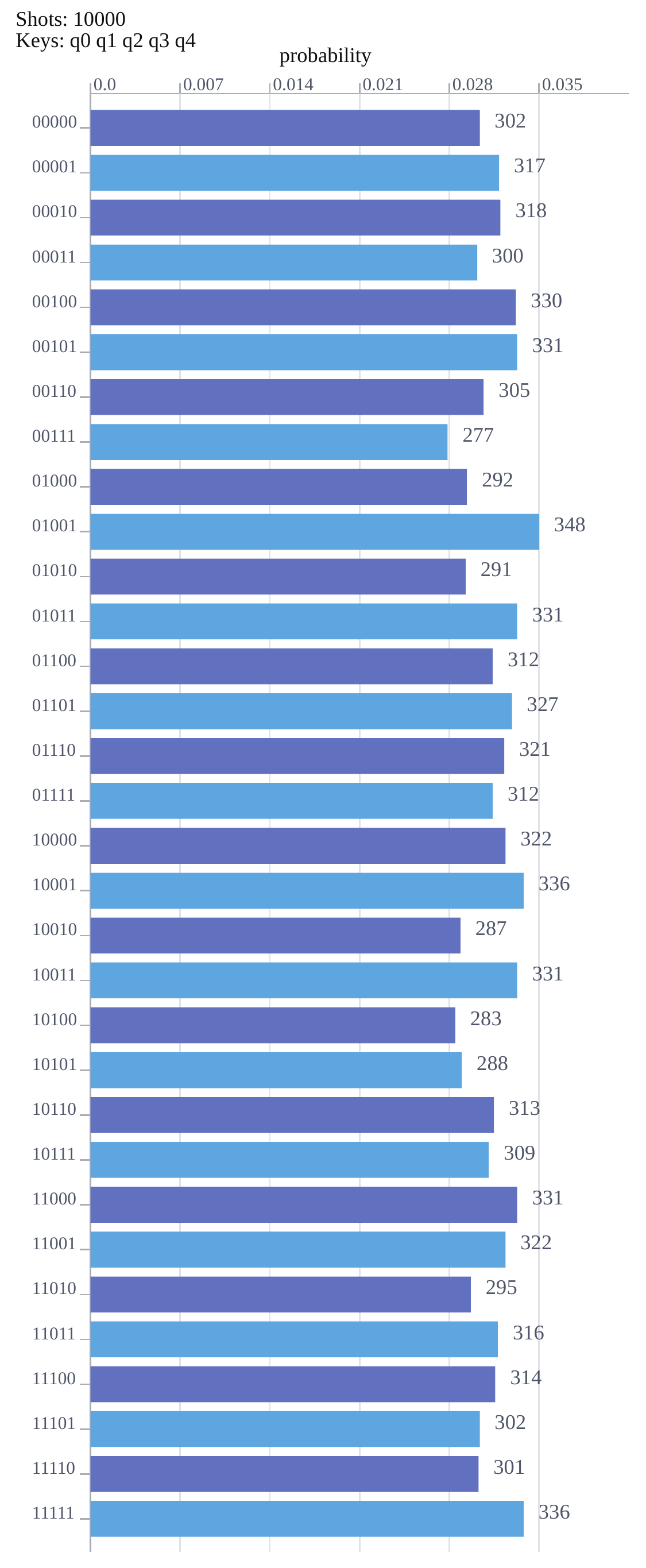}
\caption{\label{measure01001longnoise009} The sampling results of the 5-qubit algorithm by Long in the depolarizing channel. The noise parameter $p=0.09$.}
\end{minipage}
\begin{minipage}{0.32\textwidth}
\centering
\includegraphics[width=1.45in]{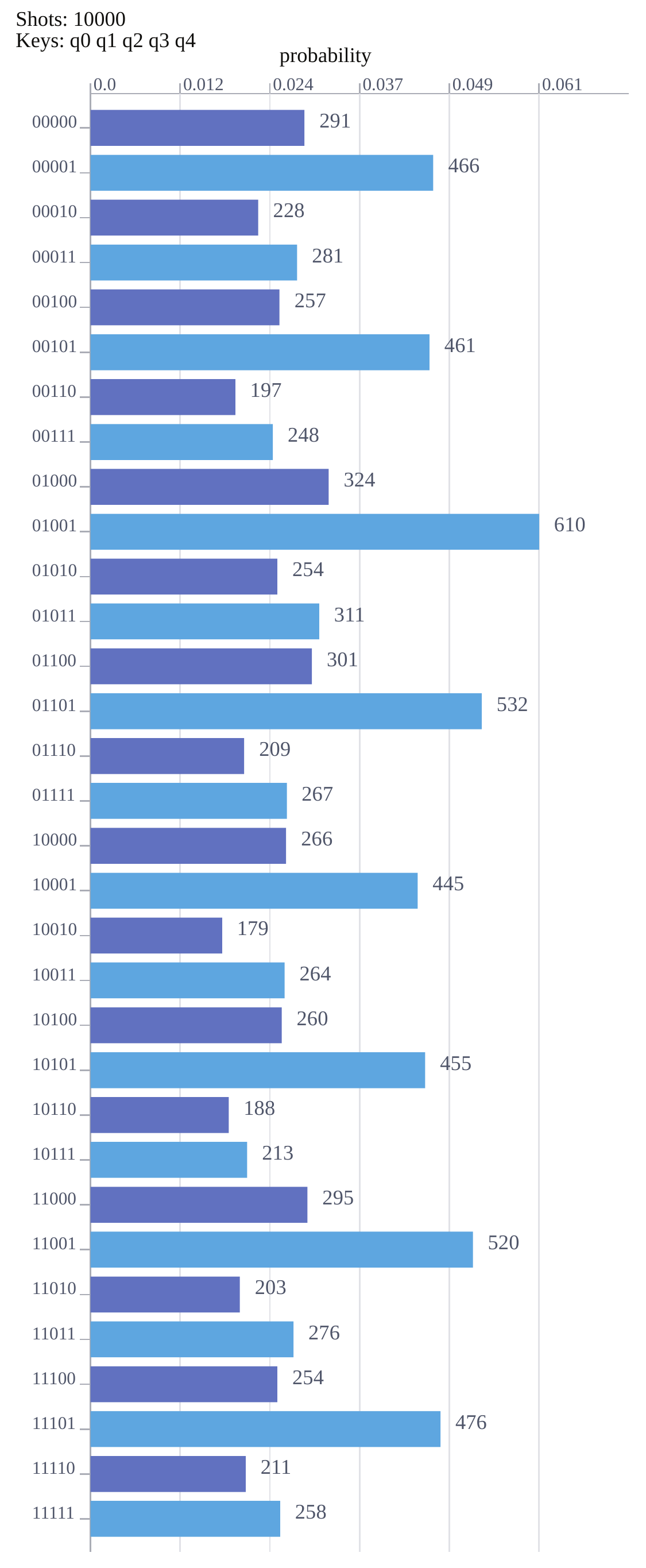}
\caption{\label{measure01001deganoise009} The sampling results of the 5-qubit DEGA in the depolarizing channel. The noise parameter $p=0.09$.}
\end{minipage}
\end{figure}

\section{Conclusion}\label{conclusion}
Distributed quantum computation has gained extensive attention since small-qubit quantum computers seem to be built more practically  in the noisy intermediate-scale quantum (NISQ) era. Hence, researchers are considering how several small-scale devices might collaborate to accomplish a task on a large one.


In this paper, we have given a distributed Bernstein-Vazirani algorithm (DBVA) with $t$ computing nodes, and a distributed exact Grover's algorithm (DEGA) that solve the search problem with only one target item in the unordered databases. In comparison to the BV algorithm, the circuit depth of DBVA is not greater than $2^{\text{max}(n_0, n_1, \cdots, n_{t-1})}+3$ instead of $2^{n}+3$. In addition, the circuit depth of DEGA is $8(n~\text{mod}~2)+9$, which is less than the circuit depth of Grover's algorithm,
$1 + 8\left\lfloor \frac{\pi}{4}\sqrt{2^n} \right\rfloor$.



Finally, we have provided particular situations of our DBVA and DEGA on MindQuantum to validate the correctness and practicability of our methods. First of all, we have shown how to decompose the original 6-qubit BV problem into two 3-qubit and three 2-qubit tasks. In the next part, we have explicated the specific steps of implementing $n$-qubit DEGA, where $n\in\{2,3,4,5\}$. In the end, we have simulated 6-qubit DBVA and 5-qubit DEGA running in the depolarization channel. 
By simulating the algorithms running in the depolarized channel, it further illustrates that distributed quantum algorithms have the superiority of resisting noise.

However, we have only designed a distributed exact quantum algorithm for the case of unique target in search problem, so it is still open that designing a distributed exact quantum algorithm solves the search problem with the case of multiple targets.



\section*{Acknowledgements}
This work was supported in part by the National Natural Science Foundation of China (Nos. 61572532, 61876195), and the Natural Science Foundation of Guangdong Province of China (No. 2017B030311011).

\balance


\nocite{*}

\end{document}